\documentclass[a4paper,12pt,reqno]{amsart}
\pdfoutput=1
\usepackage{a4}
\usepackage{amsthm}
\usepackage{amssymb}
\usepackage{hyperref}
\usepackage[T1]{fontenc}

\newtheorem{definition}{Definition}
\newtheorem{example}[definition]{Example}
\newtheorem{theorem}[definition]{Theorem}
\newtheorem{lemma}[definition]{Lemma}
\newtheorem{remark}[definition]{Remark}

\newtheorem{proposition}[definition]{Proposition}

\makeatletter
\@addtoreset{definition}{section}
\@addtoreset{equation}{section}
\makeatother
\DeclareMathOperator{\Res}{Res}
\DeclareMathOperator{\mathAE}{\textup{\AE}}

\allowdisplaybreaks[4]

\begin{document}
	
\title[BTR from extended loop equations]{Blobbed topological recursion
  from extended loop equations}

\author{Alexander Hock \and Raimar Wulkenhaar}
	
\address{Institut für Mathematik, Ruprecht-Karls-Universit\"at
  Heidelberg, \newline Im Neuenheimer Feld 205, 69120, Heidelberg,
  Germany \newline {\itshape e-mail:} \normalfont  
	\texttt{ahock@mathi.uni-heidelberg.de}}
	
\address{Mathematisches Institut, 
Universit\"at M\"unster, \newline
Einsteinstr.\ 62, 48149 M\"unster, Germany \newline
{\itshape e-mail:} \normalfont
\texttt{raimar@math.uni-muenster.de}}

\begin{abstract}
We consider the $N\times N$ Hermitian matrix model with measure
$d\mu_{E,\lambda}(M)=\frac{1}{Z} \exp(-\frac{\lambda N}{4}
\mathrm{tr}(M^4)) d\mu_{E,0}(M)$, where $d\mu_{E,0}$ is the
Gau\ss{}ian measure with covariance $\langle M_{kl}M_{mn}\rangle
=\frac{\delta_{kn}\delta_{lm}}{N(E_k+E_l)}$ for given $E_1,...,E_N>0$.
It was previously understood that this setting gives rise to
two ramified coverings $x,y$
of the Riemann sphere strongly tied by $y(z)=-x(-z)$ and a family
$\omega^{(g)}_{n}$ of meromorphic differentials conjectured to
obey blobbed topological recursion due to Borot and Shadrin. 
We develop a new approach to this problem via 
a system of six meromorphic functions which satisfy extended loop equations.
Two of these functions are symmetric in the
preimages of $x$ and can be determined from their consistency
relations. An expansion at $\infty$ gives global linear and
quadratic loop equations for the $\omega^{(g)}_{n}$. These global
equations provide the $\omega^{(g)}_{n}$ not only in the vicinity of the
ramification points of $x$ but also in the vicinity of all other poles
located at opposite diagonals $z_i+z_j=0$ and at $z_i=0$. We deduce a
recursion kernel representation valid at least for $g\leq 1$.
\end{abstract}


	\maketitle

\section{Introduction}
\subsection{Historical comments}

This paper is the fifth in a series
\cite{Grosse:2019jnv,Schurmann:2019mzu,Branahl:2020yru,Hock:2021tbl}
which investigates and solves the \textit{quartic Kontsevich
  model}. See \cite{Branahl:2021slr} for a review.
Back in 1991, Kontsevich \cite{Kontsevich:1992ti}
constructed his classical matrix model to prove Witten's conjecture
\cite{Witten:1990hr} that the generating series of intersection
numbers on the moduli space $\overline{\mathcal{M}}_{g,n}$ of stable
complex curves is a tau function of the integrable KdV hierarchy.  It
is formulated as an $N\times N$-Hermitian matrix model with covariance
$\langle M_{kl}
M_{mn}\rangle=\frac{\delta_{kn}\delta_{lm}}{N(E_k+E_l)}$ for $E_k>0$,
deformed by a cubic potential
$\frac{\mathrm{i}N}{6}\mathrm{Tr}(M^3)$. Explicit results were
computed for the correlation function in the classical Kontsevich
model for instance in
\cite{Makeenko:1991ec,Eynard:2007kz,Eynard:2016yaa} and more recently
for higher spectral dimension with smooth covariance renormalised by
quantum field theoretical techniques in
\cite{Grosse:2016pob,Grosse:2016qmk,Grosse:2019nes}.  The quartic
Kontsevich model has the same covariance
$\langle M_{kl}
M_{mn}\rangle=\frac{\delta_{kn}\delta_{lm}}{N(E_k+E_l)}$ for $E_k>0$,
but deformed by a quartic potential
$\frac{\lambda N}{4}\mathrm{Tr}(M^4)$. It is different from the
generalised Kontsevich model \cite{Belliard:2021jtj} (see for more
details \cite[\S \, 2.1]{Branahl:2020uxs}).

The Hermitian 1-matrix models with trivial covariance but arbitrary
polynomial deformation were properly understood in
\cite{BertrandEynard2004}.  The solution of the Hermitian 2-matrix
model \cite{Chekhov:2006vd} provided an unexpectedly simple
formula. It was the first time that the formula for topological
recursion was written down, explicitly incorporating the local deck
transformation (Galois involution) around a ramification point
$\beta_i$. Eynard and Orantin took this formula in
\cite{Eynard:2007kz} as a universal definition of the theory of
\emph{topological recursion} (TR) for general initial data
$(\Sigma,x,y,B)$, the so-called \textit{spectral curve}, and proved
several properties.

We assume in this article to
have $\Sigma \backsimeq \hat{\mathbb{C}}$ and two coverings
$x,y:\Sigma\to \Sigma$ with simple distinct ramification points. These
two coverings build in TR a meromorphic 1-form
$\omega^{(0)TR}_{1}=y\,dx$ on $\Sigma$; $\omega^{(0)TR}_{2}=B$ is the
symmetric meromorphic bilinear differential on $\Sigma^2$ with double
pole on the diagonal and no residue. From the initial data
$(\Sigma,x,y,B)$, Eynard and Orantin defined recursively in the
negative Euler characteristic $-\chi=2g+n-2$ a family of symmetric
meromorphic differentials $\omega^{(g)TR}_{n}$ on $\Sigma^n$ with poles just
at the ramification points of $x$ for $-\chi>0$.


The proof that the 2-matrix model is governed by the formula of TR needed
to develop a completely
new technique because the number of ramification points
can be arbitrarily large, depending on the deformation (the potential)
of the model. The starting point was to look at two different classes
of mixed correlators, where one of them is rational in $x(z)$ rather
than $z$ \cite{Chekhov:2006vd}. We denote these functions by
$H^{(g),TR}_{n+1}(y(w);z;I)$ and
$P^{(g),TR}_{n+1}(y(w);x(z);I)$\footnote{In \cite{Chekhov:2006vd}, we
  identify $U^{(g)}\to \prod_{u_i\in I} \partial_{x(u_i)} H^{(g),TR}$
  and $E^{(g)}\to \prod_{u_i\in I} \partial_{x(u_i)}
  P^{(g),TR}$}, where $I=\{u_1,...,u_n\}$.
Here, $H^{(g),TR}$ is rational in $\{y(w),z\}\cup I$
and $P^{(g),TR}$ in $\{y(w),x(z)\} \cup I$ as indicated by the
arguments.

The two functions $H^{(g),TR}$ and $P^{(g),TR}$ together with
$W^{(g),TR}_n$ defined by
$d_{u_1}...d_{u_n}W^{(g),TR}_{n+1}(z;u_1,...,u_n)dx(z)
:=\omega^{(g)TR}_{n+1}(z,u_1,...,u_n)$
satisfy a Dyson-Schwinger equation (DSE)
\begin{align}\label{DSEHP}
  &(y(w)-y(z))H^{(g),TR}_{n+1}(y(w);z;I)+P^{(g),TR}_{n+1}(y(w);x(z);I)
 \\\nonumber
 &=- \hspace*{-5mm}
 \sum_{\substack{g_1+g_2=g\\ I_1\uplus I_2=I\\ (g_2,I_2)\neq (0,\emptyset)}}
 \hspace*{-5mm} H^{(g_1),TR}_{|I_1|+1}(y(w);z;I_1)W^{(g_2),TR}_{|I_2|+1}(z;I_2)
 -\partial_{x(z')} H^{(g-1),TR}_{n+2}(y(w);z;z',I)\big|_{z'=z}\,.
\end{align}
Exactly the same DSE for some $H^{(g),TR}$ and $P^{(g),TR}$ (but with
different $x,y$) appears in several other examples which are governed
by TR (see for instance \cite{Belliard:2021jtj,Branahl:2022uge}).

The important observation is that the solution of $P^{(g),TR}$ has a
representation in terms of symmetric functions of $W^{(g'),TR}_{n'}$,
which are symmetric in \emph{all} preimages of $x$ in the variable
$z$, i.e.\ symmetric in $(\hat{z}^k)_{k=0,1,...,d}$ with
$x(z)=x(\hat{z}^k)$ and $\hat{z}^0=z$. Furthermore, the function
$H^{(g),TR}$ was chosen such that its asymptotic expansion in $y(w)$
has as leading order $W^{(g),TR}$, i.e.\
\begin{align*}
  H^{(g),TR}_{n+1}(y(w);z;I)
  =\frac{1}{y(w)}W^{(g),TR}_{n+1}(z;I)+\mathcal{O}(y(w)^{-2}) \,,
\end{align*}
plus additional terms for $(g,n)\in\{(0,0),(0,1)\}$.  From this one
can show that the asymptotic expansion of the solution of $P^{(g),TR}$
recovers the so-called \textit{linear and quadratic loop equations}
\cite{Borot:2013lpa}), which are describing the local behaviour of
$W^{(g),TR}$ around the ramification points. Finally, from the linear
and quadratic loop equations, the well-known recursion formula for the
$W^{(g),TR}_{n+1}(z;I)$ or $\omega^{(g),TR}_{n+1}(z,I)$ can be
deduced. Vice versa, for any spectral curve in TR, r\^{o}le and
formulae for $H^{(g),TR}, P^{(g),TR}$ are completely
determined. In fact, to prove that an example is governed by TR,
the starting point is, in general, the equation \eqref{DSEHP}.

A family $W^{(g),TR}_{n+1}(z;I)$ satisfies topological recursion if
and only if it satisfies the linear and quadratic loop equations
\textit{and} all of its poles are for $2g+n-1>0$ at the ramifications
points of $x$  (assuming a compact curve).  It is very natural 
to assume the linear and quadratic loop
equations but to relax the assumption on the pole
structure. This gave birth to a
generalisation of TR by \textit{blobbed topological recursion} (BTR)
\cite{Borot:2015hna}.  The motivation of formulating BTR came from the
enumeration of stuffed maps \cite{Borot:2013fla}, which is a natural
generalisation of the Hermitian 1-matrix model through a deformation
(the potential) of higher topological structure. BTR is built not just
by the initial data $(\Sigma,x,y,B)$, but enriched by additional blobs
$\phi_{g',n'}$ associated with a topology contributing if
$2g+n-2\leq 2g'+n'-2$. The meromorphic functions
$W^{(g)}_{n+1}(z;I)=\mathcal{P}_z W^{(g)}_{n+1}(z;I)+\mathcal{H}_z
W^{(g)}_{n+1}(z;I)$ decompose in BTR into a part
$\mathcal{P}_z W^{(g)}_{n+1}$ with poles at the ramification points of
$x$ and a part $\mathcal{H}_z W^{(g)}_{n+1}$ with poles elsewhere.
The decomposition is constructed by orthogonal projectors
$\mathcal{P}_z,\mathcal{H}_z$.  Due to the linear and quadratic loop
equations, it is rather easy to show that the polar part
$\mathcal{P}_z W^{(g)}_{n+1}(z;I)$ is still determined by the TR
recursion formula.  However, the part
$\mathcal{H}_z W^{(g)}_{n+1}(z;I)$ is highly dependent on the model under
consideration and therefore on the initial enriched data $\phi_{g,n}$.

\subsection{Main result}

Consider the $N\times N$ Hermitian matrix model with measure
$d\mu_{E,\lambda}(M)=\frac{1}{Z} \exp(-\frac{\lambda N}{4}
\mathrm{tr}(M^4)) d\mu_{E,0}(M)$ where $d\mu_{E,0}$ is the Gau\ss{}ian
measure with covariance
$\langle M_{kl}M_{mn}\rangle
=\frac{\delta_{kn}\delta_{lm}}{N(E_k+E_l)}$ for given $E_1,...,E_N>0$.
Let $e_1,...,e_d$ be the pairwise distinct values in $(E_1,...,E_N)$
and $r_1,...,r_d$ their multiplicities\footnote{For different choices
  of $ d$, which is the number of pairwise distinct eigenvalues
  $ e_i $ of $ E$, we can interpolate between different applications
  that we have in mind. For $ d = N $, the model possesses the same
  covariance with the same multiplicity as the classical Kontsevich
  model but with a quartic deformation instead of a cubic one. For
  $ d = 1 $, the classical quartic Hermitian 1-matrix model is
  recovered. For other values of $ d $ and $r_i$, the model finds
  applications in quartic scalar quantum field theories on
  noncommutative space-time, known as the Grosse-Wulkenhaar model.  }.
In previous work \cite{Schurmann:2019mzu,Branahl:2020yru} we showed
that this setting (called the quartic Kontsevich model) gives rise to
a generic ramified covering $x:\hat{\mathbb{C}}\to \hat{\mathbb{C}}$
of degree $d+1$ with simple poles (one located at $\infty$) and simple
ramification points. The key observation was that the other ramified
covering $y$ of the spectral curve is simply given by
\begin{align}\label{glInvolution}
	y(z)=-x(-z).
\end{align}
This strong $x,y$-entanglement has exceptional consequences. We show
in this paper that the Dyson-Schwinger equations established in
\cite{Branahl:2020yru} give rise to a system of 
seven functions where six of them come in two pairs of triples
\begin{align*}
&(P^{(g)}_{n+1}(x(w),x(z);I),\;
  H^{(g)}_{n+1}(x(w);z;I),\;
  U^{(g)}_{n+1}(w,z;I)), 
  \\
  &(Q^{(g)}_{n+1}(x(w),x(z);I),\;
  M^{(g)}_{n+1}(x(w);z;I),\;
  V^{(g)}_{n+1}(w,z;I))\quad \text{and} \\
  &W^{(g)}_{n+1}(z;I).
\end{align*}
The goal is to derive a recursive procedure to compute
$W^{(g)}_{n+1}(z;I)$ in terms of $W^{(g')}_{n'+1}(z;I')$ with
$2g'+n'-2<2g+n-2$ from the intricate system of equations of all seven
functions. Compared with the two functions
$(P^{(g),TR}_{n+1}(y(w),x(z);I), H^{(g),TR}(y(w);z;I))$ in TR, the
r\^{o}le of $x,y$ is partly flipped.  Partial fraction decompositions
given in Definition~\ref{def:partialfractions} and the equations
themselves given in Proposition \ref{Prop:DSE} are interwoven. By
extending the known techniques from the literature (applied for
instance in \cite{Belliard:2021jtj,Branahl:2022uge}), we show how our
system of equations can be solved and how the Taylor expansion about
$\frac{1}{x(w)}=0$ gives rise to \emph{globally defined} linear and
quadratic loop equations for $W^{(g)}_{n+1}(z;I)$ in the variable
$z$. The loop equations indicate poles at the ramification points
$z=\beta_i$ of $x$, at the opposite diagonal $z+u_i=0$ and (for
$g\geq 1$) at $z=0$. From the detailed structure a formula to compute
$\omega^{(g)}_n$ recursively in decreasing Euler characteristic is
deduced:
\begin{theorem}
\label{thm:mainrecursion}
Let $I=\{u_1,...,u_n\}$ and $x$ be a generic ramified cover of
$\hat{\mathbb{C}}$. Let $x,y$ be related via \eqref{glInvolution} and
$\omega^{(0)}_2(z,w)=\frac{dz dw}{(z-w)^2}+\frac{dz dw}{(z+w)^2}$.  If
the seven functions $W^{(g)}_{n+1}$, $(P^{(g)},H^{(g)},U^{(g)})$ and
$(Q^{(g)},M^{(g)},V^{(g)})$ satisfy the interwoven DSEs of Proposition
\ref{Prop:DSE} and the partial fraction decomposition of
Definition~\ref{def:partialfractions}, then the solution for
$\omega^{(g)}_{n+1}(z,u_1,...u_n)=\lambda^{2g+n-1} d_{u_1}\cdots
d_{u_n} W^{(g)}_{n+1}(z;u_1,...,u_n) dx(z)$, where
$H^{(g)}_{n+1}(x(v);z;I)=- \frac{\lambda
  W^{(g)}_{n+1}(z;I)}{x(v)}+\mathcal{O}(x(v)^{-2})$, is recursively
computed for all $I$ with $2g+|I|\geq 2$ and for $g<2$ via the
recursion kernel representation
\begin{align}
\omega^{(g)}_{|I|+1}(z,I)
&= \sum_{\beta_i}\Res\displaylimits_{q\to \beta_i}
K_i(z;q)\Big\{\hspace*{-3mm}
\sum_{\substack{ I_1\uplus I_2=I\\ g_1+g_2=g\\
(g_i,I_i)\neq (0,\emptyset)   }} \hspace*{-3mm}
\omega^{(g_1)}_{|I_1|+1}(q,I_1)
\omega^{(g_2)}_{|I_2|+1}(q,I_2)
+\omega^{(g-1)}_{|I|+2}(q,q,I)
	 	\Big\}
\nonumber
\\
& +\sum_{j=1}^{|I|} d_{u_j}\Big[\Res\displaylimits_{q\to - u_j}
K_{u_j}(z;q)
\Big\{\hspace*{-3mm}\sum_{\substack{ I_1\uplus I_2=I\\ g_1+g_2=g\\
(g_i,I_i)\neq (0,\emptyset)   }}\hspace*{-3mm}
d_{u_j}^{-1}(\omega^{(g_1)}_{|I_1|+1}(q,I_1)
	 \omega^{(g_2)}_{|I_2|+1}(q,I_2))
\nonumber\\
&
\qquad\qquad
+ d_{u_j}^{-1}\omega^{(g-1)}_{|I|+2}(q,q,I)
+\frac{(dx(q))^2}{6} \frac{\partial^2}{\partial (x(q))^2}
\Big(\frac{\omega^{(g-1)}_{|I|+1}(q,I))}{dx(q) dx(u_j)}\Big)\Big\}\Big]
\nonumber
\\
&+\Res\displaylimits_{q\to 0}
K_0(z;q)	 	\Big\{\hspace*{-3mm}
\sum_{\substack{ I_1\uplus I_2=I\\ g_1+g_2=g\\
(g_i,I_i)\neq (0,\emptyset)   }}\hspace*{-3mm}
\omega^{(g_1)}_{|I_1|+1}(q,I_1)
\omega^{(g_2)}_{|I_2|+1}(q,I_2)
+\omega^{(g-1)}_{|I|+2}(q,q,I)
	 	\nonumber\\
&\qquad\qquad +\frac{(d x(q))^2}{2}
\frac{\partial}{\partial x(q)}
\Big( \frac{d_{q'}^{-1} 
\omega^{(g-1)}_{|I|+2}(q,q',I)}{dx(q)}\Big|_{q'=q}
\Big)\Big\}\;,
\label{eq:mainformula}
\end{align}
where $\beta_i$ are the ramification points of $x$,
$\omega^{(0)}_{2}(q,q)$ should be replaced by
$\lim_{q'\to q}(\omega^{(0)}_{2}(q,q')-\frac{1}{(x(q)-x(q'))^2})$ and
the recursion kernels are given by
$K_i(z;q)=-\frac{(\frac{dz}{z-q}-\frac{dz}{z-\sigma_i(q)})
}{2(y(q)-y(\sigma_i(q))) dx(q)}$,
$K_{u}(z;q)= - \frac{(\frac{dz}{z-q}-\frac{dz}{z+ u})}{2
  (y(q)+x(u))dx(q)}$ and
$K_0(z;q)=-\frac{ (\frac{dz}{z-q}-\frac{dz}{z}) }{2(y(q)+x(
  q))dx(q)}$.
\end{theorem}
We do not see any obstacle to extend the result to all $g$;
only the combinatorics becomes extremely involved and Taylor expansions
of $P^{(g)}_{n+1}(x(w),x(z) ;I)$ up to an order which increases as $4g$ become
necessary. Already in the proof up to genus $g\leq 1$ the employment of a
\textit{loop insertion operator} as an abbreviation for
particular rational functions of 
$(P^{(g)},H^{(g)},U^{(g)})$ and
  $(Q^{(g)},M^{(g)},V^{(g)})$ was essential to master the combinatorics.

  A few remarks:\vspace*{-.5ex} 
  
\begin{itemize}  
\item We establish the linear and quadratic loop equations
  \emph{globally} on $\hat{\mathbb{C}}$ and not only in small
  neighbourhoods of the ramification points as required in the general
  formulation \cite{Borot:2015hna}.  Therefore, we not only get a
  recursion formula for the `polar' part
  $\mathcal{P}_z\omega^{(g)}_{n+1}(z,I)$ but for the entire
  $\omega^{(g)}_{n+1}(z,I)$.

\item In \cite{Hock:2021tbl} we solved the genus zero sector under the
  much weaker assumption that $x,y$ are related as $y(z)=-x(\iota z)$
  for some holomorphic involution $\iota$ on $\hat{\mathbb{C}}$.  We
  noticed that although one can solve
  $(Q^{(0)}_1,M^{(0)}_1,V^{(0)}_1)$ also for $y(z)=-x(\iota z)$
  along the lines of
  \cite{Schurmann:2019mzu}, the resulting expressions are of
  completely different type to which our tools do not apply.
  Nevertheless we would
  like to remark that the involution identity discovered in
  \cite{Hock:2021tbl} was vastly extended in
  \cite{Hock:2022wer,Hock:2022pbw,Alexandrov:2022ydc} to a general
  approach to the $x$-$y$ symmetry in topological recursion.
  
\item In the theory of TR, one can consider a more general version of
  \eqref{DSEHP} that includes intermediate correlators
  \( W^{(g),TR}_{n,m} \) (see, for example, the 2-matrix model in
  \cite{Chekhov:2006vd}). From the perspective of \( x \)-\( y \)
  duality \cite{Alexandrov:2022ydc}, new results have been obtained
  \cite{Hock:2022wer,Alexandrov:2023oov}, which show that these
  intermediate correlators satisfy the linear and quadratic loop
  equations and are thus governed by BTR. It would be interesting to
  compare the recursive structures, possibly also in the context of
  Generalised Topological Recursion \cite{Alexandrov:2024tjo}.

\item To have a meaningful extension of TR, the multidifferentials
  generated by eq.\ \eqref{eq:mainformula} should be symmetric in
  their arguments. Symmetry is a highly non-trivial property due to
  the special role of the variable $ z $ in comparison to
  $ u_i \in I $. We will not prove symmetry here. However, in a
  companion article \cite{hock2025}, we prove together with Shadrin
  the symmetry of all genus zero multidifferentials
  $ \omega^{(0)}_n $. In \cite{hock2025}, the starting point is an
  involution identity satisfied by $\omega^{(0)}_n $, which in
  \cite{Hock:2021tbl} was derived from the Dyson-Schwinger equations
  of the quartic Kontsevich model.
  
\item The assumption of a genus zero spectral curve greatly simplifies
  explicit computations but is not mandatory. However, since our
  result builds upon previous work \cite{Grosse:2019jnv,
    Schurmann:2019mzu} where this assumption was made, we must retain
  it. It might be possible to relax the assumption on the genus of the
  spectral curve, but then the results \textit{op.\ cit.} need to be
  rederived. Having a higher genus spectral curve can have an impact
  on the recursive formula for $\omega^{(g)}_n$ since these are
  assumed to be rational in this article by the rationality of $x,y$.
  
\item Important questions concerning integrability for TR were
  recently settled in \cite{Alexandrov:2024hgu,Alexandrov:2024qfe}. In
  short, any genus zero spectral curve gives rise to a KP
  \(\tau\)-function, and for higher genus, the so-called
  non-perturbative TR ensures KP-integrability. On the other hand,
  polynomial deformations of Kontsevich matrix models with \emph{any}
  potential were proven to provide a BKP $\tau$-function
  \cite{Borot:2023thu}, with the Pfaffian
  $\mathrm{Pf}(\frac{e_k-e_l}{e_k+e_l})$ of the eigenvalues $e_k$ of
  the external matrix $E$ playing a decisive role. It remains to be
  seen whether the symmetry between $e_k-e_l$ and $e_k+e_l$, which is
  also shared by $\omega^{(0)}_2$ of the quartic Kontsevich model (see
  \cite{Branahl:2020yru} and Thm.\ \ref{thm:mainrecursion}), is more
  than a coincidence.

\end{itemize}

The paper is organised as follows. In sec.~\ref{sec:setup} we derive
from previous results a system of equations for
$(P^{(g)}_n,H^{(g)}_n,U^{(g)}_n)$ and
$(Q^{(g)}_n,M^{(g)}_n,V^{(g)}_n)$ and introduce the loop insertion
operator.  Sec.~\ref{sec:g0} solves $P^{(0)}_{|I|+1}(x(v),x(z);I)$ by
exploiting its symmetry in the preimages of $x$ and its residues at
$x(v)=x(u_i)$.  The outcome gives the linear and quadratic loop
equations for $W^{(0)}_n$.  By essentially the same methods (but
including poles at $x(v)=x(0)$) we identify
$Q^{(0)}_{|I|+1}(x(v),x(z);I)$ in sec.\ \ref{sec:g0-Q} and, with much
larger effort, in particular in view of poles at $x(v)=x(z)$,
the functions $P^{(1)}_{|I|+1}(x(v),x(z);I)$ in secs.~\ref{sec:g1-0}
and \ref{sec:g1-I}.  Again, this allows us to derive the global linear
and quadratic loop equations for $W^{(1)}_n$ from which we get in
sec.~\ref{sec:recursion} the recursion formula stated in
Theorem~\ref{thm:mainrecursion}.

\subsection*{Acknowledgements}

We are grateful to anonymous referees for their valuable comments and
suggestions, which have increased the readability of the article. The
major part of the work of AH on this paper was done during his
research stay at the University of Oxford.  He would like to thank the
University of Oxford for providing a great work environment. RW thanks
the the University of Oxford for hospitality during a research
visit. AH was supported through a Walter \mbox{Benjamin}
fellowship\footnote{``Funded by
the Deutsche Forschungsgemeinschaft (DFG, German Research
Foundation) -- Project-ID 465029630.''}.
RW was supported\footnote{``Funded by
  the Deutsche Forschungsgemeinschaft (DFG, German Research
  Foundation) -- Project-ID 427320536 -- SFB 1442, as well as under
  Germany's Excellence Strategy EXC 2044 390685587, Mathematics
  M\"unster: Dynamics -- Geometry -- Structure.''} by the Cluster of
Excellence \emph{Mathematics M\"unster} and the CRC 1442 \emph{Geometry:
  Deformations and Rigidity}.

\section{Setup}
\label{sec:setup}

\subsection{Summary of previous results}

In \cite{Branahl:2020yru} we have shown that a quartic analogue of the
Kontsevich matrix model gives rise to, and is completely determined
by, a coupled system of Dyson-Schwinger equations for three families
of meromorphic functions $\Omega^{(g)}_n(z_1,...,z_n)$,
$\mathcal{T}^{(g)}_{n+1}(u_1,...,u_n\|z,w|)$ and
$\mathcal{T}^{(g)}_{n+1}(u_1,...,u_n\|z|w|)$ on the Riemann sphere
$\hat{\mathbb{C}}=\mathbb{C}\cup\{\infty\}$. Of particular interest
are the $\Omega^{(g)}_n$ which give rise to meromorphic differentials
\begin{subequations}
  \begin{align}
  \omega^{(g)}_{n}(z_1,...,z_n)&=  \Omega^{(g)}_n(z_1,...,z_n)
  \prod_{j=1}^n dx(z_j)\;,\qquad \text{where}
\\
x(z)&= z -\frac{\lambda}{N} \sum_{k=1}^d
\frac{\varrho_k}{z+\varepsilon_k}\;.
\label{R}
\end{align}
\end{subequations}
The ramified cover\footnote{The ramified cover $x$ was called $R$ in
  \cite{Schurmann:2019mzu,Branahl:2020yru}.}
$x: \hat{\mathbb{C}}\to \hat{\mathbb{C}}$ forms with its reflection
$y(z)=-x(-z)$ and
$\omega^{(0)}_{2}(z_1,z_2)=\frac{dz_1\,dz_2}{(z_1-z_2)^2}+
\frac{dz_1\,dz_2}{(z_1+z_2)^2}$ a spectral curve\footnote{In
  \cite{Branahl:2020yru}, building on \cite{Grosse:2019jnv,
    Schurmann:2019mzu}, we assumed a genus zero solution for the
  spectral curve, implying that $x,y$ are rational functions and,
  consequently, that all $\Omega^{(g)}_n$ are rational. This
  assumption can clearly be relaxed, which should be investigated in
  future work.} in the spirit of \emph{topological recursion}
\cite{Eynard:2007kz}. The parameters
$(\lambda,N,d,\{\varrho_k,\varepsilon_k\})$ are defined by the initial
data of the quartic Kontsevich model: It is a matrix model for
$N\times N$-Hermitian matrices $M$ with covariance
$\langle M_{kl}
M_{mn}\rangle=\frac{\delta_{kn}\delta_{lm}}{N(E_k+E_l)}$ for $E_k>0$,
deformed by a quartic potential $\frac{\lambda N}{4}\mathrm{Tr}(M^4)$.
If $(e_1,...,e_d)$ are the pairwise different values in $(E_k)$, which
arise with multiplicities $r_1,...,r_d$, then
$(\varepsilon_k,\varrho_k)$ are determined by $x(\varepsilon_k)=e_k$
and $\varrho_k x'(\varepsilon_k)=r_k$ with
$\lim_{\lambda \to 0} \varepsilon_k =e_k$ and
$\lim_{\lambda \to 0} \varrho_k =r_k$
\cite{Grosse:2019jnv,Schurmann:2019mzu}. The following expectation
values of the quartic Kontsevich model in a formal
$\frac{1}{N}$-expansion are recovered from the meromorphic functions
$\Omega^{(g)}_n$ by the specialisation to the points
$z=\varepsilon_i$, i.e.
\begin{align*}
    &\sum_{g=0}^\infty N^{2-2g-n}\Omega^{(g)}_n(\varepsilon_1,...,\varepsilon_n)-\frac{\delta_{2,n}\delta_{g,0}}{(E_1-E_2)^2}\\
    =&\frac{(-N)^n\partial^n}{\partial E_1...\partial E_n}\log \int_{M\in H_N} \exp\bigg(-\frac{\lambda N}{4} \mathrm{tr}(M^4)\bigg)d\mu_{E,0}(M).
\end{align*}
For the derivation in \cite{Branahl:2020yru}, it was assumed that
$\Omega^{(g)}_n(z_1,...,z_n)$ is independent of $N$, whereas
$(x(z),\varepsilon_i,\varrho_i)$ depend explicitly on $N$ with
$\varrho_i$ of order $N$ and $\frac{r_i}{N}\in \mathbb{Q}$.

The system of Dyson-Schwinger equations (DSE) for $\Omega^{(g)}_n$ and the
two variants of $\mathcal{T}^{(g)}$ in \cite{Branahl:2020yru}
is most conveniently expressed in terms of
meromorphic functions $(W^{(g)}_n, U^{(g)}_n,V^{(g)}_n)$ where multiple
derivatives $\frac{\partial}{\partial x(u_i)}$ are taken out:
\begin{align}
  \Omega^{(g)}_{n+1}(z,u_1,...,u_n) &=:
  \frac{\partial^n W^{(g)}_{n+1}(z;u_1,...,u_n)}{\partial   
    x(u_1)\cdots \partial x(u_n)} \;,
  \label{eq:OmegaW}
\\
\mathcal{T}^{(g)}_{n+1}(u_1,...,u_n\|z,w|) &=:
  \frac{\partial^n U^{(g)}_{n+1}(z,w;u_1,...,u_n)}{\partial   
x(u_1)\cdots \partial x(u_n)} \;,
\nonumber
\\
\mathcal{T}^{(g)}_{n+1}(u_1,...,u_n\|z|w|) &=:
  \frac{\partial^n V^{(g)}_{n+1}(z,w;u_1,...,u_n)}{\partial   
x(u_1)\cdots \partial x(u_n)} \;.
\nonumber
\end{align}
In terms of $(U,V,W)$ and with $I:=\{u_1,...,u_n\}$ and $|I|=n$,
the system reads:
\\[1ex]
{\bf (a) DSE for $U^{(g)}_{n+1}$:}
\begin{align}
& (x(w)+y(z))U^{(g)}_{|I|+1}(z,w;I)
+\frac{\lambda}{N}\sum_{k=1}^d \frac{r_k
U^{(g)}_{|I|+1}(\varepsilon_k,w;I)}{x(z)-x(\varepsilon_k)}
\label{DSE-U}
\\
&=\delta_{|I|,0}\delta_{g,0}
+\lambda \hspace*{-3mm}
\sum_{\substack{ I_1\uplus I_2=I \\ g_1+g_2=g \\ (g_1,I_1) \neq (0,\emptyset)  }}
\hspace*{-3mm}
W^{(g_1)}_{|I_1|+1}(z;I_1) U^{(g_2)}_{|_2|+1}(z,w;I_2)
-\lambda 
  \sum_{j=1}^{|I|}
\frac{U^{(g)}_{|I|}(u_j,w;I\setminus u_j)}{ x(u_j)-x(z)}
\nonumber
\\
&
-\lambda \frac{\partial}{\partial x(s)}U^{(g-1)}_{|I|+2}(z,w;I\cup\{s\})
\Big|_{s=z}
- \lambda \frac{V^{(g-1)}_{|I|+1}(z,w;I)-V^{(g-1)}_{|I|+1}(w,w;I)}{
  x(w)-x(z)}  \;.
\nonumber
\end{align}
{\bf (b) DSE for $V^{(g)}_{n+1}$:}
\begin{align}
&(x(z)+y(z)) V^{(g)}_{|I|+1}(z,w;I)
+\frac{\lambda}{N}\sum_{k=1}^d r_k\frac{
V^{(g)}_{|I|+1}(\varepsilon_k,w;I)}{x(z)-x(\varepsilon_k)}
\label{DSE-V}
\\
&=-
\lambda \hspace*{-3mm}
\sum_{\substack{ I_1\uplus I_2=I \\ g_1+g_2=g \\ (g_1,I_1) \neq (0,\emptyset)  }}
\hspace*{-3mm}
W^{(g_1)}_{|I_1|+1}(z;I_1) V^{(g_2)}_{|I_2|+1}(z,w;I_2)
-\lambda \sum_{j=1}^{|I|}
\frac{V^{(g)}_{|I|}(u_j,w;I{\setminus} u_i)}{x(u_j)-x(z)}
\nonumber
\\
&
-\lambda \frac{\partial}{\partial x(s)} V^{(g-1)}_{|I|+2}(z,w;I\cup\{s\})
\Big|_{s=z}
-\lambda \frac{U^{(g)}_{|I|+1}(z,w;I)
 -U^{(g)}_{|I|+1}(w,w;I)}{x(w)-x(z)} 
\;.
\nonumber
\end{align}
{\bf (c) Connecting equation for $W^{(g)}_{n+1}$:}
\begin{align}
W^{(g)}_{|I|+1}(z;I)&=\frac{\delta_{|I|,1}\delta_{g,0}}{x(z)-x(u_1)}
+\frac{\delta_{|I|,0}\delta_{g,0}}{\lambda}
\Big(x(z)+\frac{\lambda}{N}\sum_{k=1}^d \frac{r_k}{x(\varepsilon_k)-x(z)}\Big)
\label{DSE-W}
\\*
&+\frac{1}{N} \sum_{l=1}^d r_l
U^{(g)}_{|I|+1}(z,\varepsilon_l;I)
-\sum_{j=1}^{|I|} U^{(g)}_{|I|}(z,u_j;I{\setminus}u_j)
+V^{(g-1)}_{|I|+1}(z,z;I)\;.
\nonumber
\end{align}
The connecting equation (c) has been extended to include the
consistency equation $W^{(0)}_1(z;\emptyset) :=\frac{1}{\lambda} y(z)$
for the ramified cover $x$, see \cite{Schurmann:2019mzu}. This
consistency equation turned the originally non-linear equation
\cite{Grosse:2009pa, Grosse:2012uv} for
$U^{(0)}_1(z,w)\equiv U^{(0)}_1(z,w;\emptyset) $ into the linear
equation (\ref{DSE-U}) for $I=\emptyset$; its solution is
\cite{Grosse:2019jnv,Schurmann:2019mzu}
\begin{align}
U^{(0)}_1(z,w) = \frac{1}{x(z)+y(w)}
  \frac{ \prod_{k=1}^d (x(w)+y(\hat{z}^k))}{
    \prod_{k=1}^d (x(w)-x(\varepsilon_k))}
\label{U01}
\end{align}
where $\{z=\hat{z}^0,\hat{z}^1,...,\hat{z}^d\}:=x^{-1}(x(z))$ is the
set of preimages of $x$. Straightforward manipulations show the
symmetry $U^{(0)}_1(z,w)=U^{(0)}_1(w,z)$.
The basic equation
for $V^{(0)}_1(z,w)\equiv V^{(0)}_1(z,w;\emptyset)$ has been solved in
\cite{Schurmann:2019mzu}:
\begin{align}
&V^{(0)}_1(z,w)
\label{V01}
\\
&=\frac{\lambda}{(x(w)-x(z))^2}
\Big( U^{(0)}_1(z,w)
-
\frac{(x(w)+x(z)-2x(0))\mathAE(x(z))\mathAE(x(w))}{
  (x(w)+y(w))(x(z)+y(z))}
\Big)\;,
\nonumber
\end{align}
where 
$  \mathAE(x(z))
  := \prod_{k=1}^d \frac{x(z)-x(\alpha_k)}{
    x(z)-x(\varepsilon_k)}$ and $\{0,\pm \alpha_1,...,\pm \alpha_d\}$ are the solutions of $x(z)+y(z)=0$. 
The diagonal function
is also expressed in terms of $\mathAE(x(z))$
\cite{Schurmann:2019mzu}:
\begin{align}
U^{(0)}_1(z,z)
=
\frac{2(x(z)-x(0))(\mathAE(x(z)))^2}{
  (x(z)+y(z))^2}  \;.
\label{U01-diag}
\end{align}

\begin{remark}
  \label{remark:sym}
  The equations \eqref{DSE-U} and \eqref{DSE-V} were derived in
  \cite{Branahl:2020yru} from identities for the two-point functions
  $\langle M_{ab} M_{ba}\rangle= \int_{H_N} M_{ab} M_{ba}\,
  d\mu_{E,\lambda}(M)$ and
  $\langle M_{aa} M_{bb}\rangle= \int_{H_N} M_{aa} M_{bb}\,
  d\mu_{E,\lambda}(M)$, which are symmetric in $a,b$. The derivation
  made a choice in the order of $a,b$.  Repeating all steps in the
  other order one would get an equation
\begin{align}
& (x(z)+y(w))U^{(g)}_{|I|+1}(z,w;I)
+\frac{\lambda}{N}\sum_{k=1}^d \frac{r_k
U^{(g)}_{|I|+1}(z,\varepsilon_k;I)}{x(w)-x(\varepsilon_k)}
\tag{\ref{DSE-U}'}
\label{DSE-Uprine}
\\
&=\delta_{|I|,0}\delta_{g,0}
+\lambda \hspace*{-3mm}
\sum_{\substack{ I_1\uplus I_2=I \\ g_1+g_2=g \\ (g_1,I_1) \neq (0,\emptyset)  }}
\hspace*{-3mm}
W^{(g_1)}_{|I_1|+1}(w;I_1) U^{(g_2)}_{|_2|+1}(z,w;I_2)
-\lambda 
  \sum_{j=1}^{|I|}
\frac{U^{(g)}_{|I|}(z,u_i;I\setminus u_j)}{
  x(u_j)-x(w)}
\nonumber
\\
&
-\lambda \frac{\partial}{\partial x(s)}
U^{(g-1)}_{|I|+2}(z,w;I\cup\{s\})\Big|_{s=w}
- \lambda \frac{V^{(g-1)}_{|I|+1}(z,w;I)-V^{(g-1)}_{|I|+1}(z,z;I)}{
  x(z)-x(w)}  
\nonumber
\end{align}
and similarly for $V^{(g)}_{|I|+1}(z,w;I)$. Exchanging $w\leftrightarrow z$
in \eqref{DSE-Uprine} and comparing with \eqref{DSE-U} shows that  the pairs
$(U^{(g)}_{|I|+1}(z,w;I),V^{(g)}_{|I|+1}(z,w;I))$ and
$(U^{(g)}_{|I|+1}(w,z;I),V^{(g)}_{|I|+1}(w,z;I))$ satisfy the same system of
equations. Since the solution is unique according to
the construction below, we necessarily have symmetry
$U^{(g)}_{|I|+1}(z,w;I)=U^{(g)}_{|I|+1}(w,z;I)$ and
$V^{(g)}_{|I|+1}(z,w;I)=V^{(g)}_{|I|+1}(w,z;I))$ in the first two arguments.
\end{remark}

The solution of the system (\ref{DSE-U}),(\ref{DSE-V}) and (\ref{DSE-W})
in \cite{Branahl:2020yru} for
$(g,n)=\{(0,2),(0,3),(0,4),(1,1)\}$
provided strong support for the
conjecture that the meromorphic differentials $\omega^{(g)}_{n}$ obey
\emph{blobbed topological recursion}, a systematic extension of TR due
to Borot and Shadrin \cite{Borot:2015hna}.
In \cite{Hock:2021tbl} we succeeded in solving
the genus $g=0$ sector in larger generality. The result of
\cite{Hock:2021tbl}, restricted to $y(z)=-x(-z)$, is covered by
Thm.~\ref{thm:mainrecursion} for $g=0$.

\subsection{Auxiliary functions}

Our proof of Theorem \ref{thm:mainrecursion} follows a very different
strategy than the one for genus $g=0$ given in
\cite{Hock:2021tbl}. The key idea is to rewrite the Dyson-Schwinger
equations \eqref{DSE-U} and \eqref{DSE-V} into equations for auxiliary
functions which in one or two arguments are symmetric in the preimages
of $x$ (given in (\ref{R})).  For genus $g\geq 1$ it is essential that
$y(z)=-x(-z)$.  We separate the arguments in the functions below by a
comma if the function is symmetric when exchanging the arguments,
otherwise by a semicolon.
\begin{definition}
\label{def:partialfractions}
  For $I=\{u_1,...,u_n\}$ the following combinations of the functions $U,V$
  are introduced:
  \begin{subequations}
    \label{def-all}
  \begin{align}
  &  H^{(g)}_{|I|+1}(x(v);z;I)
  \label{def-H}
  \\
  &:=
  \delta_{g,0}\delta_{I,\emptyset}
  -\frac{\lambda}{N} \sum_{l=1}^d 
\frac{r_l U^{(g)}_{|I|+1}(z,\varepsilon_l;I)}{
    x(v)-x(\varepsilon_l)}  
  +\lambda \sum_{j=1}^{|I|} \frac{U^{(g)}_{|I|}(z,u_j;I{\setminus}u_j)}{
    x(v)-x(u_j)}
-\lambda \frac{V^{(g-1)}_{|I|+1}(z,z;I)}{x(v)-x(z)}\;,\nonumber
\\
&P^{(g)}_{|I|+1}(x(v),x(z);I)
\label{def-P}
\\
&:= \frac{\lambda \delta_{|I|,1}\delta_{g,0}}{x(v)-x(u_1)}
+\delta_{|I|,0}\delta_{g,0}
\Big(x(v)+x(z)
-\frac{\lambda}{N}\sum_{k=1}^d \frac{r_k}{x(v)-x(\varepsilon_k)}\Big)
\nonumber
\\
&-\frac{\lambda}{N} \sum_{k=1}^d 
\frac{r_k H^{(g)}_{|I|+1}(x(v);\varepsilon_k;I)}{x(z)-x(\varepsilon_k)}
+\lambda 
\sum_{j=1}^{|I|}\frac{H^{(g)}_{|I|}(x(v);u_j;I\setminus u_j)}{
  x(z)-x(u_j)}\;,
  \nonumber
  \\
&M^{(g)}_{|I|+1}(x(v);z;I)
  \label{def-M}
  \\
  &:=  -\frac{\lambda}{N}\sum_{l=1}^d 
\frac{r_l V^{(g)}_{|I|+1}(z,\varepsilon_l;I)}{
    x(v)-x(\varepsilon_l)}
+\lambda \sum_{j=1}^{|I|} \frac{V^{(g)}_{|I|}(z,u_j;I{\setminus}u_j)}{
  x(v)-x(u_j)}
-\lambda \frac{U^{(g)}_{|I|+1}(z,z;I)}{x(v)-x(z)}\;,
\nonumber
\\
&Q^{(g)}_{|I|+1}(x(v),x(z);I)  
\label{def-Q}
\\
&:=-\frac{\lambda}{N}\sum_{k=1}^d r_k 
  \frac{M^{(g)}_{|I|+1}(x(v);\varepsilon_k;I)}{x(z)-x(\varepsilon_k)}
+\lambda \sum_{j=1}^{|I|}
\frac{M^{(g)}_{|I|+1}(x(v);u_j;I{\setminus} u_j)}{x(z)-x(u_j)}\;.
\nonumber
\end{align}
\end{subequations}
\end{definition}

\begin{proposition}\label{Prop:DSE}
  The Dyson-Schwinger equations \eqref{DSE-U}, \eqref{DSE-V} and 
\eqref{DSE-W}  imply:  
\begin{subequations}
  \label{DSE-all}
  \begin{align}
&H^{(g)}_{|I|+1}(x(v);z;I)
\label{DSE-UH}
\\
&= (x(z)+y(v)) U^{(g)}_{|I|+1}(v,z;I)
+\lambda
\sum_{\substack{ I_1\uplus I_2=I \\ g_1+g_2=g \\ (g_1,I_1)\neq (0,\emptyset)}}
W^{(g_1)}_{|I_1|+1}(v;I_1)
U^{(g_2)}_{|I_2|+1}(v,z;I_2)
\nonumber
\\
&+\lambda \frac{\partial}{\partial x(s)}U^{(g-1)}_{|I|+2}(v,z;I\cup s)\Big|_{s=v}
+ \lambda \frac{V^{(g-1)}_{|I|+1}(v,z;I)}{x(z)-x(v)}\;,
\nonumber
\\
&P^{(g)}_{|I|+1}(x(v),x(z);I)
\label{DSE-HP}
\\
&=  (x(v)+y(z)) H^{(g)}_{|I|+1}(x(v);z;I)
+\lambda
\sum_{\substack{ I_1\uplus I_2=I \\g_1+g_1=g \\ (g_1,I_1)\neq (0,\emptyset)}}
W^{(g_1)}_{|I_1|+1}(z;I_1)
H^{(g)}_{|I_2|+1}(x(v);z;I_2)
\nonumber
\\
&+\lambda
\frac{\partial}{\partial x(s) } H^{(g-1)}_{|I|+2}(x(v);z;I\cup s)\Big|_{s=z}
+ \lambda \frac{ M^{(g-1)}_{|I|+1}(x(v);z;I)}{x(v)-x(z)}\;,
\nonumber
\\
&M^{(g)}_{|I|+1}(x(v);z;I)
\label{DSE-VM}
\\
&=(x(v)+y(v)) V^{(g)}_{|I|+1}(v,z;I)
+\lambda
\sum_{\substack{ I_1\uplus I_2=I\\  g_1+g_2=g \\ (g_1,I_1)\neq (0,\emptyset)} }
W^{(g_1)}_{|I_1|+1}(v;I_1) V^{(g_2)}_{|I_2|+1}(v,z;I_2)
\nonumber
\\
&+\lambda \frac{\partial}{\partial x(s)}
V^{(g-1)}_{|I|+2}(v,z;I\cup s)\Big|_{s=v}
+\lambda \frac{U^{(g)}_{|I|+1}(v,z;I)}{x(z)-x(v)} \;,
\nonumber
\\
&Q^{(g)}_{|I|+1}(x(v),x(z);I)
\label{DSE-MQ}
\\
&=  (x(z)+y(z))  M^{(g)}_{|I|+1}(x(v);z;I)
 +\lambda \sum_{\substack{ I_1\uplus I_2=I\\  g_1+g_2=g \\ (g_1,I_1)\neq (0,\emptyset)} }
W^{(g_1)}_{|I_1|+1}(z;I_1)
M^{(g_2)}_{|I_2|+1}(x(v);z;I_2)
\nonumber
\\
&+\lambda \frac{\partial}{\partial x(s)}
M^{(g-1)}_{|I|+2}(x(v);z;I\cup s)\Big|_{s=z}
+
\lambda \frac{ H^{(g)}_{|I|+1}(x(v);z;I)}{x(v)-x(z)}\;.
\nonumber
\end{align}
\end{subequations}
\begin{proof}
  Equations (\ref{DSE-UH}) and (\ref{DSE-VM}) are an obvious rewriting
  of \eqref{DSE-U} and \eqref{DSE-V}, respectively.  For the proof of
  (\ref{DSE-HP}) one starts from (\ref{DSE-UH}) for
  $z\mapsto \varepsilon_k$, multiplied by multiplies by
  $\frac{\lambda}{N} \frac{r_k}{x(z)-x(\varepsilon_k)}$ and summed
  over $k$. A rearrangement taking \eqref{DSE-W} into account gives the
  assertion. Similarly for (\ref{DSE-MQ}).
\end{proof}  
\end{proposition}
Inserting $U^{(0)}_1(v,z)$ given by (\ref{U01}) into (\ref{DSE-UH}) and
the result into (\ref{DSE-HP}) gives
\begin{subequations}
\begin{align}
H^{(0)}_1(x(v);z):=
H^{(0)}_1(x(v);z;\emptyset)&=
\prod_{k=1}^d \frac{x(v)+y(\hat{z}^k)}{x(v)-x(\varepsilon_k)}\;,
\label{H01}
\\
P^{(0)}_1(x(v),x(z)):=
P^{(0)}_1(x(v);x(z);\emptyset)
&= 
\frac{\prod_{k=0}^d(x(v)+y(\hat{z}^k))}{
  \prod_{k=1}^d(x(v)-x(\varepsilon_k))}\;.
\label{P01}
\end{align}
\end{subequations}

\begin{proposition}
  The functions $P^{(g)}_n$ and $Q^{(g)}_n$ are symmetric in their
  first two arguments.
 \begin{proof} 
   Expressing all $H^{(g')}_{n'}$ and $M^{(g')}_{n'}$ in
   (\ref{DSE-HP}) and  (\ref{DSE-MQ}) in terms of 
$U^{(g'')}_{n''}$ and $V^{(g'')}_{n''}$ via 
(\ref{DSE-UH}) and  (\ref{DSE-VM}) gives a
manifestly symmetric expression when taking
the symmetries $U^{(g)}_{|I|+1}(z,v;I)=U^{(g)}_{|I|+1}(v,z;I)$
and $V^{(g)}_{|I|+1}(z,v;I)=V^{(g)}_{|I|+1}(v,z;I)$ according to
Remark~\ref{remark:sym} into account.
\end{proof}
\end{proposition}

The Dyson-Schwinger equations 
(\ref{DSE-UH}), (\ref{DSE-HP}), (\ref{DSE-VM}) and  (\ref{DSE-MQ})
can be disentangled into two separate systems:
\begin{proposition}
  Let $(C_h)_{h\in \mathbb{N}}$ be the sequence of Catalan numbers, defined
  for instance by $C_0=1$ and $C_{h+1}=\sum_{l=0}^{h} C_lC_{h-l}$.
Then the linear combinations
\begin{subequations}
  \label{def-hall}
  \begin{align}
  &\hat{U}^{(g)}_{|I|+1}(v,z;I)
\label{def-hU}
  \\
  &:=
U^{(g)}_{|I|+1}(v,z;I)
+\sum_{h=0}^{g-1} \frac{(-1)^h C_h\lambda^{1+2h}}{(x(v)-x(z))^{2+4h}}
V^{(g-1-h)}_{|I|+1}(v,z;I)\;,
\nonumber
\\
  &\hat{H}^{(g)}_{|I|+1}(x(v);z;I)
\label{def-hH}
  \\
  &:=
H^{(g)}_{|I|+1}(x(v);z;I)
+\sum_{h=0}^{g-1} \frac{(-1)^h C_h\lambda^{1+2h}}{(x(v)-x(z))^{2+4h}}
M^{(g-1-h)}_{|I|+1}(x(v);z;I)\;,
\nonumber
\\
&\hat{P}^{(g)}_{|I|+1}(x(v),x(z);I)
\label{def-hP}
\\
&:=P^{(g)}_{|I|+1}(x(v),x(z);I)
+\sum_{h=0}^{g-1}  \frac{(-1)^h C_h\lambda^{1+2h}}{(x(v)-x(z))^{2+4h}}
Q^{(g-1-h)}_{|I|+1}(x(v),x(z);I)\;,
\nonumber
\\
&\hat{V}^{(g)}_{|I|+1}(z,v;I)
\label{def-tV}
  \\
  &:=
V^{(g)}_{|I|+1}(v,z;I)
-\sum_{h=0}^{g} \frac{(-1)^h C_h\lambda^{1+2h}}{(x(v)-x(z))^{2+4h}}
U^{(g-h)}_{|I|+1}(v,z;I)\;,
\nonumber
\\
 &\hat{M}^{(g)}_{|I|+1}(x(v);z;I)
\label{def-tM}
  \\
  &:=
M^{(g)}_{|I|+1}(x(v);z;I)
-\sum_{h=0}^{g} \frac{(-1)^h C_h\lambda^{1+2h}}{(x(v)-x(z))^{2+4h}}
H^{(g-h)}_{|I|+1}(x(v);z;I)\;,
\nonumber
\\
&\hat{Q}^{(g)}_{|I|+1}(x(v),x(z);I)
\label{def-tQ}
\\
&:=Q^{(g)}_{|I|+1}(x(v),x(z);I)
-\sum_{h=0}^{g}  \frac{(-1)^h C_h\lambda^{1+2h}}{(x(v)-x(z))^{2+4h}}
P^{(g-h)}_{|I|+1}(x(v),x(z);I)
\nonumber
\end{align}
\end{subequations}
satisfy
\begin{subequations}
  \label{DSE-hall}
  \begin{align}
&\hat{H}^{(g)}_{I|+1}(x(v);z;I)
\label{DSE-hUH}
\\
&= 
(x(z)+y(v)) \hat{U}^{(g)}_{|I|+1}(v,z;I)
+\lambda \frac{\partial\hat{U}^{(g-1)}_{|I|+2}(v,z;I\cup s)
}{\partial x(s)}\Big|_{s=v}
\nonumber
\\
&+\lambda \sum_{\substack{ I_1\uplus I_2=I\\  g_1+g_2=g \\ (g_1,I_1)\neq (0,\emptyset)} }
\Big(W^{(g_1)}_{|I_1|+1}(v;I_1)
- \frac{(-1)^{g_1} \lambda^{2g_1-1} \delta_{I_1,\emptyset} C_{g_1-1} }{
(x(z)-x(v))^{4g_1-1}}
\Big)
\hat{U}^{(g_2)}_{|I_2|+1}(v,z;I_2)\;,
\nonumber
\\
&\hat{P}^{(g)}_{I|+1}(x(v),x(z);I)
\label{DSE-hHP}
\\
&= 
(x(v)+y(z)) \hat{H}^{(g)}_{|I|+1}(x(v);z;I)
+\lambda \frac{\partial\hat{H}^{(g-1)}_{|I|+2}(x(v);z;I\cup s)
}{\partial x(s)}\Big|_{s=z}
\nonumber
\\
&+\lambda \sum_{\substack{ I_1\uplus I_2=I\\  g_1+g_2=g \\ (g_1,I_1)\neq (0,\emptyset)} }
\Big(W^{(g_1)}_{|I_1|+1}(z;I_1)
- \frac{(-1)^{g_1} \lambda^{2g_1-1} \delta_{I_1,\emptyset} C_{g_1-1} }{
(x(v)-x(z))^{4g_1-1}}
\Big)
\hat{H}^{(g_2)}_{|I_2|+1}(x(v);z;I_2)\;,
\nonumber
\\
&\hat{M}^{(g)}_{I|+1}(x(v);z;I)
\label{DSE-tVM}
\\
&= 
(x(v)+y(v)) \hat{V}^{(g)}_{|I|+1}(v,z;I)
+\lambda \frac{\partial\hat{V}^{(g-1)}_{|I|+2}(v,z;I\cup s)
}{\partial x(s)}\Big|_{s=z}
\nonumber
\\
&+\lambda \sum_{\substack{ I_1\uplus I_2=I\\  g_1+g_2=g \\ (g_1,I_1)\neq (0,\emptyset)} }
\Big(W^{(g_1)}_{|I_1|+1}(v;I_1)
- \frac{(-1)^{g_1} \lambda^{2g_1-1} \delta_{I_1,\emptyset} C_{g_1-1} }{
(x(v)-x(z))^{4g_1-1}}
\Big)
\hat{V}^{(g_2)}_{|I_2|+1}(v,z;I_2)\;,
\nonumber
\\
&\hat{Q}^{(g)}_{I|+1}(x(v),x(z);I)
\label{DSE-tMQ}
\\
&= 
(x(z)+y(z))\hat{M}^{(g)}_{|I|+1}(x(v);z;I)
+\lambda \frac{\partial\hat{M}^{(g-1)}_{|I|+2}(x(v);z;I\cup s)
}{\partial x(s)}\Big|_{s=z}
\nonumber
\\
&+\lambda \sum_{\substack{ I_1\uplus I_2=I\\  g_1+g_2=g \\ (g_1,I_1)\neq (0,\emptyset)} }
\Big(W^{(g_1)}_{|I_1|+1}(z;I_1)
- \frac{(-1)^{g_1} \lambda^{2g_1-1} \delta_{I_1,\emptyset} C_{g_1-1} }{
(x(v)-x(z))^{4g_1-1}}
\Big)
\hat{M}^{(g_2)}_{|I_2|+1}(x(v);z;I_2)\;.
\nonumber
\end{align}
\end{subequations}
\begin{proof}
  Combining different DSEs from Proposition \ref{Prop:DSE} with the
  Segner recursion formula \( C_{h+1} = \sum_{l=0}^{h} C_l C_{h-l} \)
  yields the desired equations after a lengthy but straightforward
  computation.
\end{proof}  
\end{proposition}
In the following cases it is safe to omit the hat:
$\hat{U}^{(0)}_{n+1}=U^{(0)}_{n+1}$,
$\hat{H}^{(0)}_{n+1}=H^{(0)}_{n+1}$,
$\hat{P}^{(0)}_{n+1}=P^{(0)}_{n+1}$. 

\subsection{Motivation of a loop insertion operator}\label{Sec.motivateloop}

The loop insertion operator plays an important role in the theory of topological recursion, but it is a very subtle object. There are two important examples with quite different realisations for the loop insertion operator in the literature:
\begin{itemize}
\item Hermitian 1-matrix model \cite[Sec. 5.4.6.2]{Eynard:2016yaa}: The loop insertion operator acts on the generating function of correlation functions. One has to assume a family of spectral curves with infinitely many parameters to make some sense of a loop insertion operator. The loop insertion operator is a deformation with respect to the coefficients of the potential, where the potential needs to include all (infinitely many) parameters.
\item Classical Kontsevich model: The loop insertion operator is
  represented as a formal derivative with respect to the
  eigenvalues\footnote{In the notation of this paper, the
    external matrix is $E$; its eigenvalues are $e_k$.} $\lambda_k$ of
  the external matrix $\Lambda$, see \cite[Thm
  6.4.5]{Eynard:2016yaa}. Also here, a family of spectral curves is
  considered with infinitely many parameters, which are the
  eigenvalues of the external matrix.
\end{itemize}

A general construction of the loop insertion operator was formulated
abstractly via the Rauch variational formula in the original article
on topological recursion \cite[Thm 5.1]{Eynard:2007kz}. However, since
our article deals with a set of DSEs whose solutions are not generally
governed by topological recursion, a rigorous understanding of a loop
insertion operator is even more subtle.

The basic idea of deriving the DSEs \eqref{DSE-U} and \eqref{DSE-V}
originates from the PhD thesis of one of the authors
\cite{Hock:2020rje}, and was worked out in detail in
\cite{Branahl:2020yru}. The main idea is to start with two initial
DSEs and act with formal derivatives of
$\frac{\partial}{\partial E_k}$ for different $k$'s on these
equations. It turns out that the objects and DSEs defined in this way
have a very natural complex continuation to meromorphic functions, see
\cite[Def. 3.1 \& 3.5]{Branahl:2020yru}, which is the origin of the
DSEs considered in this article. The differentiation with respect to
$E_k$ plays the same role as in the classical Kontsevich model, where
it plays the role of the loop insertion operator. Consequently, the
derivation with respect to $E_k$ would, from a topological recursion
perspective, consider a family of spectral curves with infinitely many
parameters $E_k$, the eigenvalues of the external matrix. The
justification that $\frac{\partial}{\partial E_k}$ plays the role of
the loop insertion operator in the present model is one important
insight of this article, even though a more rigorous construction can
not be provided, yet.

A formal extension of the derivative with respect to $E_k$
\begin{align}\label{loopinsert}
    \frac{\partial}{\partial E_k}\to D_u
\end{align}
on meromorphic functions will be defined in the next subsection. We
want to emphasise that the definition in Sec.\ \ref{sec.loopdef} is
purely symbolical, but has the interpretation of a loop insertion
operator which justifies the name. We want to work with this formal
action (as described in the subsequent subsection) to rewrite long and
involved expressions. The results derived in the rest of the article
are obtained purely from the DSEs and do not rely on the actual
existence of a \textit{rigorous} loop insertion operator. The
existence and uniqueness of a rigorously defined loop insertion
operator for the quartic Kontsevich model will not be discussed in
this article.

Disclaimer: The above symbolically defined operator $D_u$ acts in the
desired way on the DSEs, however not on thier solutions. This comes
from the fact that $D_u$ should be considered as a complex
continuation of $\frac{\partial}{\partial E_k}$ as in
\eqref{loopinsert}. Therefore, any term that vanishes at $u=E_k$ could
be added. However, there are restrictions from the fact that the
resulting meromorphic function should be symmetric after complex
continuation, see Rem. \ref{Rem:loopinsert} and \ref{Rem:loopinsert2}
later. This seems to be different from the usual rigorous loop
insertion operators in topological recursion related to an enumerative
problem; see, for instance, the 2-Hermitian matrix model
\cite{Eynard:2002kg} which is solved by topological recursion
\cite{Chekhov:2006vd}.

\subsection{Definition of a loop insertion operator
  on meromorphic functions}\label{sec.loopdef}

Let $\tilde{\mathcal{R}}$ be the ring of $\mathbb{Q}$-polynomials
in the variables
\begin{subequations}
\label{R-variables}
  \begin{align}
&\Big\{x,\ y,\ W^{(g)}_{n+1},\ \hat{U}^{(g)}_{n+1},\ \hat{H}^{(g)}_{n+1},\ 
\hat{P}^{(g)}_{n+1},
\hat{V}^{(g)}_{n+1}, \hat{M}^{(g)}_{n+1},\ \hat{Q}^{(g)}_{n+1}\Big\},\ 
\label{R-variables-primary}
\end{align}
in $x(\,.\,)$-derivatives of them (e.g.\ 
$\frac{\partial^2 y(z)}{\partial (x(z))^2}$,
  $\frac{\partial^2 W^{(0)}_{n+1}(z;u_1,...,u_n)}{\partial x(z) \partial
    x(u_1)}$,
  $\frac{\partial^4 \hat{P}^{(g)}_{n+1}(x(z),x(w);u_1,...,u_n)}{
    \partial (x(z))^2 \partial x(w) \partial x(u_n)}$)
as well as reciprocals
\begin{align}
&  \Big\{\frac{1}{x(*)-x(\star)},\ \frac{1}{x(*)+y(\star)},\
\frac{1}{U^{(0)}_1},\ \frac{1}{H^{(0)}_1},\ \frac{1}{P^{(0)}_1},\ 
\frac{1}{\hat{V}^{(0)}_1},\ \frac{1}{\hat{M}^{(0)}_1},\ 
\frac{1}{\hat{Q}^{(0)}_1}\Big\}\;.
\label{R-variables-reciprocals}
\end{align}
\end{subequations}
We let $\mathcal{I}$ be the ideal in $\tilde{\mathcal{R}}$ generated by
\begin{align*}
  \Big\{ &(x(v)+y(z))\cdot \frac{1}{(x(v)+y(z))}-1,\quad
  U^{(0)}_1(v,z)\cdot  \frac{1}{U^{(0)}_1(v,z)}-1,
  \\
 & H^{(0)}_1(x(v);z)\cdot  \frac{1}{H^{(0)}_1(x(v);z)}-1,\quad
  P^{(0)}_1(x(v),x(z))\cdot  \frac{1}{P^{(0)}_1(x(v),x(z))}-1
  \Big\}
\end{align*}
and $\mathcal{R}=\tilde{\mathcal{R}}/\mathcal{I}$ be the quotient.
The variables (\ref{R-variables}) belong to the field of meromorphic
functions on several copies of $\hat{\mathbb{C}}$.  For our purpose
they are considered as \emph{independent variables}.  We also consider
$\hat{P}^{(g)}_{n+1}(x(v),x(z);I)$ as independent of
$\hat{P}^{(g)}_{n+1}(x(v'),x(z');I')$ whenever $x(v)\neq x(v')$ or
$x(z)\neq x(z')$ or $I\neq I'$.  Similarly for $\hat{Q}^{(g)}_{n+1}$.
We consider $\hat{H}^{(g)}_{n+1}(x(v);z;I)$ as independent of
$\hat{H}^{(g)}_{n+1}(x(v');z';I')$ whenever $x(v)\neq x(v')$ or
$z\neq z'$ or $I\neq I'$. Similarly for $\hat{M}^{(g)}_{n+1}$.  We
consider $\hat{U}^{(g)}_{n+1}(v,z;I)$ as independent of
$\hat{H}^{(g)}_{n+1}(v';z';I')$ whenever $v\neq v'$ or $z\neq z'$ or
$I\neq I'$. Similarly for $\hat{V}^{(g)}_{n+1}$.

The Dyson-Schwinger equations (\ref{DSE-hall}) are polynomial
equations $f_i=0$ for $f_i \in \mathcal{R}$. In addition we have
relations between residues which follow from (\ref{def-all}) and
(\ref{def-hall}) as well as from the condition
$[(x(v))^{-1}] H^{(g)}_{|I|+1}(x(v);z;I)=-\lambda W^{(g)}_{n+1}(z;I)$.

The most decisive tool in our construction is a \emph{loop
  insertion operator}.
\begin{definition}
\label{def:loopins}
  For $u\in \hat{\mathbb{C}}$, the loop insertion operator
  $D_u:\mathcal{R}\to \mathcal{R}$ is defined on the variables
  \eqref{R-variables-primary} as
\begin{align*}
  &D_u(x(z))=0,\quad D_uy(z)=\lambda W^{(0)}_2(z;u),\quad
  D_uW^{(g)}_{|I|+1}(z;I)=  W^{(g)}_{|I|+2}(z;I\cup u),\
  \\
  &D_u\hat{U}^{(g)}_{|I|+1}(v,z;I)=
  \hat{U}^{(g)}_{|I|+2}(v,z;I\cup u),\quad
  D_u\hat{V}^{(g)}_{|I|+1}(v,z;I)=
  \hat{V}^{(g)}_{|I|+2}(v,z;I\cup u),\
  \\
  &D_u\hat{H}^{(g)}_{|I|+1}(x(v);z;I)=
  \hat{H}^{(g)}_{|I|+2}(x(v);z;I\cup u),\   
  \\
&D_u\hat{M}^{(g)}_{|I|+1}(x(v);z;I)=
\hat{M}^{(g)}_{|I|+2}(x(v);z;I\cup u),\quad
\\
&D_u\hat{P}^{(g)}_{|I|+1}(x(v),x(z);I)=
  \hat{P}^{(g)}_{|I|+2}(x(v),x(z);I\cup u),\quad   
  \\
  &
  D_u\hat{Q}^{(g)}_{|I|+1}(x(v),x(z);I)=
  \hat{Q}^{(g)}_{|I|+2}(x(v),x(z);I\cup u)
\end{align*}
and extended to $\mathcal{R}$ by linearity, Leibniz rule, the requirement
$D_u:\mathcal{I}\to \mathcal{I}$ and commutation with
any $x(\,.\,)$-derivative. For
$J=\{u_1,...,u_n\}$ we let $D_J:=D_{u_1}\cdots D_{u_n}:\mathcal{R}
\to \mathcal{R}$. Moreover, $D_\emptyset$ is defined as the identity operator. 
\end{definition} 
The condition $D_u:\mathcal{I}\to \mathcal{I}$ just means that the
loop insertion operator of the reciprocals in
(\ref{R-variables-reciprocals}) is given by the usual rules of calculus, e.g.\
\begin{align*}
\mathcal{I}\ni &D_u\Big((x(v)+y(z))\cdot \frac{1}{(x(v)+y(z))}\Big) -D_u 1
  \\
  &=(x(v)+y(z))\Big(
  \frac{\lambda W^{(0)}_2(z;u)}{(x(v)+y(z))^2}
  +
  D_u\Big(\frac{1}{x(v)+y(z)}\Big)\Big),
\end{align*}
i.e.\  $ D_u\big(\frac{1}{x(v)+y(z)}\big)=-
\frac{\lambda W^{(0)}_2(z;u)}{(x(v)+y(z))^2}+\mathcal{I}$.
We can also apply $D_I$ to logarithms 
\begin{align}
  D_I \log H^{(0)}_1(x(v);z)&= \sum_{l=1}^{|I|} \frac{(-1)^{l-1}}{l}
  \sum_{  \substack{I_1\uplus ...\uplus I_l=I\\ 
      I_1, ... I_l\neq \emptyset}} 
  \prod_{j=1}^l \frac{H^{(0)}_{|I_j|+1}(x(v);z;I_j)}{ H^{(0)}_1(x(v);z)}\;,
  \label{def-DlogH}
  \\
  D_I \log P^{(0)}_1(x(v),x(z))&= \sum_{l=1}^{|I|} \frac{(-1)^{l-1}}{l}
  \sum_{  \substack{I_1\uplus ...\uplus I_l=I\\ 
      I_1, ... I_l\neq \emptyset}} 
  \prod_{j=1}^l \frac{P^{(0)}_{|I_j|+1}(x(v),x(z);I_j)}{ P^{(0)}_1(x(v),x(z))}
  \nonumber
\end{align}  
and similarly for $D_I \log\hat{M}^{(0)}_1(x(x);z)$ and
$D_I \log \hat{Q}^{(0)}_1(x(v),x(z))$. In the same way one has 
\begin{align}
  D_I \log (x(v)+y(z))
  &= \sum_{l=1}^{|I|} \frac{(-1)^{l-1}}{l}
  \sum_{  \substack{I_1\uplus ...\uplus I_l=I\\ 
      I_1, ... I_l\neq \emptyset}} 
  \prod_{j=1}^l \frac{\lambda W^{(0)}_{|I_j|+1}(z;I_j)}{x(v)+y(z)}\;.
\label{DIlogRvz}
\end{align}
The loop insertion operator is compatible with the
Dyson-Schwinger equations (\ref{DSE-hall}).

\section{Solution for genus $g=0$}

\label{sec:g0}

\subsection{Auxiliary functions for   genus $g = 0$}

Starting point is the following observation:
\begin{lemma}
\label{lemma-DIlogPH}
  For any $I=\{u_1,...,u_n\}$ one has
\begin {align}
  D_I \log P^{(0)}_1(x(v),x(z))&=D_I\log (x(v)+y(z)) 
+  D_I \log H^{(0)}_1(x(v);z)\;,
\label{DIlogPH}
\end{align}
where the three terms are short-hand notations for
\eqref{def-DlogH} and \eqref{DIlogRvz}.
\begin{proof}
This is a combinatorial rewriting of (\ref{DSE-HP}), which reads
\begin{align}
  \frac{P^{(0)}_{|I|+1}(x(v),x(z);I)}{P^{(0)}_1(x(v),x(z))}
&=  \frac{H^{(0)}_{|I|+1}(x(v);z;I)}{H^{(0)}_1(x(v);z)}
+\frac{\lambda W^{(0)}_{|I|+1}(z;I)}{x(v)+y(z)}
\nonumber
\\
&+\sum_{\substack{I'\uplus I''\\I',I''\neq \emptyset}}
\frac{\lambda W^{(0)}_{|I'|+1}(z;I')}{x(v)+y(z)}
\frac{H^{(0)}_{|I''|+1}(x(v);z;I'')}{H^{(0)}_1(x(v);z)}\;.
\label{DSE-HPquot}
\end{align}
We arrange products of these expressions into
$D_I \log P^{(0)}_1(x(v),x(z))$ in the second line of 
(\ref{def-DlogH}). A contribution with $h\geq 1$ factors of
$\frac{H^{(0)}_{|I_i|+1}(x(v);z;I_i)}{H^{(0)}_1(x(v);z)}$ and $w\geq 1$ factors of
$\frac{\lambda W^{(0)}_{|I_j|+1}(z;I_j)}{x(v)+y(z)}$ arises via the
binomial theorem in $\frac{(w+h-k)!}{(w-k)!(h-k)k!}$ different ways
if $k$ times a factor of the second line of (\ref{DSE-HPquot}) occurs.
The prefactor of such a term in $D_I \log P^{(0)}_1(x(v),x(z))$
is $\frac{(-1)^{w+h-k-1}}{(w+h-k)}$. The sum over $k$ equals
\begin{align*}
\sum_{k=0}^{\min (w,h)}
 (-1)^k \frac{(w+h-k-1)!}{(w-k)!(h-k)k!}=0\;,
\end{align*} 
which follows e.g.\ from
\cite[vol.~4, eq.\ (10.20)]{Gould}
\[
  \sum_{k=0}^n (-1)^k\binom{n}{k}\binom{x-k}{j}=
  \binom{x-n}{j-n}
\]  
when setting $n=h$,  $x=w+h-1$ and $j=h-1$. Hence,
all contributions with at least one factor
$\frac{H^{(0)}_{|I_i|+1}(x(v);z;I_i)}{H^{(0)}_1(x(v);z)}$ and at least one
factor $\frac{\lambda W^{(0)}_{|I_j|+1}(z;I_j)}{x(v)+y(z)}$ cancel,
and the assertion follows.
\end{proof}
\end{lemma}
In complete analogy to Lemma~\ref{lemma-DIlogPH} one establishes
\begin{subequations}
\begin{align}
D_I \log H^{(0)}_1(x(v);z)&=D_I\log (x(z)+y(v)) 
+  D_I \log U^{(0)}_1(v,z)\;,
\label{DIlogU01}
\\
D_I \log \hat{M}^{(0)}_1(x(v),x(z))&=D_I\log (x(v)+y(v)) 
+  D_I \log\hat{V}^{(0)}_1(v,z)\;,
\label{DIlogV01}
\\
D_I \log \hat{Q}^{(0)}_1(x(v),x(z))&=D_I\log (x(z)+y(z)) 
+  D_I \log\hat{M}^{(0)}_1(x(v);z)\;.
\label{DIlogM01}
\end{align}
\end{subequations}

\begin{lemma}
\label{lemma:DIUH}
  \begin{align}
  D_I \frac{U^{(0)}_{1}(v,z)}{H^{(0)}_1(x(v);z)}
&\equiv
\sum_{l=0}^{|I|}
  \sum_{  \substack{I_0\uplus ...\uplus I_l=I\setminus u_j\\
      I_1,...I_l\neq \emptyset}}
      \frac{U^{(0)}_{|I_0|+1}(v,z;I_0)}{H^{(0)}_1(x(v);z)}
(-1)^{l}
\prod_{i=1}^l \frac{H^{(0)}_{|I_i|+1}(x(v);z;I_i)}{ H^{(0)}_1(x(v);z)}
\nonumber
\\
&=D_{I}\frac{1}{x(z)+y(v)}\;.
\label{DUH-g0}
\end{align}
\begin{proof}
Clear for $I =\emptyset$. For $I\neq \emptyset$, 
start from the Dyson-Schwinger equation (\ref{DSE-UH}) 
\begin{align*}
\frac{U^{(0)}_{|I_0|+1}(v,z;I)}{H^{(0)}_1(x(v);z)}
&=\frac{H^{(0)}_{|I|+1}(x(v),z;I)}{
  (x(z)+y(v)) H^{(0)}_1(x(v);z)}
+\frac{(-\lambda) W^{(0)}_{|I|+1}(v;I)}{(x(z)+y(v))^2}
\\
&+
\sum_{\substack{I'\uplus I''=I\\ I',I''\neq \emptyset}}
\frac{(-\lambda) W^{(0)}_{|I'|+1}(v;I')}{(x(z)+y(v))}
\frac{U^{(0)}_{|I''|+1}(v,z;I'')}{H^{(0)}_1(x(v);z)}\;,
\end{align*}
which iterates to   
\begin{align*}
\frac{U^{(0)}_{|I_0|+1}(v,z;I)}{H^{(0)}_1(x(v);z)}
&=\frac{H^{(0)}_{|I|+1}(x(v),z;I)}{
  (x(z)+y(v)) H^{(0)}_1(x(v);z)}
\\
&+\frac{1}{(x(z)+y(v))}
\sum_{l=1}^{|I|}
\sum_{\substack{I_1\uplus ...\uplus I_l=I\\ I_1,...,I_l\neq \emptyset}}
\prod_{i=1}^l \frac{(-\lambda) W^{(0)}_{|I_i|+1}(v;I_i)}{
  (x(z)+y(v))}
\\
&
+
\sum_{l=1}^{|I|}
\sum_{\substack{I_0\uplus I_1\uplus ...\uplus I_l=I\\ I_0,I_1,...,I_l\neq \emptyset}}
\frac{H^{(0)}_{|I_0|+1}(x(v),z;I_0)}{
  (x(z)+y(v)) H^{(0)}_1(x(v);z)}
\prod_{i=1}^l \frac{(-\lambda) W^{(0)}_{|I_i|+1}(v;I_i)}{
  (x(z)+y(v))}\;.
\end{align*}
The last line is iteratively removed by
\begin{align*}
&\sum_{l=0}^{|I|-1}(-1)^l
\sum_{\substack{I_0\uplus I_1\uplus ...\uplus I_l=I\\ I_0,I_1,...,I_l\neq \emptyset}}
\frac{U^{(0)}_{|I_0|+1}(v,z;I_0)}{H^{(0)}_1(x(v);z)}
\prod_{i=1}^l  
\frac{H^{(0)}_{|I_i|+1}(x(v),z;I_i)}{
H^{(0)}_1(x(v);z)}
\\
&=\sum_{l=1}^{I}\frac{(-1)^{l-1}}{(x(z)+y(v))}
\sum_{\substack{I_1\uplus ...\uplus I_l=I\\ I_1,...,I_l\neq \emptyset}}
\prod_{i=1}^l  
\frac{H^{(0)}_{|I_i|+1}(x(v),z;I_i)}{
H^{(0)}_1(x(v);z)}
\tag{*}
\\
&
+\frac{1}{(x(z)+y(v))}
\sum_{l=1}^{|I|}
\sum_{\substack{I_1\uplus ...\uplus I_l=I\\ I_1,...,I_l\neq \emptyset}}
\prod_{i=1}^l \frac{(-\lambda) W^{(0)}_{|I_i|+1}(v;I_i)}{
  (x(z)+y(v))}\;.
\tag{**}
\end{align*}
Taking $\frac{1}{x(z)+y(v)}=\frac{U^{(0)}_{1}(v,z)}{H^{(0)}_1(x(v);z)}$
into account, the line (*) corresponds to the case $I_0=\emptyset$
of the lhs. The line (**) equals $D_I\frac{1}{x(z)+y(v)}$ so that the
assertion follows.
\end{proof}
\end{lemma}

The previous lemmas allow us to solve the genus $g=0$ case:
\begin{proposition}
  \label{prop:DlogHP}
  The solution of the system \textup{\eqref{DSE-UH} and \eqref{DSE-HP}} with  \textup{\eqref{def-H} and \eqref{def-P}} is
  \begin{align}
D_I \log H^{(0)}_1(x(v);z)
&=\sum_{k=1}^d 
D_I \log (x(v)+y(\hat{z}^k))
+F_{|I|+1}^{(0)}(x(v);x(z);I)\;,
\label{eq:DlogH}
  \\
D_I \log P^{(0)}_1(x(v),x(z))
&=\sum_{k=0}^d 
D_I \log (x(v)+y(\hat{z}^k))
+F_{|I|+1}^{(0)}(x(v);x(z);I)
  \label{eq:DlogP}
\end{align}
[note that the sum starts with $k=1$ in \eqref{eq:DlogH}
but $k=0$ in \eqref{eq:DlogP}]
where
\begin{align}
  F^{(0)}_{|I|+1}  (x(v);x(z);I)=   \sum_{j=1}^{|I|}
D_{I\setminus u_j}\frac{\lambda}{ (x(v)-x(u_j))(x(z)+y(u_j))}\;.
\label{eq:DlogP-poleu}
\end{align}
\begin{proof}
  Both sides of (\ref{DIlogPH}) can only be a function of $x(z)$ if
  (\ref{eq:DlogH}) and (\ref{eq:DlogP}) hold for some rational
  function $F_{|I|+1}^{(0)}(x(v);x(z);I)$.  Since
  $ \frac{H^{(0)}_{|I|+1}(x(v);z;I)}{H^{(0)}_{1}(x(v);z)}$ is
  holomorphic at $x(v)+y(z)=0$ by (\ref{def-H}), the function
  $F_{|I|+1}^{(0)}$ cannot have poles at $x(v)+y(\hat{z}^j)=0$. Recall
  from (\ref{def-P}) and (\ref{def-H}) that the only poles of
  $x(v)\mapsto
  \frac{P^{(0)}_{|I|+1}(x(v),x(z);I)}{P^{(0)}_{1}(x(v),x(z))}$ are at
  $x(v)=x(u_j)$ for $u_j\in I$ and at the $z$ with
  $P^{(0)}_{1}(x(v),x(z))=0$. The latter are already given by the
  first line of (\ref{eq:DlogP}) so that the only poles of
  $F^{(0)}_{|I|+1}$ are located at $x(v)=x(u_j)$ for every $u_j\in I$.
  These poles are simple and arise via (\ref{def-H}):
\begin{align*}
  \lim_{x(v)  \to x(u_j)} (x(v)-x(u_j))H^{(0)}_{|I|+1}(x(v);z;I)=
  \lambda U^{(0)}_{|I|}(u_j,z;I\setminus u_j)\;,
\end{align*}  
if $u_j\in I$. This gives for the corresponding limit of (\ref{def-DlogH})
\begin{align}
  &  \lim_{x(v)  \to x(u_j)} (x(v)-x(u_j))D_I \log H^{(0)}_1(x(v);z))
  \nonumber
  \\
  &=
\lambda   \sum_{l=0}^{|I|-1}
  \sum_{  \substack{I_0\uplus ...\uplus I_l=I\setminus u_j\\
      I_1,...I_l\neq \emptyset}}
      \frac{U^{(0)}_{|I_0|+1}(u_j,z;I_0)}{H^{(0)}_1(x(u_j);z)}
(-1)^{l}
\prod_{i=1}^l \frac{H^{(0)}_{|I_i|+1}(x(u_j);z;I_i)}{ H^{(0)}_1(x(u_j);z)}
\nonumber
\\
&= D_{I\setminus u_j} \frac{\lambda}{x(z)+y(u_j)}
  \label{lim-DIlogH-g0a}
\end{align}
where Lemma~\ref{lemma:DIUH} has been used. This finishes the proof.
\end{proof}
\end{proposition}

\begin{remark}\label{Rem:loopinsert}
  The naive action of the loop insertion operator $D_u$ on the
  solution \eqref{eq:DlogH} or \eqref{eq:DlogP} does \textbf{not} give
  the correct solution of the DSEs for
  $D_{I\cup u} \log H^{(0)}_1(x(v);z)$ or
  $D_{I\cup u} \log P^{(0)}_1(x(v);z)$. There would be one term
  missing in $F_{|I|+2}^{(0)}(x(v);x(z);I,u)$ which is of the form
    \begin{align}\label{missterm}
        D_{I}\frac{\lambda}{ (x(v)-x(u))(x(z)+y(u))}.
    \end{align}
    However without this term, $F_{|I|+2}^{(0)}(x(v);x(z);I,u)$ would
    not be symmetric after interchanging the variables $u$ with
    $u_i\in I$. Furthermore, the missing term \eqref{missterm}
    actually vanishes at $u\to E_k$ for any $k$, since
    $\lim_{u\to E_k} y(u)$ diverges. This observation is completely
    consistent with the discussion of Sec.\
    \ref{Sec.motivateloop}. Therefore, we want to highlight again that
    the loop insertion operator as defined in Sec.\ \ref{sec.loopdef}
    is rather a formal rewriting or abbreviation for involved
    expressions. The actual computation is nonetheless performed on
    the level of DSEs.
\end{remark}

\subsection{Linear and quadratic loop equations for
genus  $g=0$}

From (\ref{def-P}), (\ref{def-H}) and the
connecting equation (\ref{DSE-W})
we read off
\begin{align*}
P^{(0)}_1(x(v),x(z))&=x(v)+x(z)
  -\frac{\lambda}{N} \sum_{k=1}^d \frac{r_k}{x(z)-x(\varepsilon_k)}
  + \mathcal{O}((x(v))^{-1})\;,
  \\
H^{(0)}_1(x(v);z)&=1+\frac{1}{x(v)}\Big(x(z)-y(z)-
\frac{\lambda}{N} \sum_{k=1}^d \frac{r_k}{x(z)-x(\varepsilon_k)}\Big)
\\
&+ \mathcal{O}((x(v))^{-2}  )\;,
\\
P^{(0)}_{2}(x(v),x(z);u_1)
&=\frac{\lambda}{x(z)-x(u_1)}
+
\frac{\lambda}{x(v)}\Big\{ \frac{\lambda}{N} \sum_{k=1}^d \frac{r_k
  W^{(0)}_2(\varepsilon_k;u_1)}{x(z)-x(\varepsilon_k)}
\\  
&+
\frac{x(z)-y(u_1)
  -\frac{\lambda}{N} \sum_{k=1}^d \frac{r_k}{x(z)-x(\varepsilon_k)}
}{(x(z)-x(u_1))}
\Big\}  + \mathcal{O}((x(v))^{-2} ) 
\end{align*}
and then for $|I|\geq 2$
\begin{align*}  
P^{(0)}_{|I|+1}(x(v),x(z);I)
  &
  =\frac{\lambda}{x(v)}\Big\{
\frac{\lambda}{N} \sum_{k=1}^d \frac{r_k
 W^{(0)}_{|I|+1}(\varepsilon_k;I)}{x(z)-x(\varepsilon_k)}
-\sum_{j=1}^{|I|}\frac{\lambda W^{(0)}_{|I|}(u_j;I\setminus u_j)}{x(z)-x(u_j)}
\\
&+ \frac{\lambda \delta_{|I|,2}}{(x(z)-x(u_1))(x(z)-x(u_2))}
\Big\}+ \mathcal{O}((x(v))^{-2} ) \;.
\end{align*}  
These expansions combine for $I\neq \emptyset$ to
\begin{align}
&  D_I\log P^{(0)}_1(x(v),x(z))
  \\
  &=\frac{\lambda }{(x(v))^2}
  \Big\{
\frac{\lambda}{N} \sum_{k=1}^d \frac{r_k
W^{(0)}_{|I|+1}(\varepsilon_k;I)
}{
  x(z)-x(\varepsilon_k)}
-\lambda(1-\delta_{|I|,1}) \sum_{j=1}^{|I|} \frac{W^{(0)}_{|I|+1}(u_j;I\setminus u_j)
}{ x(z)-x(u_j)}
\Big\}
\nonumber
\\
&+\frac{\lambda \delta_{|I|,1}}{(x(z)-x(u_1))}
\Big(\frac{1}{x(v)}-\frac{y(u_1)}{(x(v))^2}\Big)
+ \mathcal{O}((x(v))^{-3})\;.
\nonumber
\end{align}
Comparison with (\ref{eq:DlogP}) and (\ref{eq:DlogP-poleu})  yields
the following main result:
\begin{proposition}\label{pro:sumpreimage}
The functions $W^{(0)}_{|I|+1}$ satisfy for $I\neq \emptyset$ the 
linear loop equations
\begin{align}
\sum_{k=0}^d W^{(0)}_{|I|+1}(\hat{z}^k;I)
&=
\frac{\delta_{|I|,1}}{(x(z)-x(u_1))}
-\sum_{j=1}^{|I|}
D_{I\setminus u_j}\Big(\frac{1}{x(z)+y(u_j)}\Big)
\label{lle-g0}
\end{align}
and the quadratic loop equations
\begin{align}
  & -\sum_{k=0}^d
  y(\hat{z}^k) W^{(0)}_{|I|+1}(\hat{z}^k;I)
  \label{qle-g0}
  \\
&=  \frac{\lambda}{2} 
  \sum_{  \substack{I_1\uplus I_2=I\\  I_1,I_2\neq \emptyset}} 
\sum_{k=0}^d  W^{(0)}_{|I_1|+1}(\hat{z}^k;I_1)  W^{(0)}_{|I_2|+1}(\hat{z}^k;I_2)
-\sum_{j=1}^{|I|} x(u_j)
D_{I\setminus u_j}\Big(\frac{1}{x(z)+y(u_j)}\Big)
\nonumber
\\
&
+
\frac{\lambda}{N} \sum_{k=1}^d \frac{r_k
W^{(0)}_{|I|+1}(\varepsilon_k;I)}{
x(z)-x(\varepsilon_k)}
-\lambda(1-\delta_{|I|,1}) \sum_{j=1}^{|I|} \frac{
W^{(0)}_{|I|+1}(u_j;I\setminus u_j)}{ x(z)-x(u_j)}
-\frac{\delta_{|I|,1} y(u_1)}{x(z)-x(u_1)}\;.
\nonumber
\end{align}  
\end{proposition}

{\begin{remark}\label{Rem:loopinsert2}
    Similar to Rem.\ \ref{Rem:loopinsert}, the naive action of the loop insertion operator $D_u$ on the equation \eqref{lle-g0} does \textbf{not} provide the desired symmetric result for $\sum_{k=0}^d W^{(0)}_{|I|+2}(\hat{z}^k;I,u)$. The missing term of the form 
    \begin{align}\label{sumpresym}
        -D_{I}\Big(\frac{1}{x(z)+y(u)}\Big)
    \end{align}
    is not generated by a naive action of the loop insertion operator. The term \eqref{sumpresym} vanishes at $u\to E_k$, thus this appearance is consistent with the discussion of Sec.\ \ref{Sec.motivateloop}. Again, the result of Prop.\ \ref{pro:sumpreimage} clearly distinguishes the quartic Kontsevich model from any other model solved by topological recursion, where actually the corresponding sum over all preimages as in \eqref{lle-g0} sums up to zero for $|I|>1$ for any compact spectral curve of genus zero. 
\end{remark}}


\subsection{Computation of $\hat{Q}^{(0)}$}

\label{sec:g0-Q}

Combining (\ref{V01}) with (\ref{def-tV}),
(\ref{DSE-tVM}) and (\ref{DSE-tMQ}) gives
\begin{align}
  \hat{Q}^{(0)}_1(x(v),x(z))&=
  -\frac{\lambda(x(v)+x(z)-2x(0)) \mathAE(x(v)) \mathAE(x(z))}{(x(v)-x(z))^2}
  \label{tQ01}
  \\
  &=  -\frac{\lambda(x(v)+x(z)-2x(0))}{2(x(v)-x(z))^2}
  \sqrt{\frac{P^{(0)}_1(x(v),x(v))P^{(0)}_1(x(z),x(z))}{(x(v)-x(0))
(x(z)-x(0))}}\;,
  \nonumber
\end{align}
where (\ref{U01-diag}), (\ref{DSE-UH}) and (\ref{DSE-HP}) have been used.

We take the logarithm of this equation and apply the loop
insertion operator. The na\"{\i}ve expectation is actually true if
a symbolic expression $D^0_Ix(0)$ is correctly identified:
\begin{proposition}
\label{prop:DIlogQ}
For $I\neq \emptyset$ one has
\begin{align}
&  D_I \log \hat{Q}^{(0)}_1(x(v),x(z))
\label{DIlogQ}
\\
&=  
\frac{1}{2} D_I \log P^{(0)}_1(x(v),x(v))
+\frac{1}{2} D_I \log P^{(0)}_1(x(z),x(z))
\nonumber
\\
&-\frac{1}{2}\sum_{l=1}^{|I|} \frac{(-1)^{l-1}}{l}
\Big(\frac{2^{l+1}}{(x(v)+x(z)-2x(0))^l}
-\frac{1}{(x(v)-x(0))^l}
-\frac{1}{(x(z)-x(0))^l}\Big)
\nonumber
\\
& \hspace*{3cm}
\times \sum_{\substack{I_1\uplus ...\uplus I_l=I\\
    I_1,..,I_l\neq \emptyset}}
\prod_{i=1}^l D^0_{I_i} x(0)\;,
\nonumber
\end{align}
where $D^0_{I_i} x(0)$ are symbolic expressions
uniquely determined by the condition
that $D_I \log \hat{Q}^{(0)}_1(x(v),x(z))$ is holomorphic at $v=0$.
\end{proposition}
We provide parts of the proof as separate results.
Let $v\in \{0,\pm \alpha_1,...,\pm \alpha_d\}$ be the set of zeros of
$x(v)+y(v)=0$.

\begin{lemma}\label{lemma:Ualpha}
$U^{(0)}_{|I|+1}(z,z;I)$ is holomorphic at
$z=\alpha_k$.
\begin{proof}
Set $v=z$ in (\ref{DIlogU01}) and insert Proposition (\ref{prop:DlogHP}):
\begin{align*}
  D_I \log U^{(0)}_1(z,z)
  &=\sum_{l=1}^d D_I \log(x(\hat{z}^l)+y(\hat{z}^l))
-D_I \log(x(z)+y(z))
\\
&+\sum_{j=1}^{|I|} D_{I\setminus u_j}
\frac{\lambda}{(x(z)-x(u_j))(x(z)+y(u_j))}\;.
\end{align*}  
By definition  we have $0=x(\alpha_k)+y(\alpha_k)=
x(\alpha_k)-x(-\alpha_k)$ where the key identity (\ref{glInvolution}) is used.
This means that $\widehat{\alpha_k}^1=-\alpha_k$ is one of the
preimages of $x(\alpha_k)$. Again with (\ref{glInvolution}) we conclude
$x(\hat{z}^1)+y(\hat{z}^1)=-(x(z)+y(z))$ for $z=\alpha_k$.
This means that 
$\log(-x(\hat{z}^1)-y(\hat{z}^1))
-\log(x(z)+y(z))$ is holomorphic in a neighbourhood of 
$z=\alpha_k$ and leads to a holomorphic 
$D_I\log(x(\hat{z}^1)+y(\hat{z}^1))
-D_I\log(x(z)+y(z))$.
\end{proof}
\end{lemma}

The zeros $\hat{Q}^{(0)}_1(x(v),x(\alpha_k))=0$ and
$P^{(0)}_1(x(\alpha_k),x(\alpha_k))=0$ produce
(higher) poles in 
$D_I \log \hat{Q}^{(0)}_1(x(v),x(z))$ and
$D_I\log P^{(0)}_1(x(z),x(z))$ at $x(z)=x(\alpha_k)$.
But these cancel exactly:
\begin{lemma}
  \label{lemma-QP-polea}
  $D_I \log \hat{Q}^{(0)}_1(x(v),x(z))
  -\frac{1}{2}D_I\log P^{(0)}_1(x(z),x(z))$
  is holomorphic  at $x(z)=x(\alpha_k)$.
\begin{proof}
Equations (\ref{DIlogPH}) and (\ref{DIlogU01}) for $v\mapsto z$ and
(\ref{DIlogU01}) and (\ref{DIlogM01}) combine to 
\begin{align*}
&D_I \log \hat{Q}^{(0)}_1(x(v),x(z))
-
\frac{1}{2} D_I \log P^{(0)}_1(x(z),x(z))
\\
&=
D_I \log\hat{V}^{(0)}_1(v,z)
-\frac{1}{2}  D_I \log U^{(0)}_1(z,z)
+D_I\log(x(v)+y(v))
\\
&\equiv\sum_{l=1}^{|I|} \frac{(-1)^{l-1}}{l}
\sum_{\substack{I_1\uplus ... \uplus I_l=I\\ I_1,...,I_l\neq \emptyset}}
\Big(\prod_{i=1}^l\frac{\hat{V}^{(0)}_{|I_i|+1}(v,z;I_l)}{\hat{V}^{(0)}_1(v,z)}
-\frac{1}{2}\prod_{i=1}^l\frac{U^{(0)}_{|I_i|+1}(z,z;I_l)}{U^{(0)}_1(z,z)}
\\[-2ex]
&\hspace*{10cm}
+\prod_{i=1}^l\frac{W^{(0)}_{|I_i|+1}(v;I_l)}{x(v)+y(v)}
\Big)\;.
\end{align*}
Here $\hat{V}^{(0)}_1(v,\alpha_k)\neq 0$ and
$U^{(0)}_1(\alpha_k,\alpha_k)\neq 0$ 
from the explicit formulae. Every function
$V^{(0)}_{|I|+1}(v,z;I)$ and thus
also $\hat{V}^{(0)}_{|I|+1}(v,z;I)$ is holomorphic at $z=\alpha_k$
from the matrix model construction. In fact, this holomorphicity is the
key assumption for the recursive
solution \cite[Prop.~E.4]{Branahl:2020yru}
which we reproduce by our algebraic method.
The assertion follows
with Lemma~\ref{lemma:Ualpha}.
\end{proof}

\end{lemma}
By (\ref{def-P}) and (\ref{def-Q}),
$Q^{(0)}_{|I|+1}(x(v),x(z);I)$ and
$P^{(0)}_{|I|+1}(x(v),x(z);I)$ have simple poles at
$x(z)=x(u_j)$ which give rise to simple poles of
$D_I \log \hat{Q}^{(0)}_1(x(v),x(z))$ 
and $\frac{1}{2}D_I\log P^{(0)}_1(x(v),x(v))$ at $x(v)=x(u_j)$.
But they cancel:
\begin{lemma}
  \label{lemma-QP-poleu}
  $D_I \log \hat{Q}^{(0)}_1(x(v),x(z))
  -\frac{1}{2}D_I\log P^{(0)}_1(x(v),x(v))$
  is holomorphic  at every $x(v)=x(u_j)$ with $u_j\in I$.
\begin{proof}
From Proposition~\ref{prop:DlogHP} at $z=v$ we get 
\[
\lim_{v\to u_j} (x(v)- x(u_j))   
D_I \log P^{(0)}_1(x(v),x(v))
= D_{I\setminus u_j} \frac{2\lambda}{x(u_j)+y(u_j)}\;.
\]
Indeed, this term with numerator $\lambda$ instead of $2\lambda$
arises from (\ref{eq:DlogP-poleu}) at $z=v$. A second copy arises
from a unique factor
$W^{(0)}_2(\hat{v}^k;u_j)$ in the first line of (\ref{eq:DlogP});
its residue is
\begin{align*}
  \Res\displaylimits_{x(v)\to x(u_j)} D_{I\setminus u_j}
  \sum_{k=0}^d \frac{\lambda W^{(0)}_2(\hat{v}^k;u_j) dx(v)}{x(v)+y(\hat{v}^k)}
&=  \Res\displaylimits_{x(v)\to x(u_j)} D_{I\setminus u_j}
  \sum_{k=0}^d \frac{\lambda W^{(0)}_2(\hat{v}^k;u_j) dx(v)}{x(v)+y(u_j)}
  \\
  &= D_{I\setminus u_j}  \frac{\lambda}{x(u_j)+y(u_j)}
\end{align*}
by (\ref{lle-g0}). Next, from (\ref{def-tQ}), (\ref{def-P}) and (\ref{def-Q}),
all at $v\leftrightarrow z$, we get
\begin{align*}
&\lim_{v\to u_j} (x(v)- x(u_j))   
\hat{Q}^{(0)}_{|I|+1}(x(z),x(v);I)
\\
&=\lambda M^{(0)}_{|I|}(x(z);u_j;I\setminus u_j)
- \frac{\lambda^2H^{(0)}_{|I|}(x(z);u_j;I\setminus u_j)}{(x(u_j)-x(z))^2}
\equiv \lambda\hat{M}^{(0)}_{|I|}(x(z);u_j;I\setminus u_j)\;.
\end{align*}
Therefore, the residue of $D_I\log \hat{Q}^{(0)}_{1}$ is 
\begin{align*}
&\lim_{v\to u_j} (x(v)- x(u_j))   
D_I\log \hat{Q}^{(0)}_{1}(x(v),x(z))
\\
&
=\sum_{l=0}^{|I|-1}(-1)^l \sum_{\substack{I_0\uplus I_1\uplus ... \uplus
    I_l=I\setminus u_j\\ I_1,...,I_l\neq \emptyset}}
\frac{\hat{M}^{(0)}_{|I_0|+1}(x(z);u_j;I_0)}{
  \hat{Q}^{(0)}_{1}(x(z),x(u_j))}
\prod_{i=1}^l \frac{\hat{Q}^{(0)}_{|I_i|+1}(x(z);u_j;I_i)}{
  \hat{Q}^{(0)}_{1}(x(z),x(u_j))}\;.
\end{align*}
In complete analogy to the proof of Lemma~\ref{lemma:DIUH} we have
\begin{align*}
&\lim_{v\to u_j} (x(v)- x(u_j))   
D_I\log \hat{Q}^{(0)}_{1}(x(v),x(z))
=
D_{I\setminus u_j}\frac{\lambda}{x(u_j)+y(u_j)}\;,
\end{align*}
which finishes the proof.
\end{proof}
\end{lemma}
\begin{proof}[Proof of Proposition~\ref{prop:DIlogQ}]
  Observe that (\ref{def-tQ}) implies
$\lim_{v\to z}
\frac{\hat{Q}^{(0)}_1(x(v),x(z);I)}{\hat{Q}^{(0)}_1(x(v),x(z))}
=\frac{P^{(0)}_1(x(z),x(z);I)}{P^{(0)}_1(x(z),x(z))}$
so that 
$D_I\log \hat{Q}^{(0)}_{1}(x(v),x(z))$ is holomorphic at $v=z$
for $I\neq \emptyset$. 
Together with Lemma~\ref{lemma-QP-polea} and
Lemma~\ref{lemma-QP-poleu}, the only remaining candidates for poles
of
$D_I\log \hat{Q}^{(0)}_{1}(x(v),x(z))
-\frac{1}{2}D_I\log P^{(0)}_{1}(x(v),x(v))
-\frac{1}{2}D_I\log P^{(0)}_{1}(x(z),x(z))$ are:
\begin{itemize}
\item The zeros
$x(v)=x(0)$ and $x(z)=x(0)$ of
$P^{(0)}_{1}(x(v),x(v))$ and 
$P^{(0)}_{1}(x(z),x(z))$, respectively, which produce higher poles
\[
 D_I\log P^{(0)}_{1}(x(v),x(v))
 =\sum_{l=1}^{|I|} \frac{f_l(I)}{(x(v)-x(0))^l}
 +\mathcal{O}((x(v)-x(0))^0)
\]
and similarly for $D_I\log P^{(0)}_{1}(x(z),x(z))$.
Neither
$\hat{Q}^{(0)}_{1}(x(v),x(z))$ has a zero at 
$x(v)=x(0)$ or $x(z)=x(0)$ nor
$\hat{Q}^{(0)}_{|I|+1}(x(v),x(z);I)$ has a pole there so that there
is no compensation. The coefficients
$f_l(I)$ only depend on the $u_j\in I$ but not on $v,z$; they are the same for
$D_I\log P^{(0)}_{1}(x(v),x(v))$ and
$D_I\log P^{(0)}_{1}(x(z),x(z))$ because of the
symmetry of
$D_I\log \hat{Q}^{(0)}_{1}(x(v),x(z))$ in $v\leftrightarrow z$.
We can trade the $f_l(I)$ for the symbolic expressions $D^0_{I'}x(0)$.

\item The zeros
$x(v)=2x(0)-x(z)$ of
$\hat{Q}^{(0)}_{1}(x(v),x(z))$, which produce higher poles
\begin{align*}
 D_I\log \hat{Q}^{(0)}_{1}(x(v),x(z))
 &=\sum_{l=1}^{|I|} \frac{\tilde{f}_l(I)}{(x(v)+x(z)-2x(0))^l}
 \\
 &+\mathcal{O}((x(v)+x(z)-2x(0))^0)\;.
\end{align*}
Neither
$P^{(0)}_{1}(x(v),x(v))$ has a zero at 
$x(v)=2x(0)-x(z)$ nor
$P^{(0)}_{|I|+1}(x(v),x(z);I)$ has a pole there so that there
is no compensation.
Viewed as function
$x(v)\mapsto  D_I\log \hat{Q}^{(0)}_{1}(x(v),x(z))$, the
coefficients $\tilde{f}_l(I)$ are independent of $x(v)$ and then
by symmetry independent of $x(z)$. The requirement 
$\lim_{z\to v}
\frac{\hat{Q}^{(0)}_1(x(v),x(z);I)}{\hat{Q}^{(0)}_1(x(v),x(z))}
=\frac{P^{(0)}_1(x(v),x(v);I)}{P^{(0)}_1(x(v),x(v))}$
then fixes the relative factor between
$f_l(I)$ and $\tilde{f}_l(I)$ to be 
$-2^{l+1}$ as given in (\ref{DIlogQ}).
\end{itemize}
This completes the proof.
\end{proof}

\begin{example}
We have
\[
  \lim_{v\to 0} (x(v)-x(0))
  \frac{P^{(0)}_1(x(v),x(v);u)}{P^{(0)}_1(x(v),x(v))}
  =\frac{\lambda}{2} W^{(0)}_2(0;u)
\]
from Proposition~\ref{prop:DlogHP}. This identifies $D^0_ux(0)=
-\frac{\lambda}{2} W^{(0)}_2(0;u)$.
\end{example}

\section{Solution for $g=1$}

\label{sec:g1}

\subsection{The case $I=\emptyset$}

\label{sec:g1-0}

We consider the relation between the `genus insertions' into
$\log P^{(0)}_1(x(v),x(z))$ and $\log H^{(0)}_1(x(v);z)$. Equation
(\ref{DSE-hHP}), divided by $P^{(0)}_1(x(v),x(z))$, reads
\begin{align}
&\frac{\hat{P}^{(1)}_1(x(v),x(z))}{P^{(0)}_1(x(v),x(z))}
-\frac{\hat{H}^{(1)}_1(x(v);z)}{H^{(0)}_1(x(v);z)}
\label{DghPhH}
\\
&=\frac{\lambda}{x(v)+y(z)}\Big(
\frac{\partial D_w \log H^{(0)}_1(x(v);z)}{\partial x(w)} 
\Big|_{w=z}
+W^{(1)}_1(z)+\frac{\lambda}{(x(v)-x(z))^3}\Big)
\nonumber
\\
&=
\frac{\lambda W^{(1)}_1(z)}{x(v)+y(z)}
+\frac{\lambda^2}{(x(v)-x(z))^3(x(v)+y(z))}
\nonumber
\\
&+\sum_{k=1}^d
\frac{\lambda^2 \Omega^{(0)}_2(z,\hat{z}^k)}{2(x(v)+y(z))
  (x(v)+y(\hat{z}^k))}
+\sum_{j=1}^d\frac{\lambda^2 \Omega^{(0)}_2(\hat{z}^j,z)}{2(x(v)+y(\hat{z}^j))
  (x(v)+y(z))}
\nonumber
\\
&+\frac{\lambda^2}{(x(v)+y(z))}
\frac{\partial}{\partial x(w)}
\frac{1}{(x(v)-x(w))(x(z)+y(w))}\Big|_{w=z}\;,
\nonumber
\end{align}
where
$\frac{\partial D_w \log H^{(0)}_1(x(v);\hat{z}^k)}{\partial x(w)} 
\big|_{w=z}$ was provided by Proposition \ref{prop:DlogHP}.
The differentiation leads to $\Omega^{(0)}_2$ (see (\ref{eq:OmegaW}))
which we have symmetrised in
both arguments. We understand 
$\frac{\partial y(w)}{\partial x(w)}\equiv \frac{y'(w)}{x'(w)}$.

In order for 
$\frac{\hat{P}^{(1)}_1(x(v),x(z))}{P^{(0)}_1(x(v),x(z))}$ to be a
rational function of $x(z)$, we need that
\begin{align}
  \frac{\hat{P}^{(1)}_1(x(v),x(z))}{P^{(0)}_1(x(v),x(z))}
&=
\sum_{k=0}^d \frac{\lambda W^{(1)}_1(\hat{z}^k)}{x(v)+y(\hat{z}^k)}
+\frac{\lambda^2}{2} \sum_{\substack{k,j=0 \\ k\neq j}}^d
\frac{\Omega^{(0)}_2(\hat{z}^j,\hat{z}^k)}{(x(v)+y(\hat{z}^j))
  (x(v)+y(\hat{z}^k))}
\nonumber
\\
&+\sum_{k=0}^d \frac{\lambda^2}{(x(v)+y(\hat{z}^k))}
\frac{\partial}{\partial x(w)}
\frac{1}{(x(v)-x(w))(x(z)+y(w))}\Big|_{w=\hat{z}^k}
\nonumber
\\
&
+\sum_{k=0}^d \frac{\lambda^2}{(x(v)-x(z))^3(x(v)+y(\hat{z}^k))}
+ F^{(1)}_1(x(v);x(z))
\label{DghP}
\end{align}
for some rational function $F^{(1)}_1$,
and that almost the same formula holds for
$\frac{\hat{H}^{(1)}_1(x(v);z)}{H^{(0)}_1(x(v);z)}$,
only with preimage sum $\sum_{k=1}^d$ instead of $\sum_{k=0}^d$.
By (\ref{def-hH}), (\ref{def-H}) and (\ref{def-M}), the function
$\frac{\hat{H}^{(1)}_1(x(v);z)}{H^{(0)}_1(x(v);z)}$ is holomorphic
at $x(v)+y(z)=0$ so that $F^{(1)}_1$ must be holomorphic at 
$x(v)+y(\hat{z}^k)=0$. It follows from  (\ref{def-hP}), (\ref{def-P})
and (\ref{def-Q}) that
the only other poles of
$x(v)\mapsto \frac{\hat{P}^{(1)}_1(x(v),x(z))}{P^{(0)}_1(x(v),x(z))}$ are 
at $x(v)=x(z)$ of order at most 2; more precisely
\begin{align}
\frac{\hat{P}^{(1)}_1(x(v),x(z))}{P^{(0)}_1(x(v),x(z))}
&= 
\frac{\lambda}{(x(v)-x(z))^2}
\frac{Q^{(0)}_{1}(x(z),x(z))}{P^{(0)}_{1}(x(z),x(z))}
\\
&+
\frac{\lambda}{(x(v)-x(z))}
\frac{\partial}{\partial x(w)}
\frac{Q^{(0)}_{1}(x(w),x(z))}{P^{(0)}_{1}(x(w),x(z))}\Big|_{w=z}
+\text{regular},
\nonumber
\end{align}
where `regular' means $\mathcal{O}((x(v)-x(z))^0)$. To achieve this
we necessarily need
\begin{align*}
&F^{(1)}_1 (x(v);x(z))=\sum_{a=1}^3 \frac{F^{(1)a}_1(x(z))}{(x(v)-x(z))^a}\;,
\\
&F^{(1)3}_1(x(z)) = 
-\sum_{k=0}^d \frac{\lambda^2}{x(z)+y(\hat{z}^k)}\;,\qquad
F^{(1)2}_1(x(z)) = 
\frac{\lambda Q^{(0)}_{1}(x(z),x(z))}{P^{(0)}_{1}(x(z),x(z))}\;,
\\
&F^{(1)1}_1(x(z)) = \frac{\partial}{\partial x(w)}
\frac{\lambda Q^{(0)}_{1}(x(w),x(z))}{P^{(0)}_{1}(x(w),x(z))}\Big|_{w=z}
-\frac{1}{2} \frac{\partial}{\partial x(w)}
\sum_{k=0}^d \frac{\lambda^2}{(x(z)+y(w))^2}\Big|_{w=\hat{z}^k}.
\end{align*}
From (\ref{tQ01}) and (\ref{def-tQ}) we obtain a representation
\begin{align*}
&\frac{\lambda Q^{(0)}_1(x(w),x(z))}{P^{(0)}_1(x(w),x(z))}
=\frac{\lambda^2}{(x(w)-x(z))^2}\bigg(1
-
\sqrt{1+\frac{(x(z)-x(w))^2}{4(x(w)-x(0))(x(z)-x(0))}}\\
&
\hspace*{4cm} \times \exp\Big(\frac{1}{2}\log 
\frac{P^{(0)}_1(x(w),x(w))}{P^{(0)}_1(x(z),x(z))}
-\log \frac{P^{(0)}_1(x(w),x(z))}{P^{(0)}_1(x(z),x(z))}
\Big)\bigg)
\\
&=-\frac{\lambda^2}{8(x(z)-x(0))^2} - \frac{\lambda^2}{2}
\frac{\partial^2 \log P^{(0)}_1(x(w),x(z))}{\partial x(w)\partial x(z)}
\Big|_{w=z}
\\
&+(x(w)-x(z))\Big(
\frac{\lambda^2}{8(x(z)-x(0))^3} - \frac{\lambda^2}{2}
\frac{\partial^3 \log P^{(0)}_1(x(w),x(z))}{\partial (x(w))^2\partial x(z)}
\Big|_{w=z}\Big)\\
&+\mathcal{O}((x(w)-x(z))^2)\;.
\end{align*}
Inserting the resulting Taylor expansion into
$F^{(1)1}_1(x(z))$ and $F^{(1)2}_1(x(z))$ leads to:
\begin{proposition}
\label{prop:hP11}
\begin{align}
&\frac{\hat{P}^{(1)}_1(x(v),x(z))}{P^{(0)}_1(x(v),x(z))}
\label{hP11P01}
\\
&=\sum_{k=0}^d \frac{\lambda W^{(1)}_1(\hat{z}^k)}{x(v)+y(\hat{z}^k)}
+\frac{\lambda^2}{2}\sum_{\substack{k,j=0 \\ k\neq j}}^d
\frac{  \Omega^{(0)}_2(\hat{z}^j,\hat{z}^k)}{
  (x(v)+y(\hat{z}^j))(x(v)+y(\hat{z}^k))}
\nonumber
\\
&-\frac{\lambda^2}{(x(v)-x(z))^3}
\Big( \frac{\partial \log P^{(0)}_1(x(v),x(w))}{\partial x(w)}
- \frac{\partial \log P^{(0)}_1(x(z),x(w))}{\partial x(w)}
\Big)\Big|_{w=z}
\nonumber
\\
&+\frac{\lambda^2}{2(x(v)-x(z))^2}
\frac{\partial^2 \log P^{(0)}_1(x(z),x(w))}{\partial x(z) \partial x(w)}
\Big|_{w=z}
\nonumber
\\
&-\frac{\lambda^2}{8(x(v)-x(z))^2(x(z)-x(0))^2}
 +\frac{\lambda^2}{8(x(v)-x(z))(x(z)-x(0))^3}\;.
 \nonumber
\end{align}
\end{proposition}

\subsection{The case $I\neq \emptyset$}

\label{sec:g1-I}

Our goal is to determine
\begin{align*}
&D_I\frac{\hat{P}^{(1)}_1(x(v),x(z))}{P^{(0)}_1(x(v),x(z))}
\\
&=\sum_{l=0}^{|I|}(-1)^l \sum_{\substack{I_0\uplus I_1 \uplus ... \uplus I_l=I
  \\ I_1,...,I_l \neq \emptyset}}
\frac{\hat{P}^{(1)}_{|I_0|+1}(x(v),x(z);I_0)}{P^{(0)}_1(x(v),x(z))}
\prod_{i=1}^l \frac{P^{(0)}_{|I_i|+1}(x(v),x(z);I_i)}{P^{(0)}_1(x(v),x(z))}
\end{align*}
in parallel with $D_I\frac{\hat{H}^{(1)}_1(x(v);z)}{H^{(0)}_1(x(v);z)}$.
The Dyson-Schwinger equations (\ref{DSE-hHP})
for $g\leq 1$ imply:
\begin{lemma}
\begin{align}
&D_I\frac{\hat{P}^{(1)}_1(x(v),x(z))}{P^{(0)}_1(x(v),x(z))}
-D_I\frac{\hat{H}^{(1)}_1(x(v);z)}{H^{(0)}_1(x(v);z)}
\label{DIP-DIH-g1}
\\
&= 
\sum_{I_1\uplus I_2=I}
D_{I_1} \Big(\frac{\lambda}{x(v)+y(z)}\Big)
\Big\{
\frac{\partial D_{I_2\cup w} \log H^{(0)}_1(x(v);z)}{\partial x(w)}
\Big|_{w=z}
\nonumber
\\
&\hspace*{6cm}
+W^{(1)}_{|I_2|+1}(z;I_2)+\frac{\lambda \delta_{|I_2|,\emptyset}}{(x(v)-x(z))^3}
\Big\}\;.
\nonumber
\end{align}
\begin{proof}
By induction in $|I|$. The case $|I|=0$ is
(\ref{DSE-hHP}) for $g=1$. Otherwise
\begin{align*}
  &  D_I\frac{\hat{P}^{(1)}_1(x(v),x(z))}{P^{(0)}_1(x(v),x(z))}
  - D_I\frac{\hat{H}^{(1)}_1(x(v);z)}{H^{(0)}_1(x(v);z)}
\\
&= \frac{\hat{P}^{(1)}_{|I|+1}(x(v),x(z);I)}{P^{(0)}_1(x(v),x(z))}
-\sum_{\substack{I_1\uplus I_2=I \\ I_2\neq \emptyset}}
  D_{I_1}\Big(\frac{\hat{P}^{(1)}_1(x(v),x(z))}{P^{(0)}_1(x(v),x(z))}\Big)
\cdot  \frac{P^{(0)}_{|I_2|+1}(x(v),x(z);I_2)}{P^{(0)}_1(x(v),x(z))}
\\
&-\frac{\hat{H}^{(1)}_{|I|+1}(x(v);z;I)}{H^{(0)}_1(x(v);z)}  
+\sum_{\substack{I_1\uplus I_2=I \\ I_2\neq \emptyset}}
  D_{I_1}\Big(\frac{\hat{H}^{(1)}_1(x(v);z)}{H^{(0)}_1(x(v);z)}\Big)\cdot
  \frac{H^{(0)}_{|I_2|+1}(x(v);z;I_2)}{H^{(0)}_1(x(v);z)}  
\\
&=\sum_{\substack{I_1\uplus I_2=I\\ I_1 \neq \emptyset}}
\frac{\lambda W^{(0)}_{|I_1|+1}(z;I_1)}{x(v)+y(z)}
\frac{\hat{H}^{(1)}_{|I_2|+1}(x(v);z;I_2)}{H^{(0)}_1(x(v);z)}
\tag{*}
\\
&+\sum_{I_1\uplus I_2=I}
\frac{\lambda (W^{(1)}_{|I_1|+1}(z;I_1)+\frac{\lambda\delta_{|I_1|,0}}{(x(v)-x(z))^3})}{x(v)+y(z)}
\frac{H^{(0)}_{|I_2|+1}(x(v);z;I_2)}{H^{(0)}_1(x(v);z)}  
\tag{\ddag}
\\
&+\frac{\lambda}{x(v)+y(z)}
\frac{\partial}{\partial x(w)}  \frac{H^{(0)}_{|I|+2}(x(v);z;I\cup w)}{
  H^{(0)}_1(x(v);z)}  \Big|_{w=z}
\tag{\S}
\\
&-\sum_{\substack{I_1\uplus I_2 \uplus I_3=I \\ I_3\neq \emptyset}}
\Big\{
\frac{\partial
D_{I_2\cup w} \log H^{(0)}_1(x(v);z)}{\partial x(w)} \Big|_{w=z}
+W^{(1)}_{|I_1|+1}(z;I_1)+\frac{\lambda\delta_{|I_1|,0}}{(x(v)-x(z))^3}
\Big\}\nonumber
\\[-1ex]
&
\qquad\times D_{I_2} \Big(\frac{\lambda}{x(v)+y(z)}\Big)\cdot 
  \frac{P^{(0)}_{|I_3|+1}(x(v),x(z);I_3)}{P^{(0)}_1(x(v),x(z))}
  \tag{\dag}
  \\
  &-\sum_{\substack{I_1\uplus I_2\uplus I_3=I \\ I_2\neq \emptyset}}
  D_{I_1}\Big(\frac{\hat{H}^{(1)}_1(x(v);z)}{H^{(0)}_1(x(v);z)}\Big)
\cdot \frac{\lambda W^{(0)}_{|I_2|+1}(z;I_2)}{x(v)+y(z)}
\frac{H^{(0)}_{|I_3|+1}(x(v);z;I_3)}{H^{(0)}_1(x(v);z)}\;,
\tag{**}
\end{align*}
where the induction hypothesis was used.
The lines (*) and (**) cancel when distinguishing $I_3=\emptyset$ and
$I_3\neq \emptyset$.
Take (\ref{DSE-HPquot}), multiplied by
$\frac{\lambda}{x(v)+y(z)}$:
\begin{align*}
&\frac{\lambda P^{(0)}_{|\tilde{I}|+1}(x(v),x(z);\tilde{I})}{
      (x(v)+y(z))P^{(0)}_1(x(v),x(z))}
\\
&= 
\frac{\lambda H^{(0)}_{|\tilde{I}|+1}(x(v);z;\tilde{I})}{
  (x(v)+y(z))H^{(0)}_1(x(v);z)}
+\sum_{\substack{\tilde{I}'\uplus \tilde{I}''=\tilde{I}\\ \tilde{I}'\neq \emptyset}}  
\frac{\lambda W^{(0)}_{|\tilde{I}'|+1}(z;\tilde{I}')}{x(v)+y(z)}
\frac{\lambda H^{(0)}_{|\tilde{I}''|+1}(x(v);z;\tilde{I}'')}{
  (x(v)+y(z))H^{(0)}_1(x(v);z)}\;.
\end{align*}  
For $\tilde{I}''\neq \emptyset$, take this
equation for $\tilde{I}\mapsto\tilde{I}''$,
multiply by
$\frac{(-\lambda) W^{(0)}_{|\tilde{I}'|+1}(z;\tilde{I}')}{x(v)+y(z)}$
and sum over
$\tilde{I}'\uplus \tilde{I}''=\tilde{I}$ with $\tilde{I}'\neq \emptyset$. 
Repeat until all products of 
$\frac{(-\lambda) W^{(0)}_{|I_i|+1}(z;I_i)}{x(v)+y(z)}$
with $\frac{\lambda H^{(0)}_{|I_0|+1}(x(v);z;I_0)}{
  (x(v)+y(z))H^{(0)}_1(x(v);z)}$ are removed. The result is
\begin{align}
&\sum_{\substack{I_2\uplus I_3=I'\\ I_3\neq \emptyset}}  
D_{I_2} \Big(\frac{\lambda}{x(v)+y(z)}\Big)
\frac{P^{(0)}_{|I_3|+1}(x(v),x(z);I_3)}{P^{(0)}_1(x(v),x(z))}
\label{DfracPP}
\\
&= \frac{\lambda}{x(v)+y(z)}
\frac{H^{(0)}_{|I'|+1}(x(v);z;I')}{H^{(0)}_1(x(v);z)}
+\sum_{\substack{I_2\uplus I_3=I'\\ I_3\neq \emptyset}}  
\frac{\lambda W^{(0)}_{|I_3|+1}(z;I_3)}{x(v)+y(z)}
D_{I_2} \frac{\lambda}{x(v)+y(z)}\;,
\nonumber
\end{align}
which we use in $(\dag)$. Multiplied with
$(W^{(1)}_{|I_1|+1}(z;I_1)+\frac{\lambda\delta_{|I_1|,0}}{(x(v)-x(z))^3})$,
this cancels all terms with $I_2\neq \emptyset$ in $(\ddag)$ and
completes the case $I_2=\emptyset$ in $(\ddag)$ to the last line of
the assertion (\ref{DIP-DIH-g1}). Similarly,
expanding
\begin{align*}
&D_{I\cup w} \log H^{(0)}_1(x(v);z)
\\
&=\sum_{l=0}^{|I|}(-1)^l
\sum_{\substack{I_0\uplus I_1\uplus ... \uplus I_l=I
    \\ I_1,...,I_l\neq \emptyset}}
\frac{H^{(0)}_{|I_0|+2}(x(v);z;I_0\cup w)}{H^{(0)}_1(x(v);z)}
\prod_{i=1}^l  
\frac{H^{(0)}_{|I_i|+1}(x(v);z;I_i)}{H^{(0)}_1(x(v);z)}
\end{align*}
we get for the product with the first term on the rhs of (\ref{DfracPP})
\begin{align*}
&\sum_{\substack{I_2\cup I'=I \\ I'\neq \emptyset}}
D_{I\cup w} \log H^{(0)}_1(x(v);z)
\frac{\lambda}{x(v)+y(z)}
\frac{H^{(0)}_{|I'|+1}(x(v);z;I')}{H^{(0)}_1(x(v);z)}
\\
&= \frac{\lambda}{x(v)+y(z)}
\Big(
\frac{H^{(0)}_{|I|+2}(x(v);z;I\cup w)}{H^{(0)}_1(x(v);z)}
-D_{I\cup w} \log H^{(0)}_1(x(v);z)
\Big).
\end{align*}
The first term on the rhs cancels the line (\S), and the second
term completes the final term in (\ref{DfracPP})
to the second line of the assertion
(\ref{DIP-DIH-g1}).
\end{proof}
\end{lemma}

We conclude from Proposition~\ref{prop:DlogHP}:
\begin{align}
&\frac{\partial}{\partial x(w)}
D_{I\cup w} \log H^{(0)}_1(x(v);z)\Big|_{w=z}
\label{DIP-DIH-g1a}
\\
&=\sum_{k=1}^d 
\sum_{l=0}^{|I|} (-1)^l
  \sum_{  \substack{I_0\uplus I_1\uplus ...\uplus I_l=I\\ 
      I_1, ... I_l\neq \emptyset}}
\frac{\lambda D_{I_0} \Omega^{(0)}_2(\hat{z}^k,z)}{x(v)+y(\hat{z}^k)}
  \prod_{j=1}^l \frac{\lambda W^{(0)}_{|I_j|+1}(\hat{z}^k;I_j)}{x(v)+y(\hat{z}^k)}
  \nonumber
  \\
  &
- \sum_{j=1}^{|I|}
D_{I\setminus u_j}\frac{\lambda^2 \Omega^{(0)}_2(z,u_j)}{
  (x(v)-x(u_j))(x(z)+y(u_j))^2}
+
D_I\frac{\lambda}{ (x(v)-x(z))^2(x(z)+y(z))}
  \nonumber
\\
&+\frac{\partial}{\partial x(w)}
D_I\frac{\lambda}{ (x(v)-x(z))(x(z)+y(w))}\Big|_{w=z}\;.
  \nonumber
\end{align}
We insert (\ref{DIP-DIH-g1a}) into (\ref{DIP-DIH-g1}) and get with
the derivation property of $D_I$: 
\begin{align*}
&D_I\frac{\hat{P}^{(1)}_1(x(v),x(z))}{P^{(0)}_1(x(v),x(z))}
-D_I\frac{\hat{H}^{(1)}_1(x(v);z)}{H^{(0)}_1(x(v);z)}
\\
&=
D_I \frac{\lambda W^{(1)}_1(z)}{x(v)+y(z)}
+\sum_{k=1}^d D_I
\frac{\lambda^2 \Omega^{(0)}_2(\hat{z}^k,z)}{(x(v)+y(\hat{z}^k))
(x(v)+y(z))}
\\
&- \sum_{j=1}^{|I|}\frac{\lambda^3}{(x(v)-x(u_j))}
D_{I\setminus u_j}\frac{\Omega^{(0)}_2(z,u_j)}{
  (x(v)+y(z))  (x(z)+y(u_j))^2}
\\
&+ \frac{\lambda^2}{(x(v)-x(z))^3}
D_I \frac{1}{(x(z)+y(z))}
\\
&+ \frac{\lambda^2}{(x(v)-x(z))}
\frac{\partial}{\partial x(w)} D_I \frac{1}{(x(v)+y(z))(x(z)+y(w))}
\Big|_{w=z}\;.
\end{align*}
As before, in order for
$D_I\frac{\hat{P}^{(1)}_1(x(v),x(z))}{P^{(0)}_1(x(v),x(z))}$
to be a rational function of $x(z)$, we need 
\begin{subequations}
  \begin{align}
D_I\frac{\hat{P}^{(1)}_1(x(v),x(z))}{P^{(0)}_1(x(v),x(z))}
  &= K^{(1)}_{0,|I|+1}(x(v);z;I) + F^{(1)}_{|I|+1}(x(v);x(z);I)\;,
\label{DIP-g1b}
\\
D_I\frac{\hat{H}^{(1)}_1(x(v);z)}{H^{(0)}_1(x(v);z)}
&=K^{(1)}_{1,|I|+1}(x(v);z;I) + F^{(1)}_{|I|+1}(x(v);x(z);I)\;,
\label{DIH-g1b}
\end{align}
\end{subequations}
where  (note the difference in the lower subscript $A$
between  (\ref{DIP-g1b}) and (\ref{DIH-g1b}))
\begin{align}
&  K^{(1)}_{A,|I|+1}(x(v);z;I) 
\label{K1AI}
\\
&:=
\sum_{k=A}^d D_I \frac{\lambda W^{(1)}_1(\hat{z}^k)}{x(v)+y(\hat{z}^k)}
+\frac{\lambda^2}{2} \sum_{\substack{j,k=A\\ j\neq k}}^d D_I
\frac{\Omega^{(0)}_2(\hat{z}^j,\hat{z}^k)}{(x(v)+y(\hat{z}^j))
(x(v)+y(\hat{z}^k))}
\nonumber
\\
&- \sum_{j=1}^{|I|}\frac{\lambda^3}{(x(v)-x(u_j))}
\sum_{k=A}^d D_{I\setminus u_j}\frac{\Omega^{(0)}_2(\hat{z}^k,u_j)}{
  (x(v)+y(\hat{z}^k))  (x(z)+y(u_j))^2}
\nonumber
\\
&+ \frac{\lambda^2}{(x(v)-x(z))^3}
D_I \sum_{k=A}^d \frac{1}{(x(z)+y(\hat{z}^k))}
\nonumber
\\
&+ \frac{\lambda^2}{(x(v)-x(z))}\sum_{k=A}^d 
\frac{\partial}{\partial x(w)} D_I \frac{1}{(x(v)+y(\hat{z}^k))
  (x(z)+y(w))}
\Big|_{w=\hat{z}^k}
\nonumber
\end{align}
and $F^{(1)}_{|I|+1}(x(v);x(z);I)$ is some rational function
(the same in (\ref{DIP-g1b}) and (\ref{DIH-g1b}))
in both $x(v)$ and $x(z)$. In fact, also 
$K^{(1)}_{0,|I|+1}(x(v);z;I)$ is rational in both $x(v)$ and $x(z)$.
Since $D_I\frac{\hat{H}^{(1)}_1(x(v),x(z))}{
  H^{(0)}_1(x(v),x(z))}$ is holomorphic at $x(v)+y(z)=0$, the rational function
$x(v)\mapsto F^{(1)}_{|I|+1}$ can only have poles at $x(v)=x(z)$ and
at $x(v)=x(u_j)$ with $u_j\in I$.
With the behavior resulting from (\ref{def-H}) and (\ref{def-hH}),
\begin{align}
  \frac{\hat{H}^{(g)}_{|I|+1}(x(v);z;I)}{H^{(0)}_1(x(v);z)}
  &= \sum_{j=1}^{|I|}
  \frac{ \lambda}{(x(v)-x(u_j))}
  \frac{\hat{U}^{(g)}(z,u_j;I\setminus u_j)}{H^{(0)}_1(x(u_j);z)}
  +\mathcal{O}((x(v)-x(u_j))^0)
  \label{hH-pole-u}
\end{align}
and (\ref{DIP-DIH-g1a}) we see that
$F^{(1)}_{|I|+1}$ has first-order poles at $x(v)=x(u_j)$.
We determine them in 
Proposition~\ref{prop:hatFj} below. 

Near $x(v)=x(z)$ we have in (\ref{K1AI})
\begin{align}
  &K^{(1)}_{0,|I|+1}(x(v);z;I)
\label{K1AI-vz}  
\\
&=\frac{\lambda^2}{(x(v)-x(z))}\sum_{k=0}^d 
\frac{\partial}{\partial x(w)} D_I \frac{1}{(x(z)+y(\hat{z}^k))
  (x(z)+y(w))}
\Big|_{w=\hat{z}^k}
\nonumber
\\
&+ \frac{\lambda^2}{(x(v)-x(z))^3}
D_I \sum_{k=0}^d \frac{1}{(x(z)+y(\hat{z}^k))}
+\mathcal{O}((x(v)-x(z))^0)
\nonumber
\\
&=-\frac{\lambda^2}{(x(v)-x(z))}
\frac{1}{2} \frac{\partial^3}{\partial (x(z))^2
\partial x(w)}  \sum_{k=0}^d 
  D_I \log (x(z)+y(\hat{w}^k))
\Big|_{w=z}
\nonumber
\\
&+ \frac{\lambda^2}{(x(v)-x(z))^3}
D_I \sum_{k=0}^d \frac{1}{(x(z)+y(\hat{z}^k))}
+\mathcal{O}((x(v)-x(z))^0)
\nonumber
\\
&=-\frac{\lambda^2}{(x(v)-x(z))}
\Big(\frac{1}{2} \frac{\partial^3
(D_I \log P^{(0)}_1(x(z),x(w)))
}{\partial (x(z))^2
\partial x(w)}  
\Big|_{w=z}
\nonumber
\\
&\qquad
+\sum_{j=1}^{|I |} D_{I\setminus u_j}
\frac{\lambda}{(x(z)-x(u_j))^3(x(z)+y(u_j))^2}
\Big)
\nonumber
\\
&+ \frac{\lambda^2}{(x(v)-x(z))^3}
D_I \sum_{k=0}^d \frac{1}{(x(z)+y(\hat{z}^k))}
+\mathcal{O}((x(v)-x(z))^0)\;.
\nonumber
\end{align}
In the last step we have used Proposition~\ref{prop:DlogHP}.
On the other hand, from (\ref{def-hP}) we get
\begin{align}
D_I\frac{\hat{P}^{(1)}_{1}(x(v),x(z))}{P^{(0)}_1(x(v),x(z))}
&= \frac{\lambda}{(x(v)-x(z))^2}
D_I \frac{Q^{(0)}_{1}(x(z),x(z))}{P^{(0)}_{1}(x(z),x(z))}
\label{DIP-vz}
\\
&+\frac{\lambda}{(x(v)-x(z))}
\frac{\partial}{\partial x(w)}\Big(
D_I \frac{Q^{(0)}_{1}(x(w),x(z))}{P^{(0)}_{1}(x(w),x(z))}
\Big)_{w=z}
\nonumber
\\
& + \mathcal{O}((x(v)-x(z))^0)\;, \qquad \text{where}
\nonumber
\\
D_I \frac{Q^{(0)}_{1}(x(z),x(z))}{P^{(0)}_{1}(x(z),x(z))}
\nonumber
&
\\
& \hspace*{-3.5cm}
:= \sum_{l=0}^{|I|}(-1)^l \sum_{\substack{I_0\uplus I_1 \uplus ... \uplus I_l=I
  \\ I_1,...,I_l \neq \emptyset}}
\frac{Q^{(0)}_{|I_0|+1}(x(w),x(z);I_0)}{P^{(0)}_1(x(w),x(z))}
\prod_{i=1}^l \frac{P^{(0)}_{|I_i|+1}(x(w),x(z);I_i)}{P^{(0)}_1(x(w),x(z))}\;.
\nonumber
\end{align}
These properties lead to an expansion
\begin{align}
  F^{(1)}_{|I|+1}(x(v);x(z);I)
  = \sum_{a=1}^3 \frac{F^{(1)a}_{|I|+1}(x(z);I)}{((x(v)-x(z))^a}
  +\sum_{j=1}^{|I|} \frac{\hat{F}^{(1)j}_{|I|}(x(z);I\setminus u_j)}{x(v)-x(u_j)}
  \label{F1I}
\end{align}
where (\ref{K1AI-vz}) and (\ref{DIP-vz}) combine to
\begin{subequations}
\label{F1aI}
  \begin{align}
F^{(1)3}_{|I|+1}(x(z);I)&=-\sum_{k=0}^d D_I \frac{\lambda^2}{x(z)+y(\hat{z}^k)}
\\
F^{(1)2}_{|I|+1}(x(z);I)&=
\lambda
D_I \frac{Q^{(0)}_{1}(x(z),x(z))}{P^{(0)}_{1}(x(z),x(z))}
\\
F^{(1)1}_{|I|+1}(x(z);I)&=
\lambda
\frac{\partial\Big(
D_I \dfrac{Q^{(0)}_{1}(x(z),x(w))}{P^{(0)}_{1}(x(z),x(w))}
\Big)}{\partial x(w)}+\frac{\lambda^2}{2}
\frac{\partial^3
(D_I \log P^{(0)}_1(x(z),x(w)))
}{\partial (x(z))^2
\partial x(w)}  
\nonumber
\\
&\qquad
+\sum_{j=1}^{|I |} D_{I\setminus u_j} \frac{\lambda^3}{
  (x(z)-x(u_j))^3(x(z)+y(u_j))^2}
\Big|_{w=z}\;.
\end{align}
\end{subequations}

It remains to determine the functions
$\hat{F}^{(1)j}_{|I|}(x(z);I\setminus u_j)$. We will need two lemmas:
\begin{lemma}
  \label{lemma:hU1}
\begin{align*}
&
\hat{U}^{(1)}_{|I|+1}(v,z;I)
\\
&= \sum_{I_1\uplus I_2=I}
D_{I_1}\frac{1}{(x(z)+y(v))}\cdot 
\Big\{\hat{H}^{(1)}_{|I_2|+1}(x(v);z;I_2)
-\lambda \frac{\partial U^{(0)}_{|I|+2}(v,z;I\cup s)
}{\partial x(s)}\Big|_{s=v}
\nonumber
\\
&-\lambda \sum_{I_2'\uplus I_2''=I_2}
\Big(W^{(1)}_{|I_2'|+1}(v;I_2')
+ \frac{\lambda \delta_{I_2',\emptyset}}{(x(z)-x(v))^3}
\Big)
  U^{(0)}_{|I_2''|+1}(v,z;I_2'')\Big\}\;.
\end{align*}
\begin{proof}
Resolve (\ref{DSE-hUH}) at $g=1$ for the first term $\hat{U}^{(1)}(v,z;I)$,
which becomes
$-\sum_{I_1\uplus I_2=I, I_1\neq \emptyset}
\frac{\lambda W^{(0)}_{|I_1|+1}(v;I_1)}{x(z)+y(v)}\hat{U}^{(1)}(v,z;I_2)$
plus other terms. Iterate this procedure for every 
$\hat{U}^{(1)}(v,z;I_2)$ and so on. The resulting products
of $\frac{\lambda W^{(0)}_{|I_i|+1}(v;I_i)}{x(z)+y(v)}$ can be collected to
$D_{I_1}\frac{1}{(x(z)+y(v))}$.
\end{proof}
\end{lemma}

\begin{lemma}
\label{lemma:DIUsH}
\begin{align}
D_{I} \frac{U^{(0)}_{2}(v,z;s)}{H^{(0)}_1(x(v);z)}
&=-D_{I} \frac{\lambda \big(W^{(0)}_2(v;s)
-\frac{1}{x(v)-x(s)}\big)
}{(x(z)+y(v))^2}
\label{DIUsH}
\\
&+\sum_{k=1}^d 
D_{I}
\frac{\lambda W^{(0)}_{|I_0|+2}(\hat{z}^k;s)}{(x(v)+y(\hat{z}^k))
  (x(z)+y(v))}
\nonumber
\\
&-   \sum_{i=1}^{|I|}
  D_{I\setminus u_i} \frac{\lambda^2 W^{(0)}_2(u_i;s)}{
    (x(v)-x(u_i))(x(z)+y(u_i))^2(x(z)+y(v))}
\nonumber
\\
&-\sum_{I\uplus  I''=I}
D_{I' } \frac{\lambda}{x(z)+y(v)}
D_{I''} \frac{\frac{1}{(x(z)+y(v))}-\frac{1}{(x(z)+y(s))}
 }{ (x(v)-x(s))}\;.
\nonumber
\end{align}
\begin{proof}  
With (\ref{DUH-g0}) one has
\begin{align*}
D_{I} \frac{U^{(0)}_{2}(v,z;s)}{H^{(0)}_1(x(v);z)}
&=D_{I \cup s} \frac{U^{(0)}_{2}(v,z)}{H^{(0)}_1(x(v);z)}
+D_{I} \Big(\frac{U^{(0)}_{2}(v,z)}{H^{(0)}_1(x(v);z)}
D_s \log H^{(0)}_1(x(v);z)\Big)
\\
&=D_{I \cup s} \frac{1}{x(z){+}y(v)}
+\hspace*{-2mm}
\sum_{I\uplus  I''=I}\hspace*{-2mm}
D_{I' } \frac{1}{x(z){+}y(v)}
D_{I''\cup s} \log H^{(0)}_1(x(v);z).
\nonumber
\end{align*}
The first term equals
$D_{I \cup s} \frac{1}{x(z){+}y(v)}=-D_I\frac{\lambda W^{(0)}_2(v;s)}{(x(z)+y(v))^2}$ and gives partly the first line of (\ref{DIUsH}). The other part
cancels with a term in the last line of (\ref{DIUsH}); we will 
need this combination.
The last term is known from (\ref{eq:DlogH}) and 
(\ref{eq:DlogP-poleu}):
\begin{align}
  &D_{I''\cup s} \log H^{(0)}_1(x(v);z)
\label{DIUsH-2}
\\
&=
\sum_{k=1}^d 
\sum_{l=1}^{|I''|+1} \frac{(-1)^{l-1}}{l}
  \sum_{  \substack{I_1\uplus ...\uplus I_l=I''\cup s\\ 
      I_1, ... I_l\neq \emptyset}} 
  \prod_{i=1}^l \frac{\lambda W^{(0)}_{|I_i|+1}(\hat{z}^k;I_i)}{
    x(v)+y(\hat{z}^k)}
  \nonumber
  \\
  & +   \sum_{i=1}^{|I''|}
  D_{I''\setminus u_i} D_s \frac{\lambda}{ (x(v){-}x(u_i))(x(z){+}y(u_i))}
+  D_{I''} \frac{\lambda}{ (x(v){-}x(s))(x(z){+}y(s))}\;.
\nonumber
\end{align}
The middle line of (\ref{DIUsH-2}) equals
$\sum_{k=1}^d 
D_{I''}
\frac{\lambda W^{(0)}_{2}(\hat{z}^k;s)}{x(v)+y(\hat{z}^k)}$ and gives
with the derivation property of the $D_I$ the
second line of (\ref{DIUsH}). In the last line of (\ref{DIUsH-2}) we 
have $D_s \frac{\lambda}{ (x(v){-}x(u_i))(x(z){+}y(u_i))}
=-\frac{\lambda^2 W^{(0)}_2(u_i;s)}{
    (x(v)-x(u_i))(x(z)+y(u_i))^2}$ which gives the third line of 
  (\ref{DIUsH}). The final term of (\ref{DIUsH-2}) gives the missing part of
  the last line of  (\ref{DIUsH}).
\end{proof}
\end{lemma}

\begin{proposition}
  \label{prop:hatFj}
One has
\begin{align}
&\hat{F}^{(1)j}(x(z);I\setminus u_j)
\label{hatFj}
\\
&=
 D_{I\setminus u_j} \Big\{
 \frac{\lambda^3 \Omega^{(0)reg}_2(u_j,u_j)
}{(x(z)+y(u_j))^3}
+ \frac{
\frac{\lambda^3}{2}
\frac{\partial^2}{\partial (x(u_j))^2} 
\frac{1}{(x(z)+y(u_j))}
- \lambda^2 \big(W^{(1)}_{1}(u_j)
+ \frac{\lambda}{(x(z)-x(u_j))^3}
\big)}{(x(z)+y(u_j))^2}\Big\}
\nonumber
\\
&+   \sum_{\substack{i=1 \\ i\neq j}}^{|I|}
  D_{I\setminus \{u_i,u_j\}} \frac{\lambda^4 \Omega^{(0)}_2(u_i,u_j)}{
    (x(u_j)-x(u_i))(x(z)+y(u_i))^2(x(z)+y(u_j))^2}\;,
  \nonumber
\end{align}
where
$
\Omega^{(0)reg}(u,u):=\lim_{s\to u} 
\big(\Omega^{(0)}_2(u,s)-\frac{1}{(x(u)-x(s))^2}\big)
$
and $D_I\Omega^{(0)reg}_2(u,u)=\frac{\partial}{\partial x(s)}
W^{(0)}_{|I|+2}(u;I\cup s)\big|_{s=u}$ for $I\neq \emptyset$.
\begin{proof}
The residue of (\ref{F1I}) times $dx(v)$ at $x(v)=x(u_j)$ is with 
(\ref{DIH-g1b}) and (\ref{K1AI}) given by
\begin{align}
\hat{F}^{(1)j}_{|I|}(x(z);I\setminus u_j)
&= \lambda^3
\sum_{k=1}^d D_{I\setminus u_j}\frac{\Omega^{(0)}_2(\hat{z}^k,u_j)}{
  (x(u_j)+y(\hat{z}^k))  (x(z)+y(u_j))^2}
\label{hatFj-a}
\\
&+\lim_{v\to u_j} (x(v)-x(u_j))
D_I\frac{\hat{H}^{(1)}_{1}(x(v);z)}{H^{(0)}_1(x(v);z)}\;.
\nonumber
\end{align}
The limit in the last line follows from (\ref{hH-pole-u}) and the
derivation property of $D_I$:
\begin{align*}
  &\lim_{v\to u_j} (x(v)-x(u_j))
  D_I\frac{\hat{H}^{(1)}_{1}(x(v);z)}{H^{(0)}_1(x(v);z)}  
\\
&= 
\lambda
\Big\{\sum_{l=0}^{|I|-1}(-1)^l
\sum_{\substack{I_0\uplus I_1 \uplus ... \uplus I_l=I\setminus u_j
  \\ I_1,...,I_l \neq \emptyset}}
\frac{\hat{U}^{(1)}_{|I_0|+1}(z,u_j;I_0)}{H^{(0)}_1(x(u_j);z)}
\prod_{i=1}^l \frac{H^{(0)}_{|I_i|+1}(x(u_j);z;I_i)}{H^{(0)}_1(x(u_j);z)}
\\
&-\sum_{l=0}^{|I|-1}(-1)^l(l{+}1) \hspace*{-5mm}
\sum_{\substack{I_{-1}\uplus I_0\uplus I_1 \uplus ... \uplus I_l=I\setminus u_j
  \\ I_1,...,I_l \neq \emptyset}}\hspace*{-5mm}
\frac{U^{(0)}_{|I_{-1}|+1}(z,u_j;I_{-1})}{H^{(0)}_1(x(u_j);z)}
\frac{\hat{H}^{(1)}_{|I_0|+1}(x(u_j);z;I_0)}{H^{(0)}_1(x(u_j);z)}
\\
&\hspace*{6cm} \times
\prod_{i=1}^l \frac{H^{(0)}_{|I_i|+1}(x(u_j),x(z);I_i)}{H^{(0)}_1(x(u_j);z)}
\Big\}
\\
&\equiv  \lambda D_{I\setminus u_j}
\Big(\frac{\hat{U}^{(1)}_1(u_j,z)}{H^{(0)}_1(x(u_j);z)}
-\frac{U^{(0)}_1(u_j,z)}{H^{(0)}_1(x(u_j);z)}
\frac{\hat{H}^{(1)}_1(x(u_j);z)}{H^{(0)}_1(x(u_j);z)}\Big)
\\
&=\lambda D_{I\setminus u_j}
\frac{\hat{U}^{(1)}_1(u_j,z)}{H^{(0)}_1(x(u_j);z)}
-\lambda \!\!\!\! \sum_{I_1\uplus I_2=I\setminus u_j} \!\!\!
D_{I_1}\frac{U^{(0)}_1(u_j,z)}{H^{(0)}_1(x(u_j);z)}
D_{I_2} \frac{\hat{H}^{(1)}_1(x(u_j);z)}{H^{(0)}_1(x(u_j);z)}\;.
\end{align*}
We insert (\ref{DUH-g0}) and expand the other operations
$D_{I\setminus u_j}$ and $D_{I_2}$:
\begin{align}
  &\lim_{v\to u_j} (x(v)-x(u_j))
  D_I\frac{\hat{H}^{(1)}_{1}(x(v);z)}{H^{(0)}_1(x(v);z)}  
  \label{limDIH1H-a}
  \\
  &=\lambda
\sum_{l=0}^{|I|-1}(-1)^l
\sum_{\substack{I_0\uplus I_1 \uplus ... \uplus I_l=I\setminus u_j
  \\ I_1,...,I_l \neq \emptyset}}
\bigg(\frac{\hat{U}^{(1)}_{|I_0|+1}(u_j,z;I_0)}{H^{(0)}_1(x(u_j);z)}
\nonumber
\\
&\qquad
-\sum_{I_0'\uplus I_0''=I_0}
D_{I_0'} \frac{1}{x(z)+y(u_j)}
\frac{\hat{H}^{(1)}_{|I_0''|+1}(x(u_j);z;I_0'')}{H^{(0)}_1(x(u_j);z)}
\bigg)
\prod_{i=1}^l \frac{H^{(0)}_{|I_i|+1}(x(u_j);z;I_i)}{H^{(0)}_1(x(u_j);z)}
\nonumber
\\
&=-\lambda^2
\sum_{l=0}^{|I|-1}(-1)^l \hspace*{-4mm}
\sum_{\substack{I_0\uplus I_1 \uplus ... \uplus I_l=I\setminus u_j
    \\ I_1,...,I_l \neq \emptyset}}
\prod_{i=1}^l \frac{H^{(0)}_{|I_i|+1}(x(u_j);z;I_i)}{H^{(0)}_1(x(u_j);z)}
\nonumber
\\
&\times \bigg(
\sum_{I_0'\uplus I_0''=I_0} \hspace*{-3mm}
D_{I_0'}\frac{1}{x(z)+y(u_j)}\cdot
\frac{\partial}{\partial x(s)} \frac{U^{(0)}_{|I_0''|+2}(u_j,z;I_0''\cup s)
}{H^{(0)}_{1}(x(u_j);z)}
\Big|_{s=u_j}
\nonumber
\\
&+\hspace*{-2mm} \sum_{I_0'\uplus I_0''\uplus I_0'''=I_0}\hspace*{-5mm}
D_{I_0'}\frac{1}{x(z)+y(u_j)}
\Big(W^{(1)}_{|I_0''|+1}(u_j;I_0'')
+ \frac{\lambda \delta_{I_0'',\emptyset}}{(x(z){-}x(u_j))^3}
\Big)
\frac{U^{(0)}_{|I_0'''|+1}(u_j,z;I_0''')}{H^{(0)}_{1}(x(u_j);z)}
\bigg)
\nonumber
\\
&=-\lambda^2
\sum_{I'\uplus I''=I\setminus u_j}
D_{I'}\frac{1}{x(z)+y(u_j)}\cdot
\frac{\partial}{\partial x(s)}
D_{I''} \frac{U^{(0)}_{2}(u_j,z;s)}{H^{(0)}_1(x(u_j);z)}
  \Big|_{s=u_j}
  \nonumber
  \\
&- D_{I\setminus u_j}
\frac{\lambda^2 \big(W^{(1)}_{1}(u_j)
+ \frac{\lambda}{(x(z)-x(u_j))^3}
\big)}{(x(z)+y(u_j))^2}\;.
\nonumber
\end{align}
Here Lemma~\ref{lemma:hU1} was used to get the second equality and
(\ref{DUH-g0}) together with derivation property of $D_I$ to get
the last equality.

Next, in Lemma~\ref{lemma:DIUsH} we set $v\mapsto u_j$ and 
differentiate with respect to $x(s)$ at $s=u_j$.
With the definition of $\Omega^{(0)reg}_2(u_j,u_j)$ given in the
Proposition we get
\begin{align*}
  &\frac{\partial}{\partial x(s)}
  D_{I} \frac{U^{(0)}_{2}(u_j,z;s)}{H^{(0)}_1(x(u_j);z)}
  \Big|_{s=u_j}
  \nonumber
\\
&=-D_{I} \frac{\lambda \Omega^{(0)reg}_2(u_j,u_j)
}{(x(z)+y(u_j))^2}
+\sum_{k=1}^d 
D_{I}
\frac{\lambda \Omega^{(0)}_{2}(\hat{z}^k,u_j)}{(x(u_j)+y(\hat{z}^k))
  (x(z)+y(u_j))}
\\
&-   \sum_{i=1}^{|I|}
  D_{I\setminus u_i} \frac{\lambda^2 \Omega^{(0)}_2(u_i,u_j)}{
    (x(u_j)-x(u_i))(x(z)+y(u_i))^2)(x(z)+y(u_j))}
\\
&-\sum_{I\uplus I''=I}
D_{I' } \frac{\lambda}{2(x(z)+y(u_j))}
\frac{\partial^2}{\partial (x(u_j))^2} D_{I''}
\frac{1}{(x(z)+y(u_j))}\;.
\nonumber
\end{align*}
This result is inserted into (\ref{limDIH1H-a}) and then into
(\ref{hatFj-a}). The sum over
preimages cancels, the remainder simplifies to the assertion
(\ref{hatFj}).
\end{proof}
\end{proposition}

We proceed with the terms $F^{(1)a}_{|I|+1}$ in (\ref{F1aI}) which
contribute to $\hat{P}^1_{|I|+1}$. For them we need
$D_I \frac{Q^{(0)}_{1}(x(v),x(z))}{P^{(0)}_1(x(v),x(z))}$ which in the
case $I\neq \emptyset$ coincide with
$D_I \frac{\hat{Q}^{(0)}_{1}(x(v),x(z))}{P^{(0)}_1(x(v),x(z))}$.
The function $\hat{Q}^{(0)}_{1}(x(v),x(z))$ was given in
(\ref{tQ01}).  By Proposition~\ref{prop:DIlogQ}, the action of $D_I$
on functions of $\log\hat{Q}^{(0)}_{1}(x(v),x(z))$ is the same as
the combined action of $D_I$ on $P^{(0)}_1(x(z),x(z))$ and
$P^{(0)}_1(x(v),x(v))$ and symbolic expressions $D^0x(0)$ as given in
Proposition~\ref{prop:DIlogQ}.  Understanding $D_I^{[0]}$ as $D_I$
when acting on $P$ and $D^0_I$ when acting on $x(0)$, this means
\begin{align*}
  &D_I  \frac{Q^{(0)}_{1}(x(v),x(z))}{P^{(0)}_1(x(v),x(z))}
  \\
  &= -\frac{\lambda}{(x(v)-x(z))^2}
  D_I^{[0]}\Big\{
  \sqrt{1+\frac{(x(v)-x(z))^2}{4(x(v)-x(0))(x(z)-x(0))}}
  \nonumber
  \\
  &\times \exp\Big(\frac{1}{2}\log P^{(0)}_1(x(v),x(v))
  +\frac{1}{2}\log P^{(0)}_1(x(z),x(z))
  -\log P^{(0)}_1(x(v),x(z))\Big)\Big\}.
\end{align*}
We need this expression at and near the diagonal $x(v)=x(z)$.
Taylor expansion gives
\begin{align}
  &D_I  \frac{Q^{(0)}_{1}(x(v),x(z))}{P^{(0)}_1(x(v),x(z))}
  \label{DQPI-QPdiag}
  \\
  &= 
-D_I^0\frac{\lambda}{8(x(z)-x(0))^2}
-  \frac{\lambda}{2}
\frac{\partial^2\big(D_I\log P^{(0)}_1(x(w),x(z))\big)}{
    \partial x(w) \partial x(z)}\Big|_{w=z}
  \nonumber
  \\
&+(x(v)-x(z))
\Big(
  D_I^0\frac{\lambda}{8  (x(z)-x(0))^3}
-  \frac{\lambda}{2}
  \frac{\partial^3\big(D_I \log P^{(0)}_1(x(w),x(z))\big)}{
    \partial x(w) \partial (x(z))^2}\Big|_{w=z}
  \Big)
  \nonumber
  \\
  &+ \mathcal{O}((x(v)-x(z))^2)\;.
  \nonumber
\end{align}
We insert (\ref{DQPI-QPdiag}) into (\ref{F1aI}) and the result
together with Proposition~\ref{prop:hatFj} into (\ref{F1I}) to get the
part $F^{(1)}_{|I|+1}(x(v);x(z);I)$ of
$D_I\frac{\hat{P}^{(1)}_1(x(v),x(z))}{P^{(0)}_1(x(v),x(z))}$ in
(\ref{DIP-g1b}). The other part is $K_{0,|I|+1}(x(v);z;I)$ given in
(\ref{K1AI}), which for $A=0$ is a function of $x(z)$. To simplify the
total expression we write the last line of (\ref{K1AI}) for $A=0$ as
\begin{align*}
&\frac{\lambda^2}{(x(v)-x(z))}\sum_{k=0}^d 
\frac{\partial}{\partial x(w)} D_I \frac{1}{(x(v)+y(\hat{z}^k))
  (x(z)+y(w))}
\Big|_{w=\hat{z}^k}
\\
&=-\frac{\lambda^2}{(x(v)-x(z))}\sum_{k=0}^d 
D_I \frac{y'(\hat{z}^k)}{
  x'(\hat{z}^k)(x(v)+y(\hat{z}^k))
  (x(z)+y(\hat{z}^k))^2}
\\
&=-\frac{\lambda^2}{(x(v)-x(z))^3}\sum_{k=0}^d 
D_I\Big(
\frac{y'(\hat{z}^k)}{
  x'(\hat{z}^k)(x(v)+y(\hat{z}^k))}
-
\frac{y'(\hat{z}^k)}{
  x'(\hat{z}^k)  (x(z)+y(\hat{z}^k))^2}\Big)
\\
&-\frac{\lambda^2}{(x(v)-x(z))^2}\sum_{k=0}^d 
D_I \frac{y'(\hat{z}^k)}{
  x'(\hat{z}^k)(x(z)+y(\hat{z}^k))^2}
\\
&=-\frac{\lambda^2}{(x(v)-x(z))^3}\frac{\partial}{\partial x(w)}
\sum_{k=0}^d 
D_I\Big(\log(x(v)+y(\hat{w}^k))
-\log(x(z)+y(\hat{w}^k))\Big)_{w=z}
\\
&+\frac{\lambda^2}{(x(v)-x(z))^2}\frac{\partial^2}{\partial x(z)\partial x(w)}
\sum_{k=0}^d 
D_I \log (x(z)+y(\hat{w}^k))\Big|_{w=z}
\\
&=-\frac{\lambda^2}{(x(v)-x(z))^3}\frac{\partial}{\partial x(w)}
\Big(D_I \log P^{(0)}_1(x(v),x(w))-D_I \log P^{(0)}_1(x(z),x(w))\Big)_{w=z}
\\
&+\frac{\lambda^2}{(x(v)-x(z))^2}\frac{\partial^2}{\partial x(z)\partial x(w)}
D_I \log P^{(0)}_1(x(z),x(w))\Big|_{w=z}
\\
&-
\sum_{j=1}^{|I|} D_{I\setminus u_j}
\Big(\frac{\lambda^3}{(x(v)-x(z))(x(v)-x(u_j))(x(z)-x(u_j))^2(x(z)+y(u_j))^2}
\Big)\;,
\end{align*}
where Proposition~\ref{prop:DlogHP} has been used. Putting everything together,
we arrive (for $I\neq \emptyset$) at
\begin{theorem}
The auxiliary functions $\hat{P}^{(1)}_{|I|+1}$ of genus $g=1$ are 
determined by
\begin{align}
&D_I\frac{\hat{P}^{(1)}_1(x(v),x(z))}{P^{(0)}_1(x(v),x(z))}
\label{DPI-g1-final}
\\
&=
\sum_{k=0}^d D_I \frac{\lambda W^{(1)}_1(\hat{z}^k)}{x(v)+y(\hat{z}^k)}
+\frac{\lambda^2}{2} \sum_{\substack{j,k=0\\ j\neq k}}^d D_I
\frac{\Omega^{(0)}_2(\hat{z}^j,\hat{z}^k)}{(x(v)+y(\hat{z}^j))
(x(v)+y(\hat{z}^k))}
\nonumber
\\
&- \sum_{j=1}^{|I|}\frac{1}{(x(v)-x(u_j))}
\sum_{k=0}^d D_{I\setminus u_j}\frac{\lambda^3\Omega^{(0)}_2(\hat{z}^k,u_j)}{
  (x(v)+y(\hat{z}^k))  (x(z)+y(u_j))^2}
\nonumber
\\
&
-\frac{\lambda^2}{(x(v)-x(z))^3}\frac{\partial}{\partial x(w)}
\Big(D_I \log P^{(0)}_1(x(v),x(w))-D_I \log P^{(0)}_1(x(z),x(w))\Big)_{w=z}
\nonumber
\\
&+
\frac{\lambda^2}{(x(v)-x(z))^2}
\Big(
-D^0_I\frac{1}{8(x(z)-x(0))^2}
+  \frac{1}{2} 
  \frac{\partial^2 (D_I\log P^{(0)}_1(x(w),x(z)))}{
    \partial x(w) \partial x(z)}\Big|_{w=z}\Big)
\nonumber\\
&+\frac{1}{x(v)-x(z)}
D^{0}_I \frac{\lambda^2}{8(x(z)-x(0))^3}
\nonumber
\\
& +\sum_{j=1}^{|I|} \frac{1}{(x(v)-x(u_j))}
 D_{I\setminus u_j} \Big\{
 \frac{\lambda^3 \Omega^{(0)reg}_2(u_j,u_j)}{(x(z)+y(u_j))^3}
-  \frac{\lambda^2 W^{(1)}_1(u_j)}{(x(z)+y(u_j))^2}
\nonumber
\\
&\qquad\qquad +
\frac{\lambda^3}{2(x(z)+y(u_j))^2}
\frac{\partial^2}{\partial (x(u_j))^2} 
\frac{1}{(x(z)+y(u_j))}
\Big\}
\nonumber
\\
&+   \sum_{\substack{i,j=1 \\ i < j}}^{|I|}
  D_{I\setminus \{u_i,u_j\}} \frac{\lambda^4 \Omega^{(0)}_2(u_i,u_j)}{
 (x(v)-x(u_j))   (x(v)-x(u_i))(x(z)+y(u_i))^2(x(z)+y(u_j))^2}\;.
\nonumber
\end{align}

\end{theorem}
The Theorem also holds for $I=\emptyset$ where it specifies
to Proposition~\ref{prop:hP11}.

\subsection{Loop equations for genus $g=1$}

From (\ref{DPI-g1-final}) we extract
\begin{align}
  [(x(v))^{-1}]D_I\frac{\hat{P}^{(1)}_1(x(v),x(z))}{P^{(0)}_1(x(v),x(z))}
\label{DPI-g1-final-v1}
&=
\lambda\sum_{k=0}^d W^{(1)}_{|I|+1}(\hat{z}^k;I)
+
D^{0}_I \frac{\lambda^2}{8(x(z)-x(0))^3}
\\
& +\sum_{j=1}^{|I|} 
 D_{I\setminus u_j} \Big\{
 \frac{\lambda^3 \Omega^{(0)reg}_2(u_j,u_j)}{(x(z)+y(u_j))^3}
-  \frac{\lambda^2 W^{(1)}_1(u_j)}{(x(z)+y(u_j))^2}
\nonumber
\\
&\qquad -
\frac{\lambda^3}{2(x(z)+y(u_j))^2}
\frac{\partial^2}{\partial (x(u_j))^2} 
\frac{1}{(x(z)+y(u_j))}
\Big\}
\nonumber
\end{align}
and
\begin{align}
&[(x(v))^{-2}]D_I\frac{\hat{P}^{(1)}_1(x(v),x(z))}{P^{(0)}_1(x(v),x(z))}
\label{DPI-g1-final-v2}
\\
&=
-\lambda\sum_{k=0}^d y(\hat{z}^k)W^{(1)}_{|I|+1}(\hat{z}^k;I)
-\lambda^2 \sum_{\substack{I_1\uplus I_2=I\\ I_2 \neq \emptyset}}
  \sum_{k=0}^d W^{(1)}_{|I|+1}(\hat{z}^k;I_1)
  W^{(0)}_{|I|+1}(\hat{z}^k;I_2)
\nonumber
\\
&- \sum_{j=1}^{|I|} 
\sum_{k=0}^d D_{I\setminus u_j} 
\frac{\lambda^3 \Omega^{(0)}_2(\hat{z}^k,u_j)}{(x(z)+y(u_j))^2}
\tag{*}
\\
&  +\frac{\lambda^2}{2}
\sum_{\substack{j,k=0\\ j\neq k}}^d D_I\Omega^{(0)}_2(\hat{z}^j,\hat{z}^k)
+  \frac{\lambda^2}{2} 
  \frac{\partial^2 (D_I\log P^{(0)}_1(x(w),x(z)))}{
    \partial x(w) \partial x(z)}\Big|_{w=z}
\tag{\dag}
  \\
&
-D^0_I\frac{\lambda^2}{8(x(z)-x(0))^2}
+x(z) D^{0}_I \frac{\lambda^2}{8(x(z)-x(0))^3}
  \nonumber
  \\
& +\sum_{j=1}^{|I|} x(u_j)
 D_{I\setminus u_j} \Big\{
 \frac{\lambda^3 \Omega^{(0)reg}_2(u_j,u_j)}{(x(z)+y(u_j))^3}
-  \frac{\lambda^2 W^{(1)}_1(u_j)}{(x(z)+y(u_j))^2}
\nonumber
\\
&\qquad\qquad +
\frac{\lambda^3}{2(x(z)+y(u_j))^2}
\frac{\partial^2}{\partial (x(u_j))^2} 
\frac{1}{(x(z)+y(u_j))}
\Big\}
\nonumber
\\
&+  \frac{1}{2} \sum_{\substack{i,j=1 \\ i \neq  j}}^{|I|}
  D_{I\setminus \{u_i,u_j\}} \frac{\lambda^4 \Omega^{(0)}_2(u_i,u_j)}{
(x(z)+y(u_i))^2(x(z)+y(u_j))^2}\;.
\tag{**}
\end{align}
The combination of terms in the lines (*), $(\dag)$ and (**) simplifies
considerably:
\begin{lemma}
\label{lemma:qle-simplify}
\begin{align}
& - \sum_{j=1}^{|I|} 
\sum_{k=0}^d D_{I\setminus u_j} 
\frac{\lambda^3 \Omega^{(0)}_2(\hat{z}^k,u_j)}{(x(z)+y(u_j))^2}
+  \frac{\lambda^2}{2} 
  \frac{\partial^2 (D_I\log P^{(0)}_1(x(w),x(z)))}{
    \partial x(w) \partial x(z)}\Big|_{w=z}
  \nonumber
  \\
&
+\frac{\lambda^2}{2}\!
\sum_{\substack{j,k=0\\ j\neq k}}^d D_I\Omega^{(0)}_2(\hat{z}^j,\hat{z}^k)
+  \frac{1}{2} \sum_{\substack{i,j=1 \\ i \neq  j}}^{|I|}\!
  D_{I\setminus \{u_i,u_j\}} \frac{\lambda^4 \Omega^{(0)}_2(u_i,u_j)}{
(x(z)+y(u_i))^2(x(z)+y(u_j))^2}
  \nonumber
\\
&=
 -\frac{\lambda^2}{2}
 \sum_{k=0}^d D_I\Omega^{(0)reg}_2(\hat{z}^k,\hat{z}^k)
+\frac{\lambda^3}{6}
\sum_{j=1}^{|I|}
\frac{\partial}{\partial x(u_j)}
\Big(D_{I\setminus u_j}
\frac{1}{(x(z)+y(u_j))^3}\Big)\;.
\label{lemma-qle1}
\end{align}
\begin{proof}
  We start from (\ref{eq:DlogP}) for $v\mapsto z\mapsto w$ and
  differentiate with respect to $x(z)$. In the second step we take 
(\ref{lle-g0}) for $I\mapsto I\cup \hat{w}^l$ into account:
\begin{align}
&  \frac{\partial (D_I\log P^{(0)}_1(x(z),x(w)))}{\partial x(z)}
\label{DPI-partialz}  \\
  &= 
  \sum_{l=0}^d D_I \frac{1}{x(z)+y(\hat{w}^l)}
 -\sum_{j=1}^{|I|} D_{I\setminus u_j}
  \frac{\lambda}{(x(z)-x(u_j))^2(x(w)+y(u_j))}
  \nonumber
  \\
&= -D_I
\sum_{k,l=0}^d \Big(W^{(0)}_2(\hat{z}^k;\hat{w}^l)
-\frac{\delta_{k,l}}{(x(\hat{z}^k)-x(\hat{w}^l))}\Big)
  \nonumber
\\
&-\sum_{j=1}^{|I|}D_{I\setminus u_j}
\Big(
\frac{\lambda}{(x(z)-x(u_j))^2(x(w)+y(u_j))}
+
\sum_{l=0}^d
D_{\hat{w}^l} \frac{1}{x(z)+y(u_j)}
\Big)\;.
  \nonumber
\end{align}
We have 
$\sum_{l=0}^d D_{\hat{w}^l} \frac{1}{x(z)+y(u_j)}
=-\sum_{l=0}^d \frac{\lambda W^{(0)}_2(u_j;\hat{w}^l)}{(x(z)+y(u_j))^2}$.
When differentiating with respect to $x(w)$ at $w=z$, this term is
the first one in (\ref{lemma-qle1}),
but with relative factor
$-2$. Moreover, 
in the limit $w\to z$ the contributions with $l\neq k$ in the third line
of (\ref{DPI-partialz}) cancel
the first term of the second line of (\ref{lemma-qle1}),
whereas for $k=l$ the regularised $\Omega^{(0)reg}_2(\hat{z}^k,\hat{z}^k)$
appears. We thus have
\begin{align*}
& - \sum_{j=1}^{|I|} 
\sum_{k=0}^d D_{I\setminus u_j} 
\frac{\lambda^3 \Omega^{(0)}_2(\hat{z}^k,u_j)}{(x(z)+y(u_j))^2}
+  \frac{\lambda^2}{2} 
  \frac{\partial^2 (D_I\log P^{(0)}_1(x(w),x(z)))}{
    \partial x(w) \partial x(z)}\Big|_{w=z}
  \nonumber
  \\
&
+\frac{\lambda^2}{2}
\sum_{\substack{j,k=0\\ j\neq k}}^d D_I\Omega^{(0)}_2(\hat{z}^j,\hat{z}^k)
  \nonumber
\\
&= -\frac{\lambda^2}{2}
\sum_{k=0}^d D_I\Omega^{(0)reg}_2(\hat{z}^k,\hat{z}^k)
\\
&+\frac{\lambda^2}{2}
\sum_{j=1}^{|I|}\sum_{I_1\uplus I_2 = I\setminus u_j}
\hspace*{-5mm} D_{I_1}
\frac{\lambda}{(x(z)+y(u_j))^2}
D_{I_2}\Big(\frac{\delta_{I_2,\emptyset}}{(x(z)-x(u_j))^2}
-\sum_{l=0}^d \Omega^{(0)}_2(u_j,\hat{z}^l)
\Big)\;.
\end{align*}
In the $D_{I_2}$ operation we use the symmetry under
$u_j\leftrightarrow \hat{z}^l$ to take a
$x(u_j)$ derivative out.
Using (\ref{lle-g0}) we then get
\begin{align*}
&D_{I_2}\Big(\frac{\delta_{I_2,\emptyset}}{(x(z)-x(u_j)^2)}
-\sum_{l=0}^d
\Omega^{(0)}_2(u_j,\hat{z}^l)
\Big)
\\
&=\frac{\partial}{\partial x(u_j)}
\Big(\frac{\delta_{I_2,\emptyset}}{(x(z)-x(u_j))}
-\sum_{l=0}^d
W^{(0)}_2(z^l;I_2\cup u_j)
\Big)
\\
&=\frac{\partial}{\partial x(u_j)}
\Big(D_{I_2} \frac{1}{x(z)+y(u_j)}
+\sum_{i=1}^{|I_2|} D_{I_2\setminus u_i} D_{u_j}\frac{1}{x(z)+y(u_i)}
\Big)
\\
&=\frac{\partial}{\partial x(u_j)}
D_{I_2} \frac{1}{x(z)+y(u_j)}
-\sum_{i=1}^{|I_2|} D_{I_2\setminus u_i} \frac{
  \lambda \Omega^{(0)}_2(u_i,u_j)}{(x(z)+y(u_i))^2}\;.
\end{align*}
Inserted back, the last term leads to a double
sum over pairs $i\neq j$ and a $D_{I\setminus \{u_i,u_j\}}$ operation,
which exactly cancels the second term of the second line of
(\ref{lemma-qle1}). An obvious rearrangement of
$\sum_{I_1\uplus I_2=I\setminus u_j}
D_{I_1}\frac{1}{(x(z)+y(u_j))^2}  \frac{\partial}{\partial x(u_j)}
D_{I_2} \frac{1}{x(z)+y(u_j)}$ confirms the assertion.
\end{proof}
\end{lemma}

On the other hand, comparing (\ref{def-H}) 
with (\ref{DSE-W}) we get as leading coefficient
$[(x(v))^{-1}] H^{(g)}_{|I|+1}(x(v);q;I)
= -\lambda W^{(g)}_{|I|+1}(q;I)$ for any $g\geq 1$. 
Then (\ref{def-P}) and (\ref{def-hP}) show
\begin{align}
&[(x(v))^{-2}]\frac{\hat{P}^{(g)}_{|I|+1}(x(v),x(z);I)}{P^{(0)}_1(x(v),x(z))}  
\\
&= \frac{\lambda}{N}\sum_{l=1}^d \frac{\lambda
W^{(g)}_{|I|+1}(\varepsilon_l;I)}{x(z)-x(\varepsilon_l)}
-\lambda^2 \sum_{j=1}^{|I|} \frac{W^{(g)}_{|I|}(u_j;I\setminus u_j)}{
  x(z)-x(u_j)}\quad \text{for }g\geq 1
\nonumber
\end{align}
as leading coefficient. Comparison with 
(\ref{DPI-g1-final-v1}),
(\ref{DPI-g1-final-v2}) and Lemma~\ref{lemma:qle-simplify}
gives:
\begin{proposition}
  The genus $1$ meromorphic functions $W^{(1)}_{|I|+1}(z;I)$
  satisfy the linear loop equation
\begin{align}
\sum_{k=0}^d W^{(1)}_{|I|+1}(\hat{z}^k;I)
&=
-D^{0}_I \frac{\lambda}{8(x(z)-x(0))^3}
  \label{lleq-g1}
\\
& -\sum_{j=1}^{|I|} 
 D_{I\setminus u_j} \Big\{
 \frac{\lambda^2 \Omega^{(0)reg}_2(u_j,u_j)}{(x(z)+y(u_j))^3}
-  \frac{\lambda W^{(1)}_1(u_j)}{(x(z)+y(u_j))^2}
\nonumber
\\
&\qquad\qquad -
\frac{\lambda^2}{2(x(z)+y(u_j))^2}
\frac{\partial^2}{\partial (x(u_j))^2} 
\frac{1}{(x(z)+y(u_j))}
\Big\}
\nonumber
\end{align}
and the quadratic loop equation
\begin{align}
  &-\sum_{k=0}^d y(\hat{z}^k)W^{(1)}_{|I|+1}(\hat{z}^k;I)
  \label{qleq-g1}
\\
&=\lambda \sum_{\substack{I_1\uplus I_2=I\\ I_2 \neq \emptyset}}
  \sum_{k=0}^d W^{(1)}_{|I|+1}(\hat{z}^k;I_1)
  W^{(0)}_{|I|+1}(\hat{z}^k;I_2)
+\frac{\lambda}{2}
\sum_{k=0}^d D_I\Omega^{(0)reg}_2(\hat{z}^k,\hat{z}^k)
\tag{*}
\\
&  
-\frac{\lambda^2}{6}
\sum_{j=1}^{|I|}
\frac{\partial}{\partial x(u_j)}
\Big(D_{I\setminus u_j}
\frac{1}{(x(z)+y(u_j))^3}\Big)
\tag{\dag}
\\
&
+D^0_I\frac{\lambda}{8(x(z)-x(0))^2}
-x(z) D^{0}_I \frac{\lambda}{8(x(z)-x(0))^3}
\tag{\S}
  \\
& -\sum_{j=1}^{|I|} x(u_j)
 D_{I\setminus u_j} \Big\{
 \frac{\lambda^2 \Omega^{(0)reg}_2(u_j,u_j)}{(x(z)+y(u_j))^3}
-  \frac{\lambda W^{(1)}_1(u_j)}{(x(z)+y(u_j))^2}
\tag{\ddag}
\\
&\qquad\qquad -
\frac{\lambda^2}{2(x(z)+y(u_j))^2}
\frac{\partial^2}{\partial (x(u_j))^2} 
\frac{1}{(x(z)+y(u_j))}
\Big\}
\tag{\ddag}
\\
&+ \frac{\lambda}{N}\sum_{l=1}^d \frac{
W^{(1)}_{|I|+1}(\varepsilon_l;I)}{x(z)-x(\varepsilon_l)}
-\lambda \sum_{j=1}^{|I|} \frac{W^{(1)}_{|I|}(u_j;I\setminus u_j)}{
  x(z)-x(u_j)}\;.
\nonumber
\end{align}
\end{proposition}

\section{The recursion formula}

\label{sec:recursion}

We learn from the loop equations (\ref{qle-g0}) and (\ref{lle-g0})
for $g=0$, the loop equations 
(\ref{qleq-g1}) and (\ref{lleq-g1}) for $g=1$
and the identity $y(u)=-x(-u)$, see (\ref{glInvolution}):
\begin{itemize}
\item The only poles of
  $z\mapsto x'(z) W^{(g)}_{|I|+1}(z;I)$ for $2g+|I|>1$ are located at the
  ramification points $z=\beta_i$ of $x$, at the opposite diagonals $z=-u_i$
  for $u_i\in I$ and (for $g\geq 1$) at $z=0$. The pole at $z=u_i$
  in the last line of
  (\ref{qleq-g1}) is due to the particular case $I_2=\{u_i\}$ in the
  first term of the line (*) of (\ref{qleq-g1}); it is not a pole of
  $x'(z) W^{(1)}_{|I|+1}(z;I)$.  Similarly, the pole at
  $z=\varepsilon_k$ in the last line of (\ref{qleq-g1}) is due to the
  pole of $y(z)$ at $z=\varepsilon_k$ in the first line of
  (\ref{qleq-g1}). The same discussion applies to (\ref{qle-g0}).

  Moving the integration contour to the
  complementary poles, we get
\begin{align}
&x'(z) W^{(g)}_{|I|+1}(z;I)dz
= \Res\displaylimits_{q\to z}
\frac{x'(q) W^{(g)}_{|I|+1}(q;I)dqdz}{q-z}
\label{Wg1-residue-1}
\\
&= \sum_{i=1}^{2d} \Res\displaylimits_{q\to \beta_i}
\frac{dz}{z-q}
x'(q) W^{(g)}_{|I|+1}(q;I)dq
\nonumber
\\
& +\sum_{i=1}^{|I|} \Res\displaylimits_{q\to -u_i}
\frac{dz}{z-q}
x'(q) W^{(g)}_{|I|+1}(q;I)dq
+\Res\displaylimits_{q\to 0}
\frac{dz}{z-q}
x'(q) W^{(g)}_{|I|+1}(q;I)dq\;.
\nonumber
\end{align}

\item The linear loop equations (\ref{lle-g0}) and (\ref{lleq-g1}) imply
  that the poles
  of $z\mapsto x'(z) W^{(g)}_{|I|+1}(I;z)dz$ at 
  $z=-u_i$ and $z=0$ have vanishing residue,
  \[
    \Res\displaylimits_{z\to u_i} x'(z) W^{(g)}_{|I|+1}(I;z)dz=0\;,\qquad
    \Res\displaylimits_{z\to 0} x'(z) W^{(g)}_{|I|+1}(I;z)dz=0
\]    
(since the integrands are total differentials in $z$).
Renaming $z\mapsto q$, we thus have
\begin{align}
 0&= -\sum_{i=1}^{|I|} \Res\displaylimits_{q\to -u_i}
\frac{dz}{z+u_i}
x'(q) W^{(g)}_{|I|+1}(I;q)dq
-\Res\displaylimits_{q\to 0}
\frac{dz}{z}
x'(q) W^{(g)}_{|I|+1}(I;q)dq\;,
\nonumber
\end{align}
which we can safely add to (\ref{Wg1-residue-1}).

\item Let 
$\sigma_i(z)$ be the local Galois conjugate near $\beta_i$, i.e.\ the
unique preimage 
$\hat{z}^{j_i}=\sigma_i(z) \in x^{-1}(x(z))$ with
$\lim_{z\to \beta_i}\sigma_i(z)=\beta_i$ and $\sigma_i(z)\not\equiv z$.
A residue at $\beta_i$ is invariant
under Galois involution,
$\Res\displaylimits_{q\to \beta_i} f(q)dq=
\Res\displaylimits_{q\to \beta_i} f(\sigma_i(q))d\sigma_i(q)$, so that
\begin{align*}
&\Res\displaylimits_{q\to \beta_i}
\frac{dz}{z-q}
x'(q) W^{(g)}_{|I|+1}(q;I)dq
\\
&=\frac{1}{2}
\Res\displaylimits_{q\to \beta_i}
\Big(\frac{dz}{z-q}
x'(q) W^{(g)}_{|I|+1}(q;I)dq
+\frac{dz}{z-\sigma_i(q)}
x'(\sigma_i(q)) W^{(1)}_{|I|+1}(\sigma_i(q);I)d\sigma_i(q)\Big)
\\
&=\frac{1}{2}
\Res\displaylimits_{q\to \beta_i}
\Big(\frac{dz}{z-q}-\frac{dz}{z-\sigma_i(q)}\Big)
x'(q) W^{(g)}_{|I|+1}(q;I)dq
\\
&+\frac{1}{2}
\Res\displaylimits_{q\to \beta_i}
\frac{dz}{z-\sigma_i(q)}
\Big(x'(q) W^{(g)}_{|I|+1}(q;I)dq
+x'(\sigma_i(q)) W^{(g)}_{|I|+1}(\sigma_i(q);I)d\sigma_i(q)\Big)\;.
\end{align*}  
The last line vanishes because
$x'(q)dq=x'(\sigma_i)d\sigma_i(q)$
and $W^{(g)}_{|I|+1}(q;I)+W^{(g)}_{|I|+1}(\sigma_i(q);I)$
is holomorphic at $q=\beta_i$ as consequence of the
linear loop equations (\ref{lle-g0}) and (\ref{lleq-g1}).

\end{itemize}
In summary,
\begin{align}
x'(z) W^{(g)}_{|I|+1}(z;I)dz
&= \sum_{i=1}^{2d} \Res\displaylimits_{q\to \beta_i}
\Big(\frac{1}{2}\Big(\frac{dz}{z-q}-\frac{dz}{z-\sigma_i(q)}\Big)
x'(q) W^{(g)}_{|I|+1}(q;I)dq\Big)
\nonumber
\\
& +\sum_{i=1}^{|I|} \Res\displaylimits_{q\to -u_i}
\Big(
\Big(\frac{dz}{z-q}-\frac{dz}{z+u_i}\Big)
x'(q) W^{(g)}_{|I|+1}(q;I)dq\Big)
\nonumber
\\
&+\Res\displaylimits_{q\to 0}
\Big(\Big(\frac{dz}{z-q}-\frac{dz}{z}\Big)
x'(q) W^{(g)}_{|I|+1}(q;I)dq\Big)\;,
\label{Wg1-residue}
\end{align}
where the last line vanishes identically for $g=0$.

We can eventually complete the
\begin{proof}[Proof of Theorem \ref{thm:mainrecursion} for $g\leq 1$]
We identify the three contributions on the rhs of (\ref{Wg1-residue}).
\begin{enumerate}
\item  \emph{(poles at $z=\beta_i$)}
For $g=1$,  only the lhs and the line (*) of the
rhs of (\ref{qleq-g1}) have poles at
$z=\beta_i$; more precisely only the principal preimage $k=0$ and
the Galois preimage $k=k_i\in\{1,...,d\}$ with $\hat{z}^{k_i}=\sigma_i(z)$.
All terms with other $k$ and all other lines in  
(\ref{qleq-g1}) are holomorphic at $z=\beta_i$.
Similarly, only the lhs and the first term of the rhs of 
(\ref{qle-g0}) have poles at
$z=\beta_i$; more precisely only the principal preimage $k=0$ and
the Galois preimage $k=k_i\in\{1,...,d\}$ with $\hat{z}^{k_i}=\sigma_i(z)$.
The lhs of the
quadratic loop equations for $z\mapsto q$, multiplied by
$x'(q)dq=x'(\sigma_i)d\sigma_i(q)$, is written as
\begin{align*}
&-\sum_{k=0}^d   y(\hat{q}^k)  W^{(g)}_{|I|+1}(\hat{q}^k;I)  x'(q)dq
\\
&=
\big(-y(q)  W^{(g)}_{|I|+1}(q;I)  
-y(\sigma_i(q)) W^{(g)}_{|I|+1}(\sigma_i(q);I)\big)  x'(q)dq
+ \mathcal{O}((q-\beta_i)^1)dq
\\
&=
-(y(q)-y(\sigma_i(q))) W^{(g)}_{|I|+1}(q;I)  x'(q)dq
\\
&- y(\sigma_i(q)) (W^{(g)}_{|I|+1}(q;I)  +W^{(g)}_{|I|+1}(\sigma_i(q);I))  x'(q)dq
+ \mathcal{O}((q-\beta_i)^1)dq
\tag{**}
\\
&=
-(y(q)-y(\sigma_i(q))) W^{(g)}_{|I|+1}(q;I)  x'(q)dq
+ \mathcal{O}((q-\beta_i)^1)dq\;,
\end{align*}
where we used that (\ref{lle-g0}) and  (\ref{lleq-g1}) make
the whole line (**)
of order $\mathcal{O}((q-\beta_i)^1)dq$.
Thus, the quadratic loop equations multiplied by $
\frac{x'(q)dq}{-(y(q)-y(\sigma_i(q)))}$ (which is of order
$\mathcal{O}((q-\beta_i)^0)dq$), takes the form
\begin{align}
  & W^{(g)}_{|I|+1}(q;I)x'(q)dq
  \label{W1-nearbeta}
\\
&=-\frac{\lambda x'(q)dq}{ (y(q)-y(\sigma_i(q)))}
\sum_{q'\in \{q,\sigma_i(q)\}}
\Big(
\frac{1}{2}\sum_{\substack{I_1\uplus I_2=I\\ g_1+g_1=g\\
    (g_i,I_i)  \neq (0,\emptyset)}}
W^{(g_1)}_{|I_1|+1}(q';I_1)
W^{(g_2)}_{|I_2|+1}(q';I_2)
\nonumber
\\[-1ex]
&\hspace*{5cm} +\frac{1}{2}D_I \Omega^{(g-1)reg}_{2}(q',q')\Big)
+ \mathcal{O}((q-\beta_i)^0)dq\;.
\nonumber
\end{align}
When inserting this into (\ref{Wg1-residue}), both cases $q'=q$ and
$q'=\sigma_i(q)$ give the same contribution since the residue at
$\beta_i$ is invariant under local Galois involution considering the
linear loop equation.  When translating to
$\omega^{(g)}_{n+1}(z,u_1,...,u_n) =\lambda^{2g+n-1} d_{u_1}\cdots
d_{u_n} W^{(g)}_{n+1}(z;u_1,...,u_n) dx(z)$ the first line of
(\ref{eq:mainformula}) with recursion kernel $K_i(z;q)$ results.

\item \emph{(poles at $z=-u_j$)}
For $g=1$, these are present on the lhs and the line (*) of the rhs of
(\ref{qleq-g1}), there only in the principal preimage $k=0$, 
and in the lines (\dag) and
(\ddag) of (\ref{qleq-g1}), there only in the $j$-summands.
The line (\dag), by the linear loop equation
(\ref{lle-g0}) at genus $g=0$, takes the form 
\begin{align*}
  \eqref{qleq-g1}_\dag
&=-\frac{\lambda^2}{12}
\frac{\partial^3}{\partial x(u_j)\partial (x(z))^2}
\Big(D_{I\setminus u_j}
\frac{1}{(x(z)+y(u_j))}\Big)
+\mathcal{O}((z+u_j)^0)
\\
&=\frac{\lambda^2}{12}
\frac{\partial^3W^{(0)}_{|I|+1}(z;I)}{\partial x(u_j)\partial (x(z))^2}
+\mathcal{O}((z+u_j)^0)
\\
&=\frac{\lambda^2}{12}
\frac{\partial^2 D_{I\setminus u_j} \Omega^{(0)}_{2}(z,u_j)}{\partial (x(z))^2}
+\mathcal{O}((z+u_j)^0)\;.
\end{align*}
By the
linear loop equation (\ref{lleq-g1}), the $j$-summand of the lines (\ddag)
can be written as
\begin{align*}
  \eqref{qleq-g1}_\ddag
&= x(u) W^{(1)}_{|I|+1}(z;I) + \mathcal{O}((z+u_j)^0)\;.
\end{align*}
Similarly, only the first two lines of
(\ref{qle-g0}) have poles at $z=-u_j$, again only the principal
preimages $k=0$ and the $j$-summand of the last term of the second line of 
(\ref{qle-g0}). The latter is with 
(\ref{lle-g0}) also of the form 
$x(u) W^{(0)}_{|I|+1}(z;I) + \mathcal{O}((z+u_j)^0)$ that we noticed for $g=1$.
We bring these terms
$x(u) W^{(g)}_{I|+1}(z;I)$ to the lhs of the quadratic loop equations,
rename $z\mapsto q$ and
multiply by $\frac{x'(q)dq}{-(y(q)+x(u_j))}$ to get
\begin{align}
&W^{(g)}_{|I|+1}(q;I) x'(q) dq
\label{W1-nearu}
\\
&=\frac{\lambda x'(q)dq}{-(y(q)+x(u_j))}
\Big(\frac{1}{2}
\sum_{\substack{I_1\uplus I_2=I\\ g_1+g_1=g\\
    (g_i,I_i)  \neq (0,\emptyset)}}
W^{(g_1)}_{|I_1|+1}(q;I_1)
  W^{(g_2)}_{|I_2|+1}(q;I_2)
+\frac{1}{2}
D_I\Omega^{(g-1)reg}_2(q,q)
\nonumber
\\[-2ex]
&\hspace*{5.5cm}
+\frac{\lambda}{12}
\frac{\partial^2 D_{I\setminus u_j} \Omega^{(g-1)}_{2}(q,u_j)}{\partial (x(q))^2}
\Big)
+\mathcal{O}((q+u_j)^{-1})dq\;.
\nonumber
\end{align}
Note that the undetermined residue poses no problem since it is
projected away in (\ref{Wg1-residue}).
When translating to $\omega^{(g)}_{n+1}(z,u_1,...,u_n)
=\lambda^{2g+n-1} d_{u_1}\cdots d_{u_n}  W^{(g)}_{n+1}(z;u_1,...,u_n) dx(z)$,
the second and third lines of (\ref{eq:mainformula}) with recursion kernel
$K_{u_j}(z;q)$ result. Here one has take into account that the differential
$d_{u_j}$ does not commute with $K_{u_j}(z;q)$. This makes it necessary to
keep the primitive $d_{u_j}^{-1}$ inside the residue.

\item \emph{(pole at $z=0$)}
There is no such pole for $g=0$. 
For $g=1$,  this pole is present on the lhs and the line (*) of the rhs of
(\ref{qleq-g1}), there only in the principal preimage $k=0$, 
and in both terms of the line (\S). We write the first term as
\begin{align*}
D^0_I\frac{\lambda}{8(x(z)-x(0))^2}
&=-\frac{\lambda}{8}\frac{\partial^2 (D_I^0 \log (x(z)-x(0)))}{
  \partial (x(z))^2}
\\
&= -\frac{\lambda}{8}\frac{\partial^2 (D_I \log P^{(0)}_1(x(z),x(z)))}{
  \partial (x(z))^2} + \mathcal{O}(z^0)
\end{align*}
by Proposition~\ref{prop:DIlogQ}. Next,
from (\ref{DPI-partialz}) we know
\begin{align*}
\frac{\partial (D_I \log P^{(0)}_1(x(z),x(z)))}{
    \partial (x(z))} &=\lim_{w\to z}
2 \frac{\partial (D_I \log P^{(0)}_1(x(z),x(w)))}{
    \partial (x(z))} \Big|_{w=z}
  \\
  &= -2D_I W^{(0)reg}_2(z;z) + \mathcal{O}(z^0)\;.
\end{align*}
By the linear loop equation (\ref{lleq-g1}), the second term in the line
(\S) of (\ref{qleq-g1}) can be written as
\[
-x(z)D_I^0\frac{\lambda}{8(x(z)-x(0))^3}
= x(z) W^{(1)}_{|I|+1}(z;I)+ \mathcal{O}(z^0)\;.
\]
We bring this term to the lhs, rename $z\mapsto q$ and
multiply by $\frac{x'(q)dq}{-(y(q)+x(q))}$ to get with
the previous considerations
\begin{align}
&W^{(1)}_{|I|+1}(q;I) x'(q) dq
\label{W1-near0}
\\
&=\frac{\lambda x'(q)dq}{-(y(q)+x(q))}
\Big(\sum_{\substack{I_1\uplus I_2=I\\ I_2 \neq \emptyset}}
W^{(1)}_{|I_1|+1}(q;I_1)
  W^{(0)}_{|I_2|+1}(q;I_2)
+\frac{1}{2}
D_I\Omega^{(0)reg}_2(q,q)
\nonumber
\\[-2ex]
&\hspace*{5cm} + \frac{1}{4} 
\frac{\partial (D_I W^{(0)reg}_2(q;q))}{\partial x(q)}
\Big)
+ \mathcal{O}(q^{-1})\;.
\nonumber
\end{align}
Again the undetermined residue poses no problem since it is
projected away in (\ref{Wg1-residue}).
When translating to $\omega^{(g)}_{n+1}(z,u_1,...,u_n)
=\lambda^{2g+n-1} d_{u_1}\cdots d_{u_n}  W^{(g)}_{n+1}(z;u_1,...,u_n) dx(z)$,
the last two lines of (\ref{eq:mainformula}) with recursion kernel
$K_{0}(z;q)$ result.
\end{enumerate}
\end{proof}

\section{Outlook}

We continued the work of \cite{Grosse:2019jnv,Schurmann:2019mzu,
  Branahl:2020yru,Hock:2021tbl} and pushed the proof of the conjecture
\cite[Conj.~6.1]{Branahl:2020yru} that the quartic Kontsevich model
obeys blobbed topological recursion \cite{Borot:2015hna} to genus
$g\leq 1$.  The method that we developed in this paper is general and
powerful enough to achieve the proof to any $g$. The Dyson-Schwinger
equations (\ref{DSE-hHP}) provide the difference
$D_I D_g \log P^{(0)}_1(x(v),x(z)) -D_I D_g \log H^{(0)}_1(x(v);z)$,
where $D_I$ is the loop insertion operator of
Definition~\ref{def:loopins} and $D_g$ a `genus insertion' still to
make precise. Then $D_I D_g \log P^{(0)}_1(x(v),x(z))$ is this
difference symmetrised in all preimages $\hat{z}^k$ plus a function of
$(x(v),x(z);I)$ with simple poles at $x(v)=x(u_i)$ and poles at
$x(v)=x(z)$ up to order $4g+3$. The principal part of the
corresponding Laurent series is uniquely determined by (\ref{def-H}),
(\ref{def-P}), (\ref{def-hH}) and (\ref{def-hP}) in terms of Taylor
expansions of $Q^{(h)}_1(x(v),x(z);I)$ and $P^{(h)}_1(x(v),x(z);I)$
for $h<g$ at the diagonal $x(v)=x(z)$.  The required functions $Q$ are
determined before via a similar analysis of
$D_I D_{g-1} \log \hat{Q}^{(0)}_1(x(v),x(z)) -D_I D_{g-1} \log
\hat{M}^{(0)}_1(x(v);z)$. After all, there is no principal obstacle to
produce the solution $D_I D_g \log P^{(0)}_1(x(v),x(z))$ from which by
expansion about $x(v)=\infty$ we get global linear and quadratic loop
equations for $W^{(g)}_{|I|+1}$. This shows that the quartic
Kontsevich model satisfies blobbed TR.  Working out the details will
be challenging, however, as the formulae become increasingly lengthy
with larger $g$ and Taylor expansions to order $4g+1$ are necessary.
Furthermore, in order to be a meaningful extension of TR, the
multidifferentials $ \omega^{(g)}_n $ generated by eq.\
\eqref{eq:mainformula} in Theorem \ref{thm:mainrecursion} must be
symmetric in their arguments. This symmetry is highly non-trivial and
significantly harder to prove than in ordinary TR. In \cite{hock2025},
the symmetry of all $\omega^{(0)}_n$ was proved.

\bigskip

Instead of following the increasingly complicated computations
described above, two alternative approaches can be pursued:

Firstly, the involution identity (not used in this article but proven
to follow from the DSE for \( \Omega^{(g)}_n \) at least for
\( g=0 \)) coincides with a symmetry identity in the theory of
\( x \)-\( y \) duality in TR
\cite{Hock:2022wer,Alexandrov:2022ydc,Alexandrov:2023oov}. These
\( x \)-\( y \) duality identities were recently used in
\cite{Alexandrov:2024tjo} to define the so-called \textit{Generalised
  Topological Recursion}, which, however, is only valid for the
classical Bergman kernel \( B(z,w) \) and not for our extended
\( \omega_{0,2}(z,w) = B(z,w) - B(z,-w) \). Due to the similarity
between the involution identity and the \( x \)-\( y \) duality
formula, the quartic Kontsevich model might be solvable using
Generalised Topological Recursion with a deformed \( \omega_{0,2} \),
implying additional poles at the anti-diagonal and at \( z=0 \), which
must be included in the set of \textbf{key} points (in addition to the
ramification points of \( x \)) in Generalised Topological Recursion.

Secondly, as mentioned earlier, an important assumption made in this
article and in previous works on the quartic Kontsevich model was the
rationality condition (i.e., assuming a genus zero spectral curve) for
\( x \) and \( y \), which implies the rationality of
\( \Omega^{(g)}_n \). Allowing for a higher-genus or even a
non-algebraic spectral curve might reduce the solution of the quartic
Kontsevich model to ordinary TR. This presumption is supported by the
solution of the \( O(N) \) model on random lattices obtained in
\cite{Borot:2009ia}, where the loop equation can also be interpreted
as an extended version of the loop equations of the Hermitian 1-matrix
model.


\begin{thebibliography}{ABDB{\etalchar{+}}24d}
\expandafter\ifx\csname url\endcsname\relax
  \def\url#1{\texttt{#1}}\fi
\expandafter\ifx\csname doi\endcsname\relax
  \def\doi#1{\burlalt{doi:#1}{http://dx.doi.org/#1}}\fi
\expandafter\ifx\csname urlprefix\endcsname\relax\def\urlprefix{URL }\fi
\expandafter\ifx\csname href\endcsname\relax
  \def\href#1#2{#2}\fi
\expandafter\ifx\csname burlalt\endcsname\relax
  \def\burlalt#1#2{\href{#2}{#1}}\fi

\bibitem[ABDB{\etalchar{+}}22]{Alexandrov:2022ydc}
A.~Alexandrov, B.~Bychkov, P.~Dunin-Barkowski, M.~Kazarian, and S.~Shadrin.
\newblock {A universal formula for the $x-y$ swap in topological recursion}.
\newblock 2022, \burlalt{2212.00320}{http://arxiv.org/abs/2212.00320}.

\bibitem[ABDB{\etalchar{+}}24a]{Alexandrov:2024hgu}
A.~Alexandrov, B.~Bychkov, P.~Dunin-Barkowski, M.~Kazarian, and S.~Shadrin.
\newblock {Any topological recursion on a rational spectral curve is KP
  integrable}.
\newblock 2024, \burlalt{2406.07391}{http://arxiv.org/abs/2406.07391}.

\bibitem[ABDB{\etalchar{+}}24b]{Alexandrov:2024tjo}
A.~Alexandrov, B.~Bychkov, P.~Dunin-Barkowski, M.~Kazarian, and S.~Shadrin.
\newblock {Degenerate and irregular topological recursion}.
\newblock 2024, \burlalt{2408.02608}{http://arxiv.org/abs/2408.02608}.

\bibitem[ABDB{\etalchar{+}}24c]{Alexandrov:2024qfe}
A.~Alexandrov, B.~Bychkov, P.~Dunin-Barkowski, M.~Kazarian, and S.~Shadrin.
\newblock {KP integrability of non-perturbative differentials}.
\newblock 2024, \burlalt{2412.18592}{http://arxiv.org/abs/2412.18592}.

\bibitem[ABDB{\etalchar{+}}24d]{Alexandrov:2023oov}
A.~Alexandrov, B.~Bychkov, P.~Dunin-Barkowski, M.~Kazarian, and S.~Shadrin.
\newblock {Topological recursion, symplectic duality, and generalized fully
  simple maps}.
\newblock {\em J. Geom. Phys.}, 206:105329, 2024,
  \burlalt{2304.11687}{http://arxiv.org/abs/2304.11687}.
\newblock \doi{10.1016/j.geomphys.2024.105329}.

\bibitem[BCEGF21]{Belliard:2021jtj}
R.~Belliard, S.~Charbonnier, B.~Eynard, and E.~Garcia-Failde.
\newblock {Topological recursion for generalised Kontsevich graphs and $r$-spin
  intersection numbers}.
\newblock 2021, \burlalt{2105.08035}{http://arxiv.org/abs/2105.08035}.

\bibitem[BE11]{Borot:2009ia}
G.~Borot and B.~Eynard.
\newblock {Enumeration of maps with self avoiding loops and the O(n) model on
  random lattices of all topologies}.
\newblock {\em J. Stat. Mech.}, 1101:P01010, 2011,
  \burlalt{0910.5896}{http://arxiv.org/abs/0910.5896}.
\newblock \doi{10.1088/1742-5468/2011/01/P01010}.

\bibitem[BEO15]{Borot:2013lpa}
G.~Borot, B.~Eynard, and N.~Orantin.
\newblock {Abstract loop equations, topological recursion and new
  applications}.
\newblock {\em Commun. Num. Theor. Phys.}, 09:51--187, 2015,
  \burlalt{1303.5808}{http://arxiv.org/abs/1303.5808}.
\newblock \doi{10.4310/CNTP.2015.v9.n1.a2}.

\bibitem[BGHW22]{Branahl:2021slr}
J.~Branahl, H.~Grosse, A.~Hock, and R.~Wulkenhaar.
\newblock {From scalar fields on quantum spaces to blobbed topological
  recursion}.
\newblock {\em J. Phys. A}, 55(42):423001, 2022,
  \burlalt{2110.11789}{http://arxiv.org/abs/2110.11789}.
\newblock \doi{10.1088/1751-8121/ac9260}.

\bibitem[BH23]{Branahl:2022uge}
J.~Branahl and A.~Hock.
\newblock {Complete solution of the LSZ model via topological recursion}.
\newblock {\em Commun. Math. Phys.}, 401(3):2845--2899, 2023,
  \burlalt{2205.12166}{http://arxiv.org/abs/2205.12166}.
\newblock \doi{10.1007/s00220-023-04702-z}.

\bibitem[BHW21]{Branahl:2020uxs}
J.~Branahl, A.~Hock, and R.~Wulkenhaar.
\newblock {Perturbative and geometric analysis of the quartic Kontsevich
  model}.
\newblock {\em SIGMA}, 17:085, 2021,
  \burlalt{2012.02622}{http://arxiv.org/abs/2012.02622}.
\newblock \doi{10.3842/SIGMA.2021.085}.

\bibitem[BHW22]{Branahl:2020yru}
J.~Branahl, A.~Hock, and R.~Wulkenhaar.
\newblock {Blobbed topological recursion of the quartic Kontsevich model I:
  Loop equations and conjectures}.
\newblock {\em Commun. Math. Phys.}, 393(3):1529--1582, 2022,
  \burlalt{2008.12201}{http://arxiv.org/abs/2008.12201}.
\newblock \doi{10.1007/s00220-022-04392-z}.

\bibitem[Bor14]{Borot:2013fla}
G.~Borot.
\newblock {Formal multidimensional integrals, stuffed maps, and topological
  recursion}.
\newblock {\em Ann. Inst. H. Poincare D Comb. Phys. Interact.}, 1(2):225--264,
  2014, \burlalt{1307.4957}{http://arxiv.org/abs/1307.4957}.
\newblock \doi{10.4171/aihpd/7}.

\bibitem[BS17]{Borot:2015hna}
G.~Borot and S.~Shadrin.
\newblock {Blobbed topological recursion: properties and applications}.
\newblock {\em Math. Proc. Cambridge Phil. Soc.}, 162(1):39--87, 2017,
  \burlalt{1502.00981}{http://arxiv.org/abs/1502.00981}.
\newblock \doi{10.1017/S0305004116000323}.

\bibitem[BW24]{Borot:2023thu}
G.~Borot and R.~Wulkenhaar.
\newblock {A Note on BKP for the Kontsevich matrix model with arbitrary
  potential}.
\newblock {\em SIGMA}, 20:050, 2024,
  \burlalt{2306.01501}{http://arxiv.org/abs/2306.01501}.
\newblock \doi{10.3842/SIGMA.2024.050}.

\bibitem[CEO06]{Chekhov:2006vd}
L.~Chekhov, B.~Eynard, and N.~Orantin.
\newblock {Free energy topological expansion for the 2-matrix model}.
\newblock {\em JHEP}, 12:053, 2006,
  \burlalt{math-ph/0603003}{http://arxiv.org/abs/math-ph/0603003}.
\newblock \doi{10.1088/1126-6708/2006/12/053}.

\bibitem[EO07]{Eynard:2007kz}
B.~Eynard and N.~Orantin.
\newblock {Invariants of algebraic curves and topological expansion}.
\newblock {\em Commun. Num. Theor. Phys.}, 1:347--452, 2007,
  \burlalt{math-ph/0702045}{http://arxiv.org/abs/math-ph/0702045}.
\newblock \doi{10.4310/CNTP.2007.v1.n2.a4}.

\bibitem[Eyn03]{Eynard:2002kg}
B.~Eynard.
\newblock {Large N expansion of the 2 matrix model}.
\newblock {\em JHEP}, 01:051, 2003,
  \burlalt{hep-th/0210047}{http://arxiv.org/abs/hep-th/0210047}.
\newblock \doi{10.1088/1126-6708/2003/01/051}.

\bibitem[Eyn05]{BertrandEynard2004}
B.~Eynard.
\newblock Topological expansion for the 1-hermitian matrix model correlation
  functions.
\newblock {\em Journal of High Energy Physics}, 2004(11):031, 2005.
\newblock \doi{10.1088/1126-6708/2004/11/031}.

\bibitem[Eyn16]{Eynard:2016yaa}
B.~Eynard.
\newblock {\em {Counting Surfaces}}, volume~70 of {\em Prog. Math. Phys.}
\newblock Birkh{\"a}user/ Springer, 2016.
\newblock \doi{10.1007/978-3-7643-8797-6}.

\bibitem[GHW19]{Grosse:2019jnv}
H.~Grosse, A.~Hock, and R.~Wulkenhaar.
\newblock {Solution of all quartic matrix models}.
\newblock 2019, \burlalt{1906.04600}{http://arxiv.org/abs/1906.04600}.

\bibitem[GHW23]{Grosse:2019nes}
H.~Grosse, A.~Hock, and R.~Wulkenhaar.
\newblock {A Laplacian to compute intersection numbers on
  $\overline{\mathcal{M}}_{g,n}$ and correlation functions in NCQFT}.
\newblock {\em Commun. Math. Phys.}, 399:481–517, 2023,
  \burlalt{1903.12526}{http://arxiv.org/abs/1903.12526}.
\newblock \doi{10.1007/s00220-022-04557-w}.

\bibitem[Gou10]{Gould}
H.~W. Gould.
\newblock {Tables of Combinatorial Identities}.
\newblock edited by Jocelyn Quaintance, 2010.
\newblock
  \urlprefix\url{https://web.archive.org/web/20190629193344/http://www.math.wvu.edu/~
  gould/}.

\bibitem[GSW17]{Grosse:2016pob}
H.~Grosse, A.~Sako, and R.~Wulkenhaar.
\newblock {Exact solution of matricial $\Phi^3_2$ quantum field theory}.
\newblock {\em Nucl. Phys. B}, 925:319--347, 2017,
  \burlalt{1610.00526}{http://arxiv.org/abs/1610.00526}.
\newblock \doi{10.1016/j.nuclphysb.2017.10.010}.

\bibitem[GSW18]{Grosse:2016qmk}
H.~Grosse, A.~Sako, and R.~Wulkenhaar.
\newblock {The $\Phi^3_4$ and $\Phi^3_6$ matricial QFT models have reflection
  positive two-point function}.
\newblock {\em Nucl. Phys. B}, 926:20--48, 2018,
  \burlalt{1612.07584}{http://arxiv.org/abs/1612.07584}.
\newblock \doi{10.1016/j.nuclphysb.2017.10.022}.

\bibitem[GW09]{Grosse:2009pa}
H.~Grosse and R.~Wulkenhaar.
\newblock {Progress in solving a noncommutative quantum field theory in four
  dimensions}.
\newblock 2009, \burlalt{0909.1389}{http://arxiv.org/abs/0909.1389}.

\bibitem[GW14]{Grosse:2012uv}
H.~Grosse and R.~Wulkenhaar.
\newblock {Self-dual noncommutative $\phi^4$-theory in four dimensions is a
  non-perturbatively solvable and non-trivial quantum field theory}.
\newblock {\em Commun. Math. Phys.}, 329:1069--1130, 2014,
  \burlalt{1205.0465}{http://arxiv.org/abs/1205.0465}.
\newblock \doi{10.1007/s00220-014-1906-3}.

\bibitem[Hoc20]{Hock:2020rje}
A.~Hock.
\newblock {\em {Matrix Field Theory}}.
\newblock PhD thesis, {WWU M{\"u}nster}, 2020,
  \burlalt{2005.07525}{http://arxiv.org/abs/2005.07525}.

\bibitem[Hoc23]{Hock:2022pbw}
A.~Hock.
\newblock {A simple formula for the x-y symplectic transformation in
  topological recursion}.
\newblock {\em J. Geom. Phys.}, 194:105027, 2023,
  \burlalt{2211.08917}{http://arxiv.org/abs/2211.08917}.
\newblock \doi{10.1016/j.geomphys.2023.105027}.

\bibitem[Hoc24]{Hock:2022wer}
A.~Hock.
\newblock {On the $x$-$y$ symmetry of correlators in topological recursion via
  loop insertion operator}.
\newblock {\em Commun. Math. Phys.}, 405:166, 2024,
  \burlalt{2201.05357}{http://arxiv.org/abs/2201.05357}.
\newblock \doi{10.1007/s00220-024-05043-1}.

\bibitem[HSW25]{hock2025}
A.~Hock, S.~Shadrin, and R.~Wulkenhaar.
\newblock Symmetry of meromorphic differentials produced by involution
  identity, and relation to integer partitions, 2025,
  \burlalt{2501.00082}{http://arxiv.org/abs/2501.00082}.

\bibitem[HW24]{Hock:2021tbl}
A.~Hock and R.~Wulkenhaar.
\newblock {Blobbed topological recursion of the quartic Kontsevich model II:
  Genus=0}.
\newblock {\em Ann. Inst. H. Poincare D Comb. Phys. Interact.}, online first,
  2024, \burlalt{2103.13271}{http://arxiv.org/abs/2103.13271}.
\newblock \doi{10.4171/AIHPD/198}.
\newblock with an appendix by M.\ Do\l{}\k{e}ga.

\bibitem[Kon92]{Kontsevich:1992ti}
M.~Kontsevich.
\newblock {Intersection theory on the moduli space of curves and the matrix
  Airy function}.
\newblock {\em Commun. Math. Phys.}, 147:1--23, 1992.
\newblock \doi{10.1007/BF02099526}.

\bibitem[MS91]{Makeenko:1991ec}
{\relax Yu}.~Makeenko and G.~W. Semenoff.
\newblock {Properties of Hermitean matrix models in an external field}.
\newblock {\em Mod. Phys. Lett.}, A6:3455--3466, 1991.
\newblock \doi{10.1142/S0217732391003985}.

\bibitem[SW23]{Schurmann:2019mzu}
J.~Sch\"urmann and R.~Wulkenhaar.
\newblock {An algebraic approach to a quartic analogue of the Kontsevich
  model}.
\newblock {\em Math. Proc. Camb. Phil. Soc}, 174:471--495, 2023,
  \burlalt{1912.03979}{http://arxiv.org/abs/1912.03979}.
\newblock \doi{10.1017/S0305004122000366}.

\bibitem[Wit91]{Witten:1990hr}
E.~Witten.
\newblock {Two-dimensional gravity and intersection theory on moduli space}.
\newblock In {\em {Surveys in differential geometry ({C}ambridge, {MA},
  1990)}}, pages 243--310. Lehigh Univ., Bethlehem, PA, 1991.
\newblock \doi{10.4310/SDG.1990.v1.n1.a5}.

\end{thebibliography}

\newcommand{\etalchar}[1]{$^{#1}$}

\end{document}